\documentclass[11pt]{article}

\usepackage[utf8]{inputenc} 
\usepackage[T1]{fontenc}
\usepackage{lmodern} 
\usepackage{microtype}

\usepackage[margin=1in]{geometry} 

\usepackage{setspace}
\usepackage{mathtools, amsthm, amssymb, mathrsfs, bbm}
\usepackage{thmtools}

\usepackage[
  colorlinks=true,
  linkcolor=blue,
  citecolor=blue,
  urlcolor=blue
]{hyperref}       

\usepackage{url}            
\usepackage{booktabs}       
\usepackage{amsfonts}       
\usepackage{nicefrac}       
\usepackage{xcolor}         

\usepackage{graphicx, wrapfig, subfigure} 
\usepackage{enumerate}
\usepackage{comment}
\usepackage[capitalize]{cleveref}
\usepackage{todonotes}
\usepackage{algorithm}
\usepackage[noend]{algpseudocode}
\usepackage[style=alphabetic]{biblatex}
\usepackage{dsfont}

\usepackage{enumitem}
\usepackage{xparse}

\usepackage{authblk}
\newtheorem{definition}{Definition}

\newtheorem{lemma}{Lemma}

\newtheorem{theorem}{Theorem}
\newtheorem{observation}{Observation}

\newtheorem{remark}{Remark}

\newcommand{\greedyalgorithm}{\text{\textsc{Greedy}}}
\newcommand{\optalgorithm}{\text{\textsc{Opt}}}
\newcommand{\acceptrejectalgorithm}{\text{\textsc{AcceptReject}}}

\newcommand{\greedy}[1]{\mathcal{G}_{#1}}
\newcommand{\greedydet}[1]{G_{#1}}
\newcommand{\acceptreject}[2]{\mathcal{A}^{#1}_{#2}}
\newcommand{\acceptrejectdet}[2]{A^{#1}_{#2}}
\newcommand{\opt}[1]{\mathcal{O}_{#1}}
\newcommand{\optdet}[1]{O_{#1}}
\newcommand{\smax}[1]{H_{#1}}

\newcommand{\neighborsin}[2]{N_{#1}(#2)}
\newcommand{\neighbors}[1]{N(#1)}
\newcommand{\f}[1]{|\neighbors{#1}|}

\DeclareMathOperator{\E}{\mathbb{E}}

\DeclarePairedDelimiter\floor{\lfloor}{\rfloor}

\newcommand{\abs}[1]{\big| #1 \big|}

\newcommand{\todoj}[2][]{\todo[color=blue!20!white,#1]{JC: #2}}

\newcommand{\cM}{\mathcal{M}}
\newcommand{\cA}{\mathcal{A}}
\newcommand{\cB}{\mathcal{B}}
\newcommand{\cX}{\mathcal{X}}
\newcommand{\cY}{\mathcal{Y}}

\newcommand{\cF}{\mathcal{F}}
\newcommand{\cD}{\mathcal{D}}
\newcommand{\cE}{\mathcal{E}}

\newcommand{\cI}{\mathcal{I}}
\newcommand{\cU}{\mathcal{U}}

\newcommand{\cG}{\mathcal{G}}

\newcommand{\RR}{\mathbb{R}}
\newcommand{\bigabs}[1]{\bigl| #1 \bigr|}
\newcommand{\bigpar}[1]{\left( #1 \right)}

\newcommand{\ind}{\perp\!\!\!\!\perp}

\allowdisplaybreaks

\usepackage{color-edits}
\addauthor{eb}{blue}
\addauthor{fs}{magenta}
\addauthor{jc}{red}

\title{
    On the Average-Case Performance of Greedy \\ for Maximum Coverage\thanks{Eric Balkanski and Jason Chatzitheodorou were supported in part by NSF grants CCF-2210501 and IIS-2147361.}
}

\author[a]{Eric Balkanski\thanks{eb3224@columbia.edu}}
\author[a]{Jason Chatzitheodorou\thanks{ic2621@columbia.edu}}
\author[b]{Flore Sentenac\thanks{sentenac@hec.fr}}
\affil[a]{Columbia University, IEOR}
\affil[b]{HEC Paris, ISOM}
\date{}

\begin{document}

\maketitle

\begin{abstract}
    For the classical maximum coverage problem, the greedy algorithm achieves a worst-case $1-1/e$ approximation, which is optimal unless $\text{P} = \text{NP}$. The notion of coverage appears in a wide range of optimization tasks, where empirical evaluations indicate approximation ratios close to $1$ for the greedy algorithm on real data. Random models have provided average-case justifications for the empirical performance of many well-known algorithms, but little is known about the average-case performance of  greedy for maximum coverage.
    
    We analyze the expected approximation ratio of the greedy algorithm in a random model, which we call the left-regular random model. We first show that, for all parameter settings of this model, the expected approximation ratio of the greedy algorithm improves by a constant over its worst-case $1-1/e$ guarantee. We then identify two simple conditions, either of which ensures that the expected approximation ratio is close to $1$ for sufficiently large graphs. Finally, we show that there is a regime where greedy does not achieve an expected approximation better than $0.94$. To obtain these results, we develop analytical tools, including a novel application of the differential equation method and a connection to maximum matching in Erd\H{o}s-Rényi graphs, which may be of independent interest for other random models.
\end{abstract}

\thispagestyle{empty} 
\clearpage            

\pagenumbering{arabic}


\section{Introduction}
            
Maximizing coverage is a fundamental problem in computer science and operations research, as the notion of coverage arises across a wide range of optimization tasks. Influence maximization in social networks~\cite{influence-social-networks}, facility and sensor placement~\cite{facility-sensor-allocation}, information retrieval~\cite{information-retrieval}, content recommendation~\cite{blog-watch}, and finding essential webpages~\cite{essential-web-pages}  are examples of problems that have been formulated as maximum coverage problems. Maximum coverage is also closely related to set cover, one of Karp's original 21 NP-complete problems~\cite{karp1972reducibility}.

Maximum coverage can be defined as follows. Given a bipartite graph with left nodes $L$ and right nodes $R$, the goal is to find the $k$ nodes in $L$ that cover the largest number of nodes in $R$, i.e., $\max_{S \subseteq L: |S| \leq k} |N(S)|$. 
The greedy algorithm, which iteratively selects the node in $L$ that covers the largest number of nodes in $R$ not yet covered, achieves a $1-1/e$ worst-case approximation ratio guarantee~\cite{greedy-worst-case}, which is optimal unless $\text{P} = \text{NP}$~\cite{maxcoverage-inapproximability}. However,  worst-case guarantees are often driven by pathological instances, leading to pessimistic bounds that can be uninformative in practice. There have been numerous approaches to bridging this gap between theory and practice for various problems through beyond–worst-case analyses of algorithms. A canonical example is the smoothed analysis of the simplex method~\cite{spielman2004smoothed}.

For maximum coverage, the greedy algorithm empirically performs significantly better than its $1-1/e$ 
approximation ratio guarantee~\cite{essential-web-pages, submodular-instance-specific}. There have been numerous attempts going beyond worst-case to explain this strong performance in practice for the broader problem of monotone submodular maximization under a cardinality constraint. These include the parameterized properties of curvature~\cite{conforti1984submodular}, stability~\cite{chatziafratis2017stability}, and sharpness~\cite{pokutta2020unreasonable}. More recently, \cite{rubinstein2022budget} considered a budget-smoothed model where the cardinality constraint $k$ is stochastic.

However, except for a few notable exceptions discussed in \cref{sec:relatedwork}, little is known about the average-case performance of greedy for maximum coverage on random graphs. In contrast, random models have provided average-case justifications for the performance of classical algorithms such as quicksort~\cite{hoare1962quicksort}, first-fit decreasing for bin packing~\cite{frederickson1980probabilistic},  local search for the traveling salesman problem~\cite{englert2014worst}, and the greedy algorithm for both max-value~\cite{mastin2013greedy, arnosti2022greedy} and min-cost matching~\cite{frieze1990greedy}. Motivated by the strong empirical performance of greedy for maximum coverage, we explore the following question.

\begin{center}
     \emph{ What is the average-case performance of the greedy algorithm for maximum coverage?    }
\end{center}

We study the approximation ratio of greedy under a model that we call the left-regular random model $\textsc{LRR}(n, d)$. This model generates a bipartite graph with $n$ left nodes $L$ and $n$ right nodes $R$, where each node $u \in L$ independently selects $d$ neighbors  uniformly at random from $R$. We also consider extensions in which $|L| \neq |R|$ and the nodes in $L$ have non-uniform degrees. Our motivation for the $d$-regularity assumption for left nodes $L$ in the base model is that it is a challenging regime for greedy. In particular, there are simple bad instances for greedy in which the left nodes have equal degree  (see~\cref{sec:badinstance} for an example). Moreover, in empirical evaluations of extensions of this model, greedy's approximation ratio was lowest when the nodes in $L$ had equal degree (see~\cref{sec:unequal degrees} for an example).

\paragraph*{Our results.}  
Our first main result is an unconditional bound on the expected approximation ratio of greedy in the left-regular random model. We show that for all values of $n, d,$ and $k$, there is a constant improvement for the approximation ratio of greedy in this random model compared to its worst-case approximation ratio of $1-1/e$. For all results, we assume that greedy breaks ties according to an arbitrary but consistent ordering of the nodes in $L$.

\begin{restatable}{theorem}{TheoremOne}
\label{thm:Theorem 1}
 There exists a constant $c > 0 $ such that for all $n \in \mathbb{N} , d \in [n], k \in [n]$, the greedy algorithm achieves, in expectation, a $(1-1/e + c)$-approximation under the left-regular random model.
\end{restatable}

The second main result provides two alternative sufficient conditions, either of which ensures a near-optimal approximation ratio for greedy for sufficiently large graphs.

\begin{restatable}{theorem}{TheoremTwo}
\label{thm:Theorem 2}
For any  $\epsilon \in (0,1)$, $n= \Omega(\epsilon^{-8})$,   the greedy algorithm achieves, in expectation, a $1- \epsilon$ approximation under the left-regular random  model if either $d = \Omega(\epsilon^{-8})$ or $k \in \left[0, \frac{\epsilon n}{2d}\right] \cup \left[\frac{2n}{\epsilon d},\, n \right]$.
\end{restatable}

This result implies  that, for any constant $\epsilon$,  the degree $d$ of nodes in $L$ being super-constant or the cardinality constraint $k$ lying outside of a constant factor neighborhood of $n/d$ are sufficient conditions for a $1-\epsilon$ approximation. We note that when $k = n/d$, then the sum of the degrees of  $k$ left nodes  is equal to the number of right nodes. The only regime for which \cref{thm:Theorem 2} does not give near-optimality is where $d$  is constant and $k = \Theta(n/d)$. The third main result shows that for $d = 2$ and some  $k$ such that $k =  \Theta(n/d)$, the approximation ratio of greedy is not near-optimal as $n$ grows large: it is upper bounded by $0.94$.

\begin{restatable}{theorem}{TheoremThree}
\label{thm:Theorem 3}
    For $d=2$, there exists a sequence $k(n)$ of cardinality constraints such that the expected approximation ratio of the greedy algorithm is  at most $0.94$ in the limit as $n \rightarrow \infty$ under the left-regular random model.
\end{restatable}

These results indicate that the left-regular random model is sufficient to always get a constant improvement in the approximation ratio of greedy, but that it requires additional assumptions for near-optimality; \cref{thm:Theorem 2} gives two such assumptions.

\vspace{-.05cm}
\paragraph{Extensions.} We  consider two extensions of the left-regular random model.

\begin{itemize}[leftmargin=*]
    \item \emph{Unbalanced bipartite graphs.}
    We consider an extension $\textsc{ULRR}(n, m, d)$ of the base model $\textsc{LRR}(n, d)$ where  the number of right nodes is denoted by $m$ and is potentially different from the number of left nodes $n$. We show that for all $n, m$ such that $n = cm$  for some constant $c > 0$ and for all $d \in [m]$ and $k \in [n]$, there exists a constant $\epsilon_c > 0$ such that greedy achieves a $1 -1/e +\epsilon_c$ approximation in expectation. 
    \item \emph{Non-regular unbalanced bipartite graphs.} 
    We consider an extension $\textsc{GenR}(n, m, \{d_i\}_{i \in [n]})$ of the base model  where the left nodes have degree $d_1 \geq \cdots \geq d_n$ and the bipartite graph is unbalanced. We show that if either (1) the sum of the $k$ largest degrees is bounded away from $m$: $\sum_{i=1}^{k}d_i \not \in [(\epsilon/2) m, (\epsilon/2)^{-1}m]$, (2) the average degree of the $k$ largest degrees is sufficiently large: $\frac{1}{k} \sum_{i=1}^{k}d_i = \Omega(\max(1, (n/m)^2)\epsilon^{-8})$, or (3) $d_1, \ldots, d_n$ are drawn from a power law distribution and $k$ is sufficiently small, then the expected approximation ratio of greedy is at least $1 - \epsilon$. 

\end{itemize}

\vspace{-.2cm}

\paragraph{Our techniques.} The positive results  rely on two main  analytical tools, which  may also be of interest for analyzing greedy in  other random bipartite graph models.

First, we develop a framework using the differential equation method to analyze an \textsc{AcceptReject} algorithm that is equivalent to  greedy. \textsc{AcceptReject} proceeds in phases; in each phase, it iterates through every node to either add it to the solution or reject it. We show that the evolution of the value of the current solution of \textsc{AcceptReject} during the first phase can be approximated by a function that is the unique solution to a differential equation. This differential equation is obtained by approximating the expected per-iteration increase in the solution size as a function of the current solution size.

Second, we show that greedy's expected value is at least the expected value of a fixed set of $k$ left nodes, chosen independently of the random graph's realization.  
Although this claim may seem intuitive, it requires a careful argument because there are iterations in which the marginal gain of the fixed set is larger than that of greedy, where these marginal gains are measured with respect to their respective current solutions.  We circumvent that obstacle by analyzing a hybrid algorithm that runs greedy for $t$ iterations and then completes the solution by selecting the $k-t$ nodes with the largest residual degrees at iteration $t$. 

The negative result is established by relating maximum coverage in $2$-left regular graphs and maximum matching in Erd\H{o}s-R\'{e}nyi random graphs. This connection yields a lower bound on the coverage of the optimal solution via known results on maximum matching. At the same time, we tightly control the performance of greedy through our differential equation analysis.

\vspace{-.12cm}
\paragraph{Open questions.}  An alternative  random  model in which it would be natural to study the approximation ratio of greedy  is the Erd\H{o}s–R\'{e}nyi
random bipartite graph model $G_{n,n,p}$ with $n$ left nodes, $n$ right nodes, and  an edge between each left and right node independently with probability $p$. To the best of our knowledge, there is no known result for the maximum coverage problem under Erd\H{o}s–Rényi
random bipartite graphs and analyzing greedy in that model seems non-trivial. 

Obtaining stronger unconditional bounds on the approximation ratio of greedy in the left-regular random model is a natural direction for future work. In particular, using the differential equation method to analyze additional phases of \textsc{AcceptReject} is a promising direction for establishing stronger approximation guarantees for greedy in the regime where it is not near-optimal.

\vspace{-.12cm}
\paragraph{Paper organization.} We discuss additional related work in \cref{sec:relatedwork} and introduce preliminary definitions and notation in \cref{sec:prelims}. In \cref{sec:guarantees}, we show the  approximation guarantees achieved by greedy in the left-regular random model. In \cref{sec:limitation}, we show the negative result for greedy. In \cref{sec: theorem 1 generalization n = c m} and \cref{sec:nonregular}, we study the unbalanced and non-regular extensions of the model.

\vspace{-.12cm}
\subsection{Related Work}
\label{sec:relatedwork}
\vspace{-.05cm}

\cite{essential-web-pages} use the dual of the maximum coverage LP to provide bounds on the approximation ratio of greedy that are in practice much stronger than its worst-case $1-1/e$ approximation. In particular, for $k \geq 50$, they obtain bounds on real data that are always at least $0.95$. They also show that in the random model, where left nodes and right nodes both have degrees that follow a power law distribution and edges are added randomly according to these degree distributions, greedy achieves a $1-\epsilon$ approximation if it has a marginal contribution of at least $O(\epsilon^{-2}\log n)$ at every iteration. In \cite{typical-approximations-greedy}, three algorithms for maximum coverage, including greedy, are analyzed heuristically using the differential equation method on sparse biregular random bipartite graphs. In particular, they use a differential equation to predict a threshold for the cardinality constraint $k$ below which greedy achieves exact optimality with high probability and above which it is no longer optimal with high probability.

For the closely related set cover problem where the goal is to cover all the nodes in $R$ with a subset of $L$ of minimum size, \cite{hitting-set} and \cite{blot1995average} also study a random model with fixed degrees and also use the differential equation method. However, their primary focus is on  methods that, given parameter values for the random model, yield numerical estimates of the approximation ratio of greedy on random instances for these specific parameters, but not closed-form approximation ratios. In addition to these numerical estimates,  \cite{hitting-set} shows that for a fixed ratio $|R|/|L|$, greedy is asymptotically optimal as the degree of right nodes tends to infinity.   When each edge is included with constant probability, \cite{TELELIS2005171} shows that there is an algorithm, distinct from greedy, that is asymptotically optimal for set cover.

\section{Preliminaries}
\label{sec:prelims}

In the maximum coverage problem, there is a bipartite graph  $B = (L, R, E)$ with  $n$ left nodes $L = \{u_1, \ldots, u_n\}$, $m$ right nodes $R = \{v_1, \ldots, v_m\}$, and edges $E \subseteq L \times R$. The goal is to find the $k$ nodes in $L$ that cover the largest number of nodes in $R$, i.e., $\max_{S \subseteq L : |S| \leq k} |N_B(S)|$ where $N_B(S)$ are the neighbors of $S$ in $B$. We abuse notation and write $N(S)$ when $B$ is clear from context. The greedy algorithm (\greedyalgorithm) iteratively adds to the current solution $S$ the left node that covers the most right nodes not yet covered.  We refer to $|N(u) \setminus N(S)|$  and $|N(S)|$ as the marginal contribution of $u$ to $S$ and the value of $S$, respectively.

\begin{algorithm}[H]
\begin{algorithmic}
    \caption{$\greedyalgorithm(B,  k)$} 
    \State \textbf{Input:} bipartite graph $B = (L, R, E)$, cardinality constraint $k$
    \State $S \gets \emptyset$
        \For{$i = 1, 2, \dots, k$}
            \State Add $\arg\max_{u \in L \setminus S} |N(u) \setminus N(S)|$ to $S$ \vspace{.1cm}
        \EndFor
    \Return $S$
\end{algorithmic}
\end{algorithm}

\paragraph{The random graph model.} We consider instances of the maximum coverage problem with random bipartite graphs drawn from the following random graph model that we call the left-regular random model. Throughout the paper, uppercase letters represent fixed sets/graphs, while calligraphic letters denote random sets/graphs.

\begin{definition}
The  $\textsc{LRR}(n, d)$ model generates a bipartite graph
  $\cB = (L, R, \cE)$ by setting $|L| = |R| = n$ and, independently  for each $u \in L$, selecting $d$ neighbors in $R$ uniformly at random without replacement.
\end{definition}

We note that  a left-regular random graph satisfies $\Pr\left[ v_j \in \neighbors{u_i} \right]      = \frac{d}{n}$ for all $v_j \in R$ and $u_i \in L$. When analyzing $\greedyalgorithm$ under the left-regular random model, we assume it breaks ties according to an arbitrary, but consistent, ordering of the nodes in $L$, and that the left nodes $u_i$ are indexed in the order in which greedy breaks ties.

We let $\mathcal{G}_k(\mathcal{B})$ and $\mathcal{O}_k(\mathcal{B})$  denote the greedy solution and the optimal solution, respectively, for the maximum coverage problem with cardinality constraint $k$ over a bipartite graph  $\mathcal{B} \sim \textsc{LRR}(n, d)$. Throughout the paper, approximation guarantees in expectation are defined as ratios of expected values. For  fixed $n, d,$ and $k$, the greedy algorithm achieves, in expectation, an $\alpha$-approximation if
$$\E\limits_{\cB \sim \textsc{LRR}(n, d)} [|N(\mathcal{G}_k(\cB))|] \geq \alpha \cdot \E\limits_{\cB \sim \textsc{LRR}(n, d)}[|N(\mathcal{O}_k(\cB))|]. $$

We abuse notation and write $\mathcal{G}_k$ and $\mathcal{O}_k$ when $\cB$ is clear from context. It will be useful in our analysis to compare $\E[|N(\mathcal{G}_k)|]$  to the expected value of a fixed set of size $k$. For that purpose, we denote the first $k$ left nodes as    $\smax{k} = \{u_1, u_2, \dots, u_k\}.$ We also use the notation $[n] = \{1, \ldots, n\}$.

\section{Approximation Guarantees}
\label{sec:guarantees}

In this section, we show that greedy achieves an approximation ratio that is, for all parameters $n, d$, and $k$, a constant improvement over its worst-case $1-1/e$ guarantee and, in a wide range of regimes,  asymptotically optimal. We first discuss two central technical tools that we use for these two results in \cref{sec: technical tools n=m}, then partition the space of all parameters in three regions and analyze greedy in each of them in \cref{sec: different regions n=m}, and finally combine the analyses of these regions to obtain the desired results  in \cref{sec: putting everything together n=m}. Missing proofs are deferred to \cref{sec: omitted proofs of section 3}.


\subsection{Technical Tools}
\label{sec: technical tools n=m}

\vspace{-.05cm}
\paragraph{The near-linear, saturated, and critical regions.} We distinguish three parameter regimes that we analyze separately. 
First, for small values of $k$, we have the near-linear region, where $\greedyalgorithm$ selects nodes that have neighborhoods with almost no overlap. Therefore, their value increases almost linearly with $k$. 
Second, for large values of $k$, we have the saturated region, where  $\greedyalgorithm$ covers almost all the right nodes.
Between these two regions lies the critical region, which is the most challenging to analyze.

\vspace{-.12cm}
\paragraph{The differential equation method.} 
The first technical tool is a novel application of the differential equation method. We use it to  estimate  the number of iterations 
\vspace{-.04cm}
$$t_d = \max\{t \geq 1: |\neighbors{\greedy{t}(\mathcal{B})} \setminus \neighbors{\greedy{t-1}(\mathcal{B}) }| = d \}$$
\vspace{-.05cm}
in which $\greedyalgorithm$, over a left-regular random graph $\mathcal{B} \sim \text{\textsc{LRR}}(n, d)$, selects a node  with marginal contribution to the current solution $\mathcal{G}_{t-1}$ that is equal to $d$, i.e.,  all $d$ neighbors of the selected node are disjoint from the neighbors of $\greedy{t-1}$. The main lemma for this estimate $t^*_{d}$ of $t_d$ is a bound on its error $|t_d - t^*_d|$ that, for any $d \leq \sqrt{n/2}$, has an $\tilde{O}(\sqrt{n})$ dependence on $n$  with high probability.

\begin{restatable}{lemma}{FirstPhaseAcceptsN}
\label{lem:first phase accepts for n=m}
    Let $n \in \mathbbm{N}$, $d \leq \sqrt{n/2}, k \in [n]$ and  $\cB \sim \text{\textsc{LRR}}(n, d)$. 
    With probability at least $1 - \frac{1}{n}$, estimate $t^*_d = \left( 1 - \left( 1 + d(d-1) \right)^{-\frac{1}{d-1}} \right) \cdot \frac{n}{d}$ is such that
    $|t_d - t^*_d | \leq 3e^{d^2} \cdot \sqrt{8n\log{2n}}$.
\end{restatable}

This estimate plays a key role in the proof showing that greedy achieves a constant improvement over $1-1/e$ in the critical region when $d$ is small. To prove \cref{lem:first phase accepts for n=m}, we employ the differential equation method, formalized by \cite{wormald-differential-equations}. 
Under appropriate regularity conditions, including bounded increments and Lipschitz continuity, this method yields a high-probability bound on the discrepancy between a discrete stochastic process and the deterministic solution of an associated differential equation. It has found great use in the field of random graphs~\cite{hitting-set,pittel1996sudden,wormald-differential-equations,wormald1999differential,mastin2013greedy,frieze1998MaximumMatchingRevisited}.

We do not apply the method directly to the standard description of $\greedyalgorithm$ because its evolution is difficult to express through a low-dimensional Markovian state, which is crucial for tractability. Instead, we analyze
 $\acceptrejectalgorithm$, an equivalent reformulation of $\greedyalgorithm$.
$\acceptrejectalgorithm$, which is formally defined in \cref{alg:threshold greedy}, proceeds in phases from $p=\max_{u \in L} |N(u)|$ to $1$. In phase $p$, it iterates over the left nodes $u_i$ of the bipartite graph in the order in which greedy breaks ties  and either accepts (meaning a node gets added to the solution) or rejects each of them based on whether their marginal contribution to the current solution $A$ is at least $p$. This iterative exploration of the graph in each phase implies that its randomness is revealed sequentially, allowing low-dimensional Markovian states to describe the evolution of the current solution size $|A|$.

\vspace{-.11cm}

\begin{algorithm}[H]
\begin{algorithmic}
    \caption{$\acceptrejectalgorithm(B,  k)$} \label{alg:threshold greedy}
    \State \textbf{Input:} bipartite graph $B = (L, R, E)$, cardinality constraint $k$
    \State $\acceptrejectdet{}{} \gets \emptyset$
    \For{$p = \max_{u \in L} |N(u)|$ to $0$}
        \For{$i = 1, 2, \dots, n$}
            \If{$|\neighbors{u_i} \setminus \neighbors{\acceptrejectdet{}{}}| \geq p$ and $u_i \notin A$ and $|\acceptrejectdet{}{}| < k$} 
                \State $\acceptrejectdet{}{} \gets \acceptrejectdet{}{} \cup \{u_i\}$ 
            \EndIf
        \EndFor
    \EndFor
    \Return $\acceptrejectdet{}{}$
\end{algorithmic}
\end{algorithm}

\vspace{-.17cm}

We denote by $\acceptreject{p}{i}$ the set $A$ of accepted nodes at iteration $i$ of phase $p$ of $\acceptrejectalgorithm$. For the proof of 
 \cref{lem:first phase accepts for n=m}, we focus on analyzing the first phase $p=\max_{u \in L} |N(u)| = d$ of $\acceptrejectalgorithm$ on a graph $\cB \sim \text{\textsc{LRR}}(n, d)$.

\begin{proof}[Proof sketch of \cref{lem:first phase accepts for n=m} (full proof in \cref{sec: technical tools n=m appendix})]

The solution returned by $\acceptrejectalgorithm$ and $\greedyalgorithm$ is not only the same, but the ordering in which nodes are added to this solution is also identical since the nodes $u_i$ are indexed in the order in which $\greedyalgorithm$ breaks ties. Thus, the number of iterations $t_d$ where $\greedyalgorithm$ adds nodes with marginal contribution $d$ is equal to the number of nodes accepted during the first phase $p=d$ of $\acceptrejectalgorithm$, i.e., $t_d = |\acceptreject{d}{n}|$. We define the function $F(y) = \left(1 - d y\right)^d$ for all $y \in \mathbbm{R}$. 

The central part of the proof consists of showing that $|\acceptreject{d}{i}| \approx n \cdot y(i/n)$ where $y(\cdot)$ is the unique solution to the following differential equation:
    \begin{align*}
        y'(t) = F(y(t)) = \left( 1 - d \cdot y(t) \right)^d,\ y(0) = 0. 
    \end{align*}

    This approximation is obtained by using the differential equation method, which requires the following four conditions.
    \begin{itemize}
    \item     The expected one-step change $\E\left[ |\acceptreject{d}{i+1}| - |\acceptreject{d}{i}| \middle| |\acceptreject{d}{i}| \right]$ is tightly approximated by $F\left( |\acceptreject{d}{i}| / n \right)$. For this condition, we show the following:
    \vspace{-.2cm}
    \begin{align*}
        \hspace{-0.9cm}
        \E\left[ |\acceptreject{d}{i+1}| \! - \! |\acceptreject{d}{i}| \middle| |\acceptreject{d}{i}| \right]
        = \Pr\left[|N(u_{i+1}) \! \setminus \! N(\acceptreject{d}{i})| = d \middle|  |\acceptreject{d}{i}| \right]  
        = \frac{\binom{n - |\acceptreject{d}{i}| d}{d}}{\binom{n}{d}}
        \approx \left( 1 \! - \! \frac{d|\acceptreject{d}{i}|}{n}  \right)^d 
        \! = F\left(\frac{|\acceptreject{d}{i}|}{ n} \right).
    \end{align*}
    \vspace{-.5cm}
    \item The function $F$ is Lipschitz.
    \item  A deterministic bound for the one-step change $|\acceptreject{d}{i+1}| - |\acceptreject{d}{i}|$. In our case, $|\acceptreject{d}{i+1}| - |\acceptreject{d}{i}| \leq 1$.  
    \item $|\acceptreject{d}{0}| \approx y(0)$. In our case,  $|\acceptreject{d}{0}| = y(0) = 0$.
    \end{itemize}
    We show that $\acceptreject{d}{i}$ and $F$ satisfy these four conditions with parameters such that the resulting approximation $t_d = |\acceptreject{d}{n}| \approx n \cdot y(1)$ is sufficiently tight to imply the desired bound.
\end{proof}
A promising direction for future work is to leverage estimates of the number of accepts across all phases, rather than only the first, as this could yield a tighter lower bound on the coverage of $\greedyalgorithm$. 
The main challenge is that the solution of the differential equations in each phase depends on the solutions from the preceding phases, making closed-form expressions analogous to those obtained for $t_d$ unlikely. This interdependence substantially complicates the derivation of the approximation ratio. 
Nevertheless, we believe that a sufficiently careful analysis could overcome these technical obstacles and uncover the true extent of the improvement that $\greedyalgorithm$ achieves over $1 - 1/e$.

\paragraph{The fixed-set lower bound.} The second technical tool provides a seemingly simple lower bound on the expected value of the greedy solution. The lower bound is the expected value of a fixed set of left nodes chosen independently of the realization of the random graph. For this fixed set of nodes, we consider the $k$ nodes with lowest indices $\smax{k} = \{u_1, \dots, u_k\}$. 

\begin{restatable}{lemma}{GreedyLowerBoundSmaxN}
\label{lem:greedy lower bound by smax for n=m}
    Let $n \in \mathbbm{N}$, $d\in [n], k \in [n]$, and $\cB \sim \text{\textsc{LRR}}(n, d)$, then $\E\left[ |\neighbors{\greedy{k}}| \right] \geq \E\left[ |\neighbors{\smax{k}}| \right]$.
\end{restatable}

In \cref{sec: different regions n=m}, we use this bound to show that greedy is near-optimal in the near-linear, the saturated, and the critical region when $d$ is large. 
The only region where it does not provide a constant improvement is in the critical region when $d$ is small, where we use \cref{lem:first phase accepts for n=m} instead.

It may seem natural that $\greedyalgorithm$  performs better than $\smax{k}$ over left-regular random graphs. However, a difficulty that arises in proving \cref{lem:greedy lower bound by smax for n=m}, is that the expected marginal improvement  $\E[|\neighbors{\greedy{t}} \setminus \neighbors{\greedy{t-1}}|]$ of $\greedyalgorithm$ is not always larger than the expected marginal improvement $\E[|\neighbors{\smax{t}} \setminus \neighbors{\smax{t-1}}|]$ of $\smax{k}$ at all iterations $t$, which we illustrate in \cref{fig:marginal of greedy vs smax} in \cref{sec: technical tools n=m appendix}.
Nevertheless, in the early iterations, where a small number of right nodes have been covered, it does hold that $\E[|\neighbors{\greedy{t}} \setminus \neighbors{\greedy{t-1}}|] \geq \E[|\neighbors{\smax{t}} \setminus \neighbors{\smax{t-1}}|]$. The challenge is thus to show that the expected value of $\greedyalgorithm$ remains larger  in later iterations. 

\begin{proof}[Proof sketch of~\cref{lem:greedy lower bound by smax for n=m} (full proof in \cref{sec: technical tools n=m appendix})]
    We define a hybrid algorithm $\mathcal{Y}^t$ that runs in two phases. First, it runs $t$ iterations of $\greedyalgorithm$, after which its current solution is $\greedy{t}$. In the second phase, it adds to the solution the $k-t$  nodes with  largest marginal contribution to $\greedy{t}$.
    
    The main part of the proof consists in showing that trading one iteration of $\greedyalgorithm$ in the first phase for one more node in the second phase reduces the expected value, i.e., $$\E\left[ |\neighbors{\mathcal{Y}^t}| \right] \geq \E\left[ |\neighbors{\mathcal{Y}^{t-1}}| \right].$$
    
    To show that inequality, we first note that $\greedy{t} \subseteq \mathcal{Y}^t \cap \mathcal{Y}^{t-1}$, which implies that it suffices to compare the marginal contribution of the $k-t$ other nodes in the two solutions to $\greedy{t}$. The crucial difference that we leverage is that the $k-t$ other nodes in $\mathcal{Y}^t$ are selected to be the nodes with largest marginal contribution to $\greedy{t}$, whereas for $\mathcal{Y}^{t-1}$ they are the nodes with largest marginal contribution to $\greedy{t-1}$. With $\E\left[ |\neighbors{\mathcal{Y}^t}| \right] \geq \E\left[ |\neighbors{\mathcal{Y}^{t-1}}| \right],$
    the lemma then follows  since $\E\left[ |\neighbors{\greedy{k}}| \right]
        = \E\left[ |\neighbors{\mathcal{Y}^k}| \right]$ and $ \E\left[ |\neighbors{\mathcal{Y}^{0}}| \right]
        = \E\left[ |\neighbors{\smax{k}}| \right].$
\end{proof}

\subsection{Analysis For The Different Regions}
\label{sec: different regions n=m}

\paragraph{Near-linear and saturated regions.}
In both regions, we show that $\greedyalgorithm$ achieves a ratio close to $1$ in \cref{lem:ratio lower bound for n=m and uniform d - outside critical region k large or small}. 
Intuitively, in the near-linear region (i.e., $k << n/d$), even the nodes of $\smax{k}$ share very few neighbors; therefore, the expected value of $\greedyalgorithm$ is close to $k d$. On the other hand, in the saturated region (i.e., $k >> n/d$), $\smax{k}$ covers most of the right nodes, even if the nodes of $\smax{k}$ share many neighbors. 
\begin{restatable}{lemma}{OutsideCriticalRegionND}
\label{lem:ratio lower bound for n=m and uniform d - outside critical region k large or small}
    Let $\epsilon \in (0, 1), n \in \mathbbm{N}, d \in [n], k \in \left[0, \frac{\epsilon n}{2 d}\right] \cup [\frac{2n}{\epsilon d}, n ]$ and $\cB \sim \text{\textsc{LRR}}(n, d)$.
    Then  we have $\E\left[ |\neighbors{\greedy{k}}| \right] \geq  \left( 1 - \epsilon \right) \cdot \E\left[ |\neighbors{\opt{k}}| \right]$.
\end{restatable}

The proof uses the fixed set lower bound from \cref{lem:greedy lower bound by smax for n=m}, whose value is lower bounded by the following lemma (proof deferred to \cref{sec: technical tools n=m appendix}).
\begin{restatable}{lemma}{SmaxExpectedCoverageN}
\label{lem:smax expected coverage for n=m}
    Let $n \in \mathbbm{N}$, $d \in [n], k \in [n]$, and $\cB \sim \text{\textsc{LRR}}(n, d)$.
    Then $\E\left[ |\neighbors{\smax{k}}| \right] 
    \geq \left( 1 - e^{-\frac{kd}{n}} \right) \cdot n$.
\end{restatable}

The proof also uses the following inequality, whose proof is deferred to \cref{sec: analysis for the different regions n=m appendix}.
\begin{restatable}{claim}{CLAIMONE}
\label{claim one}
    For all $\epsilon \in (0, 1)$, $\frac{1 - e^{-\epsilon}}{\epsilon} \geq 1 - \epsilon$.
\end{restatable}

We are now ready to prove~\cref{lem:ratio lower bound for n=m and uniform d - outside critical region k large or small}.
\begin{proof}[Proof of \cref{lem:ratio lower bound for n=m and uniform d - outside critical region k large or small}]
    For $k \geq \frac{2n}{\epsilon d}$:
    \begin{align*}
        \E\left[ |\neighbors{\greedy{k}}| \right] 
        \geq \E\left[ |\neighbors{\smax{k}}| \right]
        \geq \left( 1 - e^{-kd/n} \right)  n
        \geq \left( 1 - e^{-1/\epsilon} \right) \cdot n
        \geq \left( 1 - \epsilon \right) \cdot \E\left[ |\neighbors{\opt{k}}| \right],
    \end{align*}
    where the first inequality is because $\E\left[ |\neighbors{\greedy{k}}| \right]$ is lower bounded by $\E\left[ |\neighbors{\smax{k}}| \right]$ by \cref{lem:greedy lower bound by smax for n=m}, the second inequality is due to \cref{lem:smax expected coverage for n=m}, the third inequality is because $kd/n \geq 2/\epsilon \geq 1/\epsilon$ by assumption, the fourth inequality is due to $e^{-1/\epsilon} \leq \epsilon$ and $|\neighbors{\opt{k}}|$ is upper bounded by $n$. 
    For $k \leq \frac{\epsilon n}{2d}$:
    \begin{align*}
        \E\left[ |\neighbors{\greedy{k}}| \right] 
        \geq \E\left[ |\neighbors{\smax{k}}| \right]
        \geq \left( 1 \! - \! e^{-kd/n} \right)  n
        = \frac{1 - e^{-kd/n}}{kd/n} \cdot kd 
        \geq \frac{1 \! - \! e^{-\epsilon}}{\epsilon} \cdot kd
        \geq \left( 1 \! - \! \epsilon \right) \cdot \E\left[ |\neighbors{\opt{k}}| \right],
    \end{align*}
    where the first inequality is because $\E\left[ |\neighbors{\greedy{k}}| \right]$ is lower bounded by $\E\left[ |\neighbors{\smax{k}}| \right]$ by \cref{lem:greedy lower bound by smax for n=m}, the second inequality is due to \cref{lem:smax expected coverage for n=m}, the third inequality is because $kd/n \leq \epsilon/2 \leq \epsilon$ by assumption, the fourth inequality is due to \cref{claim one} and $|\neighbors{\opt{k}}|$ is upper bounded by $kd$.
\end{proof}

\paragraph{Critical region: large degrees.}
We use \cref{lem:greedy lower bound by smax for n=m} and  Chernoff bounds over negatively correlated variables to show a $1 - \epsilon$ ratio in the critical region (i.e. $k \approx n/d$) when the degrees are large. 
Intuitively, it is possible to upper bound the optimal solution, since any $k$ nodes suffer from sharing too many neighbors. 
\begin{restatable}{lemma}{CriticalRegionLargeDegreesND}
\label{lem:ratio lower bound for n=m and uniform d - large d}
    Let $\epsilon \in (0, 1),  d \geq \frac{20^4}{\epsilon^8}, n \geq d, k \in \left[ \frac{\epsilon n}{2d}, \frac{2n}{\epsilon d} \right],$ and $\cB \sim \text{\textsc{LRR}}(n, d)$.
    Then we have $\E\left[ |\neighbors{\greedy{k}}| \right] \geq \left( 1 - \epsilon \right) \cdot \E\left[ |\neighbors{\opt{k}}| \right]$.
\end{restatable}
\begin{proof}[Proof sketch (full proof in \cref{sec: analysis for the different regions n=m appendix})]
    We show that the values of all sets concentrate with high probability. To this end, we express the value $|N(S)|  = \sum_{v \in R} \mathbbm{1}\{ v \in N(S) \}$ of a set $S$ as the sum over  $v \in R$ of the indicator variables for whether $v$ is covered by $S$. We first show that these indicator variables are negatively correlated.
    
    We then use a generalization of the Chernoff bound for negatively correlated variables and a union bound to show that  $\Pr\left[ |\neighbors{\opt{k}}| \geq (1 + \delta) \cdot \E\left[ |\neighbors{\smax{k}}| \right] \right]\leq \frac{1}{n}$ for an appropriate choice of $\delta$. 
    It then follows that
    $\E\left[ |\neighbors{\greedy{k}}| \right] \geq \E\left[ |\neighbors{\smax{k}}| \right] \geq \frac{1}{1 + \epsilon}  \E\left[ |\neighbors{\opt{k}} \right]$, where the first inequality is by Lemma~\ref{lem:greedy lower bound by smax for n=m} and the second inequality by the choice of $\delta$ and the assumption on $n$.
\end{proof}

\paragraph{Critical region: small degrees.}
We use \cref{lem:first phase accepts for n=m} to lower bound the ratio in the critical region, when the degrees are small.
The main challenge in this case is that the optimal solution can achieve a high value close to $n$, so \cref{lem:greedy lower bound by smax for n=m,lem:smax expected coverage for n=m} only retrieve worst-case ratio $1 - \frac{1}{e}$. 
Therefore, we need a tighter lower bound for $\greedyalgorithm$.
This is the most challenging case and the only one in which we do not show a $1 - \epsilon$ ratio. Nevertheless, we still improve over the $1 - \frac{1}{e}$ ratio by a constant, albeit small.
\begin{restatable}{lemma}{CriticalRegionSmallDegreesND}
\label{lem:ratio lower bound for n=m and uniform d - small d}
There exists a function $n_{0}(x): (0,1) \rightarrow \mathbbm{N}$ such that for  $\epsilon \in (0,1), n \geq n_{0}(\epsilon), d \leq \frac{20^4}{\epsilon^8}, k \in \left[ \frac{\epsilon n}{2d}, \frac{2n}{\epsilon d} \right],$ and $\cB \sim \text{\textsc{LRR}}(n, d)$, we have $\E\left[ |\neighbors{\greedy{k}}| \right] \geq \left( 1 - \frac{1}{e} + \Omega\left( \epsilon^{24} \right) \right) \cdot \E\left[ |\neighbors{\opt{k}}| \right]$.
\end{restatable}
The proof uses worst-case analysis after iteration $t_d$, to obtain a result for $k \geq t_d$. 
Specifically, we show that as long as $t_d$ is on the order of $n/d$, $\greedyalgorithm$ is bounded a constant away from $1 - \frac{1}{e}$, which we formalize in the following lemma.
\begin{restatable}{lemma}{AugmentedWorstCaseAnalysisN}
\label{lem:augmented worst case analysis for n=m}
    Let $B = (L, R, E)$ be a $d$-left-regular bipartite graph with $|L|=|R|=n$.
    For all $k \geq t_d$ we have $|\neighbors{\greedydet{k}}| \geq \left( 1 - \frac{1}{e} + \frac{1}{e} \cdot \left( \frac{t_d}{n/d} \right)^3 \right) \cdot |\neighbors{\optdet{k}}|$.
\end{restatable}

We are now ready to give a proof  sketch of Lemma~\ref{lem:ratio lower bound for n=m and uniform d - small d}.

\begin{proof}[Proof sketch of \cref{lem:ratio lower bound for n=m and uniform d - small d} (full proof in \cref{sec: analysis for the different regions n=m appendix})]
    Let $\delta$ be an appropriately chosen error term. For $k \leq t^*_d - \delta$, $\greedyalgorithm$ gains $d$ on each iteration with high probability by \cref{lem:first phase accepts for n=m}, so
$        \E\left[ |\neighbors{\greedy{k}}| \right] 
        \geq \left( 1 - \frac{1}{n} \right) \cdot \E\left[ |\neighbors{\opt{k}}| \right].$ For $k \geq t^*_d - \delta$, $\greedyalgorithm$ gains $d$ on each iteration until $t^*_d - \delta$ with high probability by \cref{lem:first phase accepts for n=m}, and for later iterations we use \cref{lem:augmented worst case analysis for n=m}:
    \begin{align*}
        \E\left[ |\neighbors{\greedy{k}}| \right]
        \geq \left( 1 - \frac{1}{e} + \frac{1}{e} \cdot \left( \frac{t^*_d - \delta}{n/d} \right)^3 - \frac{8}{\epsilon n} \right) \cdot \E\left[ |\neighbors{\opt{k}}| \right].
    \end{align*}
    We choose $n_{0}(\epsilon)$ large enough to make terms $\frac{\delta}{n/d} = O(\epsilon^8)$ and $\frac{8}{\epsilon n} = O(\epsilon^{25})$.
    The statement follows by lower bounding $\frac{t^*_d - \delta}{n/d} = \Omega(\epsilon^{8})$.
\end{proof}

\subsection{Putting Everything Together}
\label{sec: putting everything together n=m}

By combining \cref{lem:ratio lower bound for n=m and uniform d - outside critical region k large or small} and \cref{lem:ratio lower bound for n=m and uniform d - large d}, we immediately get the near-optimality of greedy in the near-linear region, saturated region, and critical region with large degrees.
\TheoremTwo*

The second main result regarding the approximation guarantee of greedy requires  the following inequality (proof in \cref{sec: analysis for the different regions n=m appendix}).

\begin{restatable}{claim}{CLAIMTWO}
\label{claim}
    For all $k \geq 1$, $\left( 1- \frac{1}{k}\right)^k \leq \frac{1}{e} - \frac{1}{4ek}$.
\end{restatable}

By combining \cref{lem:ratio lower bound for n=m and uniform d - small d}, which is for the critical region with small degrees, and \cref{thm:Theorem 2}, which is for all other regions, we get the constant improvement over $1-1/e$ for all regions.

\TheoremOne*
\begin{proof}
    Let $\epsilon = 0.3$ and $n_0(x)$ be the function from the statement of \cref{lem:ratio lower bound for n=m and uniform d - small d}.
    For $n \geq n_{0}(\epsilon)$, the statement is implied by \cref{thm:Theorem 2,lem:ratio lower bound for n=m and uniform d - small d}.
    For $n < n_{0}(\epsilon)$, the worst-case approximation ratio of greedy is $1 - \left( 1 - \frac{1}{k} \right)^k$, so we get
    \begin{align*}
         1 - \left( 1 - \frac{1}{k} \right)^k  \geq  1 - \frac{1}{e} + \frac{1}{4ek}  \geq  1 - \frac{1}{e} + \frac{1}{4e n_{0}(\epsilon)},
    \end{align*}
    where  the first inequality is by \cref{claim}  and the second since $k \leq n < n_{0}(\epsilon)$. 
\end{proof}

\section{Limitation of Greedy}
\label{sec:limitation}

In the previous section, \cref{thm:Theorem 2} does not guarantee an approximation ratio close to $1$ for all parameter values. The remaining difficult regime is when the degree $d$ is small and $k$ lies in a critical window around $n/d$.
In this section, we establish \cref{thm:Theorem 3}, which demonstrates that this gap is due to a fundamental limitation of $\greedyalgorithm$ rather than a proof artifact. 
Specifically, we identify a counterexample with $d=2$, where the ratio is bounded away from $1$ as $n$ goes to infinity.
The main results are \cref{lem: lower bound for optimal solution for d=2,lem: coverage of greedy for d=2}, which lower bound $\optalgorithm$ and upper bound $\greedyalgorithm$ respectively.
We defer the complete proofs to \cref{sec: omitted proofs of section 4}.

\subsection{Lower Bound for \texorpdfstring{$\optalgorithm$}{Opt}}

The first challenge is to lower-bound the value of the optimal solution.
We resolve this difficulty by relating the value of the optimal solution in $2$-left-regular random graphs to the size of a maximum matching in Erd\H{o}s-Renyi random graphs. 
This process takes two steps. 

First, the value of the optimal solution in any $2$-left-regular bipartite graph $B = (L, R, E)$ is at least twice the minimum between $k$ and the size of the largest set of left nodes with no common neighbors, i.e. 
\begin{align}\label{eq:lowerboundoptlambda}
    |N(\mathcal{O}_k(B))| \geq 2 \min\left(k,\lambda(B)\right),
\end{align}
with
\begin{align}\label{eq:deflambda}
\lambda(B) = \max\{ |S|:  S \subseteq L \text{ s.t. } \neighbors{u} \cap \neighbors{u'} = \emptyset \ \forall u\neq u' \in S \}.\end{align}

The following lemma reduces the computation of $\lambda(B)$ to finding the size of a maximum matching in a specific graph $I_B$, defined as follows: the vertex set of $I_B$ is the set of right nodes of $B$, and for each left node $u\in L$, we add an edge in $I_B$ between the two neighbors of $u$. If two nodes on the left side share the same neighbor set, the parallel edges are merged; in other words, $I_B$ is a simple graph.

\begin{restatable}{lemma}{ReductionToMatching}
\label{lem: reduction to matching for d=2}
    Let $B = (L, R, E)$ be a $2$-left-regular bipartite graph. 
    Also let $I_B = (R, E_B)$ be a graph such that $E_B = \{ e_1, \dots, e_n \}$ where $e_i = \neighbors{u_i}$ for all $i \in [n]$.
    Then $\lambda(B) = \mu(I_B)$, where $\mu(I_B)$ is the size of the maximum matching in $I_B$.
\end{restatable}
\begin{proof} Consider a set $S \subseteq L$ attaining the maximum in Equation \eqref{eq:deflambda}, i.e., $|S|=\lambda(B)$ and
$\neighbors{u} \cap \neighbors{u'} = \emptyset$ for all distinct $u,u' \in S$.
    We construct \[M = \{ e_i: u_i \in S \},\]
    which is a valid matching in $I_B$ by definition of $S$, and clearly $|M| = |S|$.
    We now show that $M$ is a maximum matching in $I_B$.
    Towards a contradiction, assume there exists a matching $M'$ in $I_B$ such that $|M'| > |M|$.
    Then we can define:
    \begin{align*}
        S' = \{u_i : e_i \in M'\},
    \end{align*}
    where $\neighbors{u_i} \cap \neighbors{u_j} = \emptyset$ for all $u_i, u_j \in S'$ since $M'$ is a matching.
    Then $|S'| = |M'| > |M| = |S| = \lambda(B)$, which contradicts optimality of $S$.
    The statement follows as $\lambda(B) = |S| = |M| = \mu(I_B)$.
\end{proof}

Now suppose that $\cB$ is sampled from $\text{\textsc{LRR}}(n,2)$. Then $\cI_{\cB}$ is distributed as a graph on $n$ vertices obtained by sampling $n$ edges independently and uniformly at random \emph{with replacement} from the $\binom{n}{2}$ possible edges, merging parallel edges. This distribution is closely related to the classical Erd\H{o}s--R\'enyi model $\mathbbm{G}_{n, m}$, in which $m$ edges are sampled uniformly at random \emph{without replacement}. Sharp estimates for the maximum matching size in Erd\H{o}s--R\'enyi graphs are known; see, e.g., \cite{karp1981MaximumMatchingRandomGraphs}.
As we show in the following lemma, these results yield a bound on $\mu(\cI_\cB)$, which can be translated into a lower bound on the expected value of the optimal solution.
Before we proceed, we define constants $\gamma_*, \gamma^*$, where $\gamma_*$ is the smallest root of $x = 2 e^{-2 e^{-x}}$ and $\gamma^* = 2 e^{-\gamma_*}$ since we use them extensively in our proofs.

\begin{restatable}{lemma}{OptCoverageDegreeTwo}
\label{lem: lower bound for optimal solution for d=2}
    Let $n \in \mathbbm{N}$ and $\cB \sim \text{\textsc{LRR}}\left(n, 2\right)$.
    Then for any $k \leq \left(1 - \frac{\gamma^* + \gamma_* + \gamma^* \gamma_*}{4} - n^{-\frac{1}{9}} \right) \cdot n$, we have $\E\left[ |\neighbors{\opt{k}}| \right] 
    \geq (1 - o(1)) \cdot 2 k.$
\end{restatable}

The proof leverages the following lemma, which provides a probabilistic lower bound on the size of the maximum matching in $\cI_\cB$. 

\begin{restatable}{lemma}{LBmaxmatchGB}\label{lem:LBmaxmatchGB}
    Let $n \in \mathbbm{N}$ and $\cB \sim \text{\textsc{LRR}}\left(n, 2\right)$, then
        \[ 
        \Pr\left[ \mu(\cI_\cB) \geq \left(1 - \frac{\gamma^* + \gamma_* + \gamma^* \gamma_*}{4} - n^{-\frac{1}{9}} \right) \cdot n \right]
        \geq 1 - o(1).
        \]
\end{restatable}
The proof, deferred to \cref{sec: omitted proofs of section 4}, bounds how much the maximum matching size can change due to the discrepancy between the generation process of  $\cI_\cB$ and the seminal Erd\H{o}s-Renyi model, and leverages what is known about the size of maximum matchings in Erd\H{o}s-Renyi. We are now ready to prove \cref{lem: lower bound for optimal solution for d=2}.
\begin{proof}[Proof of \cref{lem: lower bound for optimal solution for d=2}] 
Consider any $k\leq \left(1 - \frac{\gamma^* + \gamma_* + \gamma^* \gamma_*}{4} - n^{-\frac{1}{9}} \right) \cdot n$, then:
\begin{align*}
    \E\left[ |\neighbors{\opt{k}}| \right] 
    \geq 2k \cdot \Pr\left[ \lambda(\cB) \geq k \right]
    = 2k \cdot \Pr\left[ \mu(\cI_\cB) \geq k \right],
\end{align*}
where the inequality is due to \cref{eq:lowerboundoptlambda} and the equality is due to $\lambda(\cB) = \mu(\cI_{\cB})$ by \cref{lem: reduction to matching for d=2}.
We now lower bound $\mu(\cI_\cB)$:
\begin{align*}
    \Pr\left[ \mu(\cI_\cB) \geq k \right]
    \geq \Pr\left[ \mu(\cI_\cB) \geq \left(1 - \frac{\gamma^* + \gamma_* + \gamma^* \gamma_*}{4} - n^{-\frac{1}{9}} \right) \cdot n \right]
    \geq 1 - o(1),
\end{align*}
where the first inequality is due to the assumption on $k$ and the second due to \cref{lem:LBmaxmatchGB}.
Chaining with the previous inequality establishes the lemma.
\end{proof}

\subsection{Upper Bound for \texorpdfstring{$\greedyalgorithm$}{Greedy}}
The second challenge is to upper bound the value of $\greedyalgorithm$. We do so through the machinery of \cref{lem:first phase accepts for n=m} in the following lemma.

\begin{restatable}{lemma}{GreedyCoverageDegreeTwo}
\label{lem: coverage of greedy for d=2}
    Let $n \in \mathbbm{N}$ and $\cB \sim \text{\textsc{LRR}}\left(n, 2\right)$.
    Then we have that, for any $k > \frac{1}{3} \cdot n$, $\E\left[ |\neighbors{\greedy{k}}| \right] \leq (1 + o(1)) \cdot \left(\frac{1}{3} \cdot n + k\right)$.
\end{restatable}
\begin{proof}
By definition of $t_2$, up to iteration $t_2$, $\greedyalgorithm$ will gain $2$ per iteration, and thereafter it gains at most $1$ per iteration. Consequently, for any $k$, the following holds  
\begin{equation*}
\mathbb{E}\left[|\neighbors{\greedy{k}}|\right] \leq \mathbb{E}\left[2t_2+(k-t_2)\right]=\mathbb{E}[t_2]+k.\end{equation*}
As $t_2\leq n$ always holds,  we have:
\[
\mathbb{E}[t_2]\leq t^*_2+ 3e^{4} \cdot \sqrt{ 8n \log{2n}}+ n \cdot \Pr\left[ t_2\geq t^*_2+ 3e^{4} \cdot \sqrt{ 8n \log{2n}} \right].
\]
We upper bound the probability using \cref{lem:first phase accepts for n=m}:
\[
\Pr\left[ t_2\geq t^*_2+ 3e^{4} \cdot \sqrt{ 8n \log{2n}} \right] \leq \frac{1}{n}.
\]
Chaining those three inequalities yields:
\begin{equation*}
\mathbb{E}\left[|\neighbors{\greedy{k}}|\right] \leq t^*_2+ 3e^{4} \cdot \sqrt{ 8n \log{2n}} +1+k=\left(1+\frac{3e^{4} \cdot \sqrt{ 8n \log{2n}} +1}{t^*_2+k}\right)(t^*_2+k).
\end{equation*}
The statement then follows, since $t^*_2= n/3$ by definition.
\end{proof}

\subsection{Putting Everything Together}
We proceed to show \cref{thm:Theorem 3} by combining \cref{lem: lower bound for optimal solution for d=2,lem: coverage of greedy for d=2}.
\TheoremThree*
\begin{proof}
    Let $k(n) = \left\lfloor\left(1 - \frac{\gamma^* + \gamma_* + \gamma^* \gamma_*}{4} - n^{-\frac{1}{9}} \right) \cdot n \right\rfloor$, then we have:
    \begin{align*}
        k(n)
        \geq \left(1 - \frac{\gamma^* + \gamma_* + \gamma^* \gamma_*}{4} - n^{-\frac{1}{9}} \right) \cdot n - 1
        = \left(1 - \frac{\gamma^* + \gamma_* + \gamma^* \gamma_*}{4} - n^{-\frac{1}{9}} - n^{-1} \right) \cdot n. \tag{1}
    \end{align*}
    We upper bound the ratio:
    \begin{align*}
        \frac{\E\left[ |\neighbors{\greedy{k(n)}}| \right]}{\E\left[ |\neighbors{\opt{k(n)}}| \right]}
        \leq \frac{\E\left[ |\neighbors{\greedy{k(n)}}| \right]}{(1 - o(1)) \cdot 2k(n)}
        \leq \frac{(1 + o(1)) \cdot \left( \frac{1}{3} \cdot n + k(n) \right)}{(1 - o(1)) \cdot 2k(n)}
        \overset{n \rightarrow \infty}{\rightarrow} \frac{1}{2} + \frac{\frac{1}{3}}{2 - \frac{\gamma^* + \gamma_* + \gamma^* \gamma_*}{2}}
        \leq 0.94,
    \end{align*}
    where the first inequality is due to \cref{lem: lower bound for optimal solution for d=2}, the second inequality is due to \cref{lem: coverage of greedy for d=2} which applies since $k(n) \geq t^*_2 = \frac{1}{3} \cdot n$ for large enough $n$, the limit is due to Equation (1) and because all $o(1)$ terms vanish, and the third inequality is because both $\gamma_*, \gamma^* \leq 0.853$.
\end{proof}

\newpage
\printbibliography

\newpage 

\appendix

\section{Additional Discussion about the Left-Regular Random Model}

\subsection{A Simple Bad Instance for Greedy over a Left-Regular Graph}
\label{sec:badinstance}
The following instance corresponds to the classical bad instance for greedy originally provided in~\cite{submodular-inapproximability} to show that greedy achieves at best a $1-1/e$ approximation, as described using the notation from~\cite{pokutta2020unreasonable}. We highlight that this simple and original instance from~\cite{submodular-inapproximability} is left-regular with all the left nodes having equal degree. The instance is the following
\begin{itemize}
\item cardinality constraint $k$,
\item right nodes $R = \{1, \ldots, k\}^k$ with $|R| = k^k$, 
\item left nodes $L = \{a_1, \ldots, a_{k-1}\} \cup \{b_1, \ldots, b_k\}$ with 
\begin{itemize}
\item $N(a_i) = \{x \in R: x_i = 1\}$ for $i \in [k-1]$, and
\item $N(b_i) = \{x \in R: x_k = i\}$ for $i \in [k]$
\end{itemize}
\end{itemize}

We note that this graph is left-regular with all the nodes in $L$ having degree $k^{k-1}$. The optimal solution is $\{b_1, \ldots, b_k\}$, which covers all the right nodes and thus has value $k^k$. Note that for any $S \subseteq \{a_1, \ldots, a_{k-1}\}$, the marginal contribution of each node $u \in L \setminus S$ to $S$ is identical and equal to $\frac{1}{k} \cdot | R \setminus N(S)|$, i.e., a $1/k$ fraction of the right nodes not covered by $S$. If greedy breaks ties according to order $(a_1, \ldots, a_{k-1}, b_1, \ldots, b_k)$, it returns solution $\{a_1, \ldots, a_{k-1}, b_1\}$, which has value 
$$\left(1 - \left(1 - \frac{1}{k}\right)^k\right)|R|,$$
whose limit approaches $1-1/e$ as $k$ grows large. This bipartite graph can of course be made balanced with $|L| = |R|$ by adding a sufficient number of copies of an arbitrary left node.

\subsection{Empirical Performance of Greedy over Random Graphs with Uneven Left-Node Degrees}
\label{sec:unequal degrees}
We evaluate $\greedyalgorithm$ over random graphs with left nodes whose degrees are either equal or have different level of unevenness. We use this family of instances as an example to illustrate that, for random graphs and at least experimentally, greedy seems to  perform worse when the degrees of the left nodes are all equal.

To make the comparison clear, we use parameter $a \in [0,1]$ to interpolate between equal degrees ($a = 0$) and unequal degrees ($a = 1$), while keeping the average degree fixed.
The instances we create for our experiments have left degrees randomly chosen from $1, 4, 7$ with probabilities $\frac{a}{2}, 1-a, \frac{a}{2}$, so that the average degree is always $4$.
For each value of $a$, we estimate:
\begin{align*}
    \min_{k \in [n]} \frac{\E\left[ |\neighbors{\greedy{k}}| \right]}{\E\left[ |\neighbors{\opt{k}}| \right]}.
\end{align*}
As we can see in \cref{fig:uniform vs non-uniform}, the lowest ratio is achieved for equal degrees ($a = 0$).

\begin{figure}[H] 
    \centering
    \includegraphics[width=0.7\textwidth]{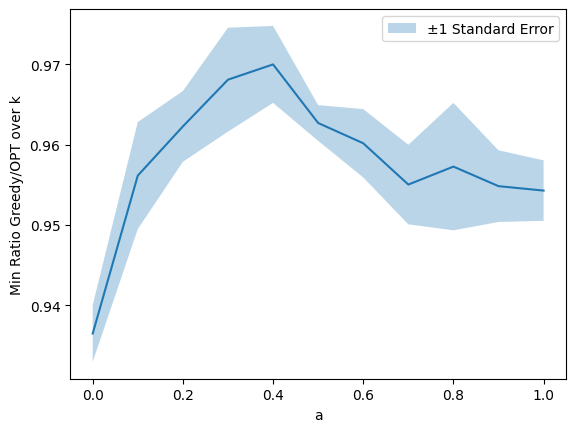}
    \caption{The computation is carried out for $n = m = 100$ and the ratio between the value of $\greedyalgorithm$ and $\optalgorithm$ is calculated by averaging over $50$ runs. $\optalgorithm$ was implemented in Python using Gurobi to solve the integer program. The shaded region depicts the standard error for our estimation.}
\label{fig:uniform vs non-uniform}
\end{figure}
\section{Generalization Of Technical Tools}
\label{sec:generalization of technical tools}

In this section, we show that our main techniques from the proof of \cref{thm:Theorem 1,thm:Theorem 2}, can be used in a more general random model that allows different number of nodes on each side of the bipartite graph as well as non-uniform degrees for the left nodes.
\begin{definition}
    Let $n, m \in \mathbbm{N}$ and $d \in [m]$. The random model \textsc{GenR}($n, m, d_1, \dots, d_n$) constructs a random bipartite graph $\cB = \left( L, R, \cE \right)$ such that the number of nodes is $|L| = n, |R| = m$ and the edges $\cE$ are such that $\neighbors{u_i} \in \binom{R}{d_i}$ is selected uniformly and independently for all $i \in [n]$.
\end{definition}
Notice that, $\textsc{GenR}$ generalizes $\textsc{LRR}$ and $\textsc{ULRR}$.

\subsection{Lower Bound By Differential Equation Method}
The following lemma shows that this algorithm is equivalent to $\greedyalgorithm$ for arbitrary $d$-left-regular bipartite graphs.
\begin{restatable}{lemma}{AcceptRejectEquivalence}
\label{lem: greedy is equivalent to acceptreject}
    Let $B = (L, R, E)$ be a $d$-left-regular bipartite graph. 
    Let $P_t = \max\{p \in [d]: |\acceptreject{p}{0}| \leq t\}$ 
    and $I_t = \min\{i \in [n]: |\acceptreject{P_t}{i}| \geq t\}$. 
    Then $\acceptrejectdet{p}{i} = \greedydet{t}$ for all $(p, i) \in \{(p', i'): p'=P_t, I_t \leq i' \leq n \text{ or } P_{t+1} < p' < P_{t}, i' \in [n] \text{ or } p' = P_{t+1}, 0 \leq i' < I_{t+1}\}$ and $t \in [k]$.
\end{restatable}
\begin{proof}
    Notice that $I_t$ is by definition the iteration on which $\acceptrejectalgorithm$ accepts the $t$-th node in phase $P_t$ and the additional nodes are never accepted.
    We start by showing that $\acceptreject{P_t}{I_t} = \greedy{t}$ by induction on $t$.
    For $t = 1$, we have that:
    \begin{align*}
        |\neighbors{u_1} \setminus \neighbors{\greedy{0}}| = |\neighbors{u_1} \setminus \neighbors{\acceptreject{d}{0}}| = d,
    \end{align*}
    where the first equality is due to $\greedy{0} = \acceptreject{d}{0} = \emptyset$. 
    Therefore $\acceptreject{d}{1} = \{u_1\}$ since on iteration $1$ it will accept $u_1$ and $\greedy{1} = \{u_1\}$ since among all the nodes with contribution $d$, it has the lowest index.
    For the inductive step $t = t_0 + 1$, let $u_g$ be the node such that $\greedy{t_0 + 1} \setminus \greedy{t_0} = \{u_g\}$.
    By definition of $\greedyalgorithm$ that implies:
    \begin{align*}
        g = \min\{i \in [n]: |\neighbors{u_{i}} \setminus \neighbors{\greedy{t_0}}| \geq |\neighbors{u_{i'}} \setminus \neighbors{\greedy{t_0}}|, \forall i' \in [n]\} \tag{*}.
    \end{align*}
    Also let $u_a$ be the node such that $\acceptreject{P_{t_0 + 1}}{I_{t_0 + 1}} \setminus \acceptreject{P_{t_0}}{I_{t_0}} = \{u_a\}$ and we have that:
    \begin{align*}
        |\neighbors{u_a} \setminus \neighbors{\acceptreject{P_{t_0}}{I_{t_0}}}|
        = |\neighbors{u_a} \setminus \neighbors{\greedy{t_0}}|
        \leq |\neighbors{u_g} \setminus \neighbors{\greedy{t_0}}| 
        = |\neighbors{u_g} \setminus \neighbors{\acceptreject{P_{t_0}}{I_{t_0}}}|,  \tag{**}
    \end{align*}
    where the equality is by inductive hypothesis $\acceptreject{P_{t_0}}{I_{t_0}} = \greedy{t_0}$ and the inequality by (*).
    Assume towards a contradiction that $|\neighbors{u_a} \setminus \neighbors{\acceptreject{P_{t_0}}{I_{t_0}}}| < |\neighbors{u_g} \setminus \neighbors{\acceptreject{P_{t_0}}{I_{t_0}}}|$, so $u_a$ is accepted on phase $|\neighbors{u_a} \setminus \neighbors{\acceptreject{P_{t_0}}{I_{t_0}}}|$, which is later than phase $|\neighbors{u_g} \setminus \neighbors{\acceptreject{P_{t_0}}{I_{t_0}}}|$ and implies that $\acceptrejectalgorithm$ rejected $u_g$ on phase $|\neighbors{u_g} \setminus \neighbors{\acceptreject{P_{t_0}}{I_{t_0}}}|$. 
    Overall we have that:
    \begin{align*}
        |\neighbors{u_a} \setminus \neighbors{\acceptreject{P_{t_0}}{I_{t_0}}}| 
        = |\neighbors{u_g} \setminus \neighbors{\acceptreject{P_{t_0}}{I_{t_0}}}|. \tag{***}
    \end{align*}
    Assume towards a contradiction that $a \neq g$.
    We start with the case $a < g$ which contradicts the minimality of $g$ in (*) due to (***).
    On the other hand, if $a > g$, $u_g$ was reached in an earlier iteration of $\acceptrejectalgorithm$ and must have been rejected in phase $|\neighbors{u_a} \setminus \neighbors{\acceptreject{P_{t_0}}{I_{t_0}}}|$ which is a contradiction by (***).
    Overall we have that:
    \begin{align*}
        \acceptreject{P_{t_0 + 1}}{I_{t_0 + 1}}
        = \acceptreject{P_{t_0}}{I_{t_0}} \cup \{u_a\}
        = \acceptreject{P_{t_0}}{I_{t_0}} \cup \{u_g\} 
        = \greedy{t_0} \cup \{u_g\} 
        = \greedy{t_0 + 1},
    \end{align*}
    where the first equality is by definition of $P_t, I_t$, the second due to $g = a$, the third is by inductive hypothesis and the fourth by definition of $u_g$.
    This concludes the induction and in order to obtain the statement note that for any $(p, i) \in \{(p', i'): p'=P_t, I_t \leq i' \leq n \text{ or } P_{t+1} < p' < P_{t}, i' \in [n] \text{ or } p' = P_{t+1}, 0 \leq i' < i_{t+1}\}$ we have that $\acceptreject{p}{i} = \acceptreject{P_t}{I_t}$ since no node is accepted in these iterations.
\end{proof}

\subsubsection{Number of accepts in first phase}

In the following lemma we calculate the accepts of phase 1 by using the differential equation method.
\begin{lemma}
\label{lem:first phase number of accepts}
    Let $n, m \in \mathbbm{N}$, $d \in [m], k \in [n]$ and draw bipartite graph $\cB = (L, R, \cE)$ from random model $\textsc{GenR}(n, m, d, \dots, d)$.
    Also let $t^*_d = \left( 1 - \left( 1 + \frac{n d(d-1)}{m} \right)^{-\frac{1}{d-1}} \right) \cdot \frac{m}{d}$ and $\delta = 3e^{\frac{n d^2}{m}} \cdot \sqrt{8n\log{2n}}$.
    Then $|t_d - t^*_d | \leq \delta$ with probability at $1 - \frac{1}{n}$ for all $m \geq 2d^2$.
\end{lemma}
\begin{proof}
    Due to the equivalence between $\acceptrejectalgorithm$ and $\greedyalgorithm$ \cref{lem: greedy is equivalent to acceptreject}, on any graph $B = (L, R, E)$ we have $t_d = |\acceptreject{d}{n}|$, since at the end of phase $d$, $\acceptrejectalgorithm$ has accepted all the nodes with marginal contribution $d$. Therefore, it suffices to estimate $|\acceptreject{d}{n}|$. Towards this end we define $Y(i) = |\acceptreject{d}{i}|$, i.e. the number of accepted nodes on iteration $i$ of phase $d$ and we have:
    \begin{align*}
        \E\left[ Y(i+1) - Y(i) \middle| Y(i) \right] 
        = \E\left[ \mathbbm{1}\{u_{i+1} \text{ accepted }\} \middle| Y(i) \right]
        = \Pr\left[ u_{i+1} \text{ accepted } \middle| Y(i) \right]
        = \frac{\binom{m - Y(i) d}{d}}{\binom{m}{d}},
    \end{align*}
    where the last equality is because $u_{i+1}$ is accepted only if $|\neighbors{u_{i+1}} \setminus \neighbors{\acceptreject{d}{i}}| = d$. But since $|\neighbors{\acceptreject{d}{i}}| = Y(i) \cdot d$ and $\neighbors{u_{i+1}}$ are selected uniformly at random from the $m$ right nodes, the only way $u_{i+1}$ was accepted was if it sampled $d$ nodes from the $m - Y(i) \cdot d$ remaining.   
    We can now show one by one that the assumptions of the differential equation hold:
    \begin{enumerate}
        \item Trend Hypothesis: $\bigabs{ \E\left[ Y(i+1) - Y(i) \middle| Y(i) \right] - \left( 1 - \frac{nd}{m} \frac{Y(i)}{n} \right)^d } \leq \frac{2d^2}{m}$ by \cref{lem: approximation of success probability} for $r = m - Y(i) d$ since $m \geq 2d^2$,
        \item Lipschitz Hypothesis: $F(y) = \bigpar{1 - \frac{nd}{m} y}^d$ is $\frac{nd^2}{m}$-Lipschitz since: 
        \begin{align*}
            \bigabs{F(y_1) - F(y_2)} 
            &= \bigabs{ \bigpar{1 - \frac{nd}{m} y_1}^d - \bigpar{1 - \frac{nd}{m} y_2}^d } \\
            &= \frac{nd}{m} \cdot \bigabs{ \bigpar{\bigpar{1 - \frac{nd}{m} y_1}^{d-1} + \dots + \bigpar{1 - \frac{nd}{m} y_2}^{d-1}}} \cdot \bigabs{ y_2 - y_1 } \\
            &\leq \frac{nd}{m} \cdot \bigabs{ 1 + \dots + 1 } \cdot \bigabs{ y_2 - y_1 } \\
            &= \frac{nd^2}{m} \cdot \bigabs{ y_1 - y_2 },
        \end{align*}
        \item Boundedness Hypothesis: $\bigabs{Y(i+1) - Y(i)} \leq 1$ as only node $u_{i+1}$ can be accepted on iteration $i+1$,
        \item Initial Condition: $Y(0) = \hat{y} = 0$ so any $\lambda \geq \frac{2}{n} \geq \frac{d^2}{m} \min\{T, \frac{m}{nd^2} \} + \frac{R}{n} = \frac{d^2}{m} \min\{1, \frac{m}{nd^2} \} + \frac{1}{n}$ works. We choose $\lambda = \sqrt{\frac{8\log{(2n)}}{n}}$.
    \end{enumerate}
    
    Model $Y(i)$ by $n y\left( \frac{i}{n} \right)$ and we obtain the following differential equation:
    \begin{align*}
        y'(t) = \left( 1 - \frac{nd}{m} \cdot y(t) \right)^d,\ y(0) = 0.
    \end{align*}
    By \cref{thm:differential equation method} with probability $1 - 2 e^{-n \lambda^2 / 8 T} = 1 - \frac{1}{n}$ we have that:
    \begin{align*}
        \bigabs{Y(i) - n y(i/n)} < 3 e^{\frac{nd^2}{m} T} \cdot \lambda n = 3 e^{\frac{nd^2}{m}} \cdot \sqrt{8n \log(2n)}.
    \end{align*}
\end{proof}

In the following lemma we approximate the probability of success.
\begin{lemma} \label{lem: approximation of success probability}
    Let $d \in \mathbbm{N}$, $m \geq 2d^2$, $r \in [m]$
    then $ |\binom{r}{d} / \binom{m}{d} -  \left( \frac{r}{m} \right)^d | \leq \frac{2d^2}{m}$.
\end{lemma}
\begin{proof}
    We start by showing the upper bound:
    \begin{align*}
        \frac{\binom{r}{d}}{\binom{m}{d}}
        = \frac{r (r - 1) \cdots (r - d + 1))}{m (m - 1) \cdots (m - d + 1)}
        = \prod_{i=0}^{d-1} \left( \frac{r - i}{m - i} \right)
        \leq \left( \frac{r}{m} \right)^d.
    \end{align*}
    We proceed to show the lower bound:
    \begin{align*}
        \frac{\binom{r}{d}}{\binom{m}{d}}
        = \prod_{i=0}^{d-1} \left( \frac{r - i}{m - i} \right)
        = \prod_{i=0}^{d-1} \left( \frac{r}{m} \cdot \frac{1 - \frac{i}{r}}{1 - \frac{i}{m}} \right)
        &= \left( \frac{r}{m} \right)^d \cdot \prod_{i=0}^{d-1} \left( 1 - \frac{i(m-r)}{r (m-i)} \right) \\
        &\geq \left( \frac{r}{m} \right)^d \cdot \left( 1 - \sum_{i=0}^{d-1} \frac{i(m-r)}{r (m-i)} \right) \\
        &\geq \left( \frac{r}{m} \right)^d \cdot \left( 1 - \frac{m-r}{r(m-d)} \cdot \sum_{i=0}^{d-1}i \right) \\
        &= \left( \frac{r}{m} \right)^d \cdot \left( 1 - \frac{m-r}{r(m-d)} \cdot \frac{d(d-1)}{2} \right) \\
        &= \left( \frac{r}{m} \right)^d - \left( \frac{r}{m} \right)^{d-1} \cdot \frac{m-r}{m-d} \cdot \frac{d^2}{2m} \\
        &\geq \left( \frac{r}{m} \right)^d - \frac{m}{m-d} \cdot \frac{d^2}{2m} \\
        &\geq \left( \frac{r}{m} \right)^d - \frac{2d^2}{m},
    \end{align*}
    where the first inequality is due to $\prod_{i=0}^{d-1}(1 - x_i) \geq 1 - \sum_{i=0}^{d-1} x_i$ for all $x_i \in [0,1]$, the second inequality is due to $m-i \geq m-d$ for all $i \in [d-1]$, the third inequality is due to $\frac{r}{m} \leq 1$ and the fourth inequality is because $d \leq \sqrt{m}$, so $\frac{m}{m-d} \leq \frac{m}{m - \sqrt{m}} = \frac{1}{1 - 1/\sqrt{m}} \leq 4$ for all $m \geq 2d^2 \geq 2$.
\end{proof}

In the following lemma we solve the main differential equation.
\begin{lemma} \label{lem: first phase differential equation}
    The unique solution of $y'(t) = \left( 1 - \frac{nd}{m} \cdot y(t) \right)^d$ with initial condition $y(0) = 0$ is 
    $y(t) = \frac{m}{nd} \cdot \left( 1 - \left( 1 + \frac{nd(d-1)}{m} \cdot t \right)^{- \frac{1}{d-1}} \right)$.
\end{lemma}
\begin{proof}
    We rewrite the differential equation in the following form:
    \begin{align*}
        \frac{y'(t)}{\left( 1 - \frac{nd}{m} \cdot y(t)\right)^d} = 1,\ y(0) = 0. \tag{1}
    \end{align*}
    By integrating both sides of Equation (1) in $[0, t]$ we have:
    \begin{align*}
        \frac{\left( 1 - \frac{nd}{m} y(t) \right)^{-(d-1)}}{\frac{nd(d-1)}{m}} - \frac{1}{\frac{nd(d-1)}{m}} = t. 
    \end{align*}
    We obtain the statement by solving for $y(t)$.
\end{proof}

\subsubsection{Differential Equation Method}
\setcounter{equation}{0}
\begin{lemma}[Theorem 2 in \cite{warnke2019wormaldsdifferentialequationmethod}]\label{thm:differential equation method}%
Given integers~$a,n \ge 1$, 
a bounded domain~$\cD \subseteq \RR^{a+1}$, 
functions~$(F_k)_{1 \le k \le a}$ with~$F_k:\cD \to \RR$, 
and $\sigma$-fields~$\cF_0 \subseteq \cF_1 \subseteq \cdots$, 
suppose that the random variables~$((Y_k(i))_{1 \le k \le a}$ are~$\cF_i$-measurable for~$i \ge 0$.   
Furthermore, assume that, for all~$i \ge 0$ and~$1 \le k \le a$, the following 
conditions hold whenever~$(i/n,Y_1(i)/n,...,Y_a(i)/n) \in \cD$: 
%
\begin{enumerate}%
\item[(i)]\label{dem:trend}%
 $\bigabs{\E\bigpar{Y_k(i+1)-Y_k(i) \mid \cF_{i}}-F_k\bigpar{i/n,Y_1(i)/n,...,Y_a(i)/n}} \le \delta$, 
where the function~$F_k$ is~$L$-Lipschitz-continuous on~$\cD$ \ \emph{(the `Trend hypothesis' and `Lipschitz hypothesis')}, 
\item[(ii)]\label{dem:bounded}%
$\bigabs{Y_k(i+1)-Y_k(i)}\le \beta$ \ \emph{(the `Boundedness hypothesis')}, 
\end{enumerate}
and that the following condition holds initially: 
\begin{enumerate}%
\item[(iii)]\label{dem:init}%
$\max_{1 \le k \le a} \bigabs{Y_k(0)- \hat{y}_k n} \le \lambda n$ for some~$(0,\hat{y}_1, \ldots, \hat{y}_a) \in \cD$ 
 \ \emph{(the `Initial condition')}. 
\end{enumerate}
Then there are~$R=R(\cD,(F_k)_{1 \le k \le a},L) \in [1,\infty)$ and~$T=T(\cD) \in (0,\infty)$ such that, 
whenever~$\lambda \ge \delta \min\{T,L^{-1}\} + R/n$,  
with probability at least $1-2a e^{-n\lambda^2/(8T \beta^2)}$ we have 
\begin{equation}\label{dem:error}
\max_{0 \le i \le \sigma n} \max_{1 \le k \le a}\bigabs{Y_k(i)-y_k\bigpar{\tfrac{i}{n}}n} \; < \; 3 e^{L T} \lambda n ,
\end{equation}
where~$(y_k(t))_{1 \le k \le a}$ is the unique solution to the system of differential equations
\begin{equation}\label{dem:sol}
y'_k(t) =F_k\bigpar{t,y_1(t), \ldots, y_a(t)} \quad \text{ with } \quad y_k(0) = \hat{y}_k \qquad \text{for~$1 \le k \le a$,}
\end{equation}
and~$\sigma=\sigma(\hat{y}_1, \ldots \hat{y}_a) \in [0,T]$ is any choice of~$\sigma \ge 0$ with the property
that~$(t,y_1(t), \ldots y_a(t))$ has~$\ell^{\infty}$-distance at least~$3 e^{L T} \lambda$ 
from the boundary~of~$\cD$ for all~$t \in [0,\sigma)$.  
\end{lemma}

\subsection{Fixed Set Lower Bound}
We proceed to show that $\greedyalgorithm$ achieves higher value than the fixed set of the $k$ nodes with the highest degrees, $\smax{k}$. First, we calculate the expected value of this set.

\begin{lemma} 
\label{lem:smax expected coverage}
    Let $m, n \in \mathbbm{N}$, $m \geq d_1 \geq \dots \geq d_n \geq 0, k \in [n]$ and draw bipartite graph $\cB = (L, R, \cE)$ from random model $\textsc{GenR}(n, m, d_1, \dots, d_n)$.
    Then $\E[|\neighbors{\smax{k}}|] = \left(1 - \prod_{i=1}^{k}\left( 1 - \frac{d_i}{m} \right) \right) \cdot m \geq \left( 1 - e^{- \frac{\sum_{i=1}^{k} d_i }{m}} \right) \cdot m$.
\end{lemma}
\begin{proof}
    Notice that $\smax{k}$ covers node $v_j$ as long as at least one of its nodes covers $v_j$, so for all $j \in [m]$: 
    \begin{align*}
        \Pr\left[ v_j \in \neighbors{\smax{k}} \right] 
        = 1 - \prod_{u_i \in \smax{k}}{ \left( 1 - \Pr\left[ v_j \in \neighbors{u_i} \right] \right) }
        = 1 - \prod_{i = 1}^{k}{\left(1 - \frac{d_i}{m}\right)}
        \geq 1 - e^{- \frac{\sum_{i=1}^{k} d_i }{m}}.
    \end{align*}
    We have that:
    \begin{align*}
        \E\left[ |\neighbors{\smax{k}}| \right] 
        = \E\left[ \sum_{j \in [m]}{\mathbbm{1}\{ v_j \in \neighbors{\smax{k}} \} } \right] 
        = \sum_{j \in [m]}{\Pr[v_j \in \neighbors{\smax{k}}]}
        = \Pr\left[ v_1 \in \neighbors{\smax{k}} \right] \cdot m,
    \end{align*}    
    where the first equality is by definition, the second equality is by linearity of expectation and the third equality because the probabilities are equal for all $j \in [m]$. Then the statement follows.
\end{proof}

We introduce notation that is useful for the proofs in this section. 
Informally, $M_t(S)$ is the set of the $t$ nodes with the highest marginal contributions to a set of nodes $S$.
Let $S \subseteq L$ and define:
\begin{align*}
    u^{(i)} = 
    \begin{cases}
        \arg \max\{|\neighbors{u} \setminus \neighbors{S}|: u \in L \setminus S\} &\text{ if } i = 1 \\
        \arg \max\{|\neighbors{u} \setminus \neighbors{S}|: u \in L \setminus \left(S \cup \{u^{(1)}, \dots, u^{(i-1)}\}\right)\} &\text{ otherwise }
    \end{cases},
\end{align*}
where we abuse notation so that $\arg \max$ breaks ties like $\greedyalgorithm$, which uniquely defines $u^{(i)}$.
Then we define $M_t(S) = \{u^{(1)}, u^{(2)}, \dots, u^{(t)} \}$.
Now we define the hybrid algorithm:
\begin{align*}
    Y^t = \greedydet{t} \cup M_{k-t}(\greedydet{t}).
\end{align*}
Notice that $Y^k = \greedydet{k}$ while $Y^0 = \smax{k}$.

\begin{lemma} \label{lem: greedy lower bound by smax}
    Let $m, n \in \mathbbm{N}$, $m \geq d_1 \geq \dots \geq d_n \geq 1$, $k \in [n]$ and draw graph $\mathcal{B}$ from \textsc{GenR}$(n, m, d_1, \dots, d_n)$. Then $\E[\f{\greedy{k}}] \geq \E[|\neighbors{\smax{k}}|]$ for any $n, m, k$.
\end{lemma}
\begin{proof} 
    We claim that $\E[|\neighbors{\mathcal{Y}^{t}}|] \geq \E[|\neighbors{\mathcal{Y}^{t-1}}|]$ for any $t \in [k]$.
    Assuming it holds we apply it repeatedly and we have that:
    \begin{align*}
        \E[|\neighbors{\greedy{k}}|] 
        = \E[|\neighbors{\mathcal{Y}^{k}}|]
        \geq \E[|\neighbors{\mathcal{Y}^{k-1}}|]
        \geq \dots 
        \geq \E[|\neighbors{\mathcal{Y}^{0}}|]
        = \E[|\neighbors{\smax{k}}|],
    \end{align*}
    which implies the statement.
    For the rest of the proof, we fix $t$ and focus on showing the claim.
    We denote $\cM_{k-t}(\greedy{t}) = \{u^{(1)}, \dots, u^{(k-t)}\}$ and $\cM_{k-t+1}(\greedy{t-1}) = \{u'^{(1)}, \dots, u'^{(k-t+1)}\}$.
    Let $X_l = \neighbors{u^{(\ell)}} \setminus \neighbors{\greedy{t}}$ and $Y_{\ell} = \neighbors{u'^{(\ell+1)}} \setminus \neighbors{\greedy{t}}$ for all $\ell \in [k-t]$.
    For the left-hand side we have that:
    \begin{align*}
        |\neighbors{\mathcal{Y}^t}| 
        &= |\neighbors{\greedy{t} \cup \cM_{k-t}(\greedy{t})}| \\
        &= |\neighbors{\greedy{t}} \cup \neighbors{\cM_{k-t}(\greedy{t})}| \\
        &= |\neighbors{\greedy{t}} \cup \neighbors{\{u^{(1)}, \dots, u^{(i)}\}}| \\
        &= |\neighbors{\greedy{t}}| + |\neighbors{\{u^{(1)}, \dots, u^{(k-t)}\}} \setminus \neighbors{\greedy{t}}| \\
        &= |\neighbors{\greedy{t}}| + |\bigcup_{\ell=1}^{k-t}{\neighbors{u^{(\ell)}} \setminus \neighbors{\greedy{t}}}| \\
        &= |\neighbors{\greedy{t}}| + |\bigcup_{\ell=1}^{k-t}{X_{\ell}}|, \tag{1}
    \end{align*}
    where the first line is by definition of $\mathcal{Y}^t$, the second line is due to $\neighbors{A \cup B} = \neighbors{A} \cup \neighbors{B}$ for any sets of nodes $A, B \subseteq L$, the third line is by definition of $\cM_{k-t}(\greedy{t})$, the fourth line is by splitting the cardinality of the union with the sum of of cardinalities of disjoint sets, the fifth line is by $\neighbors{A \cup B} = \neighbors{A} \cup \neighbors{B}$ and the sixth line is by definition of $X_{\ell}$.
    For the right-hand side we have that:
    \begin{align*}
        |\neighbors{\mathcal{Y}^{t-1}}| 
        &= |\neighbors{\greedy{t-1} \cup \cM_{k-t+1}(\greedy{t-1})}| \\
        &= |\neighbors{ (\greedy{t-1} \cup \{u'^{(1)}\} ) \cup (\cM_{k-t+1}(\greedy{t-1}) \setminus \{u'^{(1)}\} )}| \\
        &= |\neighbors{ \greedy{t} \cup (\cM_{k-t+1}(\greedy{t-1}) \setminus \{u'^{(1)}\} )}| \\
        &= |\neighbors{ \greedy{t} } \cup \neighbors{\cM_{k-t+1}(\greedy{t-1}) \setminus \{u'^{(1)}\}}| \\
        &= |\neighbors{ \greedy{t} } \cup \neighbors{\{u'^{(2)}, \dots, u'^{(k-t+1)}\}}| \\
        &= |\neighbors{ \greedy{t} }| + |\neighbors{\{u'^{(2)}, \dots, u'^{(k-t+1)}\}} \setminus \neighbors{ \greedy{t} }| \\
        &= |\neighbors{ \greedy{t} }| + |\bigcup_{\ell=1}^{k-t}{\neighbors{u'^{(\ell+1)}} \setminus \neighbors{ \greedy{t} }}| \\
        &= |\neighbors{\greedy{t}}| + |\bigcup_{\ell=1}^{i=k-t}{Y_{\ell}}|, \tag{2}
    \end{align*}
    where the first equality is by definition of $\mathcal{Y}^{t-1}$, the second follows by removing $u'^{(1)}$ from the right set and adding it to the left which does not affect the union, the third is due to $\greedy{t} = \greedy{t-1} \cup \{u'^{(1)}\}$ which follows by definition of $\greedy{t-1}$ and $u'^{(1)}$, the fourth is due to $\neighbors{A \cup B} = \neighbors{A} \cup \neighbors{B}$ for any sets of nodes $A, B \subseteq L$, the fifth is by definition of $\cM_{k-t+1}(\greedy{t-1})$, the sixth is by splitting the cardinality of the union with the sum of of cardinalities of disjoint sets and the seventh is by $\neighbors{A \cup B} = \neighbors{A} \cup \neighbors{B}$.

    Let $A = \{a_1, \dots, a_t\} \subseteq L$, $T \subseteq R$ and $D_i \in [d_i]$ for all $i$ such that $u_i \in L \setminus A$ arbitrarily. Also let $\mathscr{F}$ be the family of all edge sets $E$ such that $\greedy{t} = A, \neighbors{A} = T, |\neighbors{u_i} \setminus T| = D_i$ for all $i$ such that $u_i \in L \setminus A$.
    We now show that that:
    \begin{align*}
        \E\left[ |\bigcup_{\ell=1}^{k-t}{X_{\ell}}| \mid \cE \in \mathscr{F} \right] 
        \geq \E\left[ |\bigcup_{\ell=1}^{k-t}{Y_{\ell}}| \mid \cE \in \mathscr{F} \right], \tag{3}
    \end{align*}
    by using \cref{lem: uniform A1 A2 dominate B1 B2 then their union dominates as well}, which requires the following:
    \begin{enumerate}
        \item $X_1, \dots, X_{k-t}$ and $Y_1, \dots, Y_{k-t}$ are pairwise independent, conditionally on $\mathscr{F}$,
        \item $X_{\ell}, Y_{\ell}$ uniform in $R \setminus \neighbors{\greedy{t}}$ for all $\ell \in [k-t]$, conditionally on $\mathscr{F}$,
        \item $\E[|X_{\ell}| \mid \cE \in \mathscr{F}] \geq \E[|Y_{\ell}| \mid \cE \in \mathscr{F}]$.
    \end{enumerate}
    For the first point, by \cref{lem: marginal contributions are independent across remaining nodes} we have that $X_1, \dots, X_i$ are pairwise independent and $Y_1, \dots Y_i$ are pairwise independent.
    For the second point, by \cref{lem: marginal contributions are uniformly distributed}
    we have that $X_{\ell}$ and $Y_{\ell}$ are uniform in $R \setminus \neighbors{\greedy{t}}$ for all $\ell \in [k-t]$.
    For the third point, notice that:
    \begin{align*}
        \E[|X_{\ell}| \mid \cE \in \mathscr{F}]
        = \E[|\neighbors{u^{(\ell)}} \setminus \neighbors{\greedy{t}}| \mid \cE \in \mathscr{F}] 
        &\geq \E[|\neighbors{u'^{(\ell+1)}} \setminus \neighbors{\greedy{t}}| \mid \cE \in \mathscr{F}] \\
        &= \E[|Y_{\ell}| \mid \cE \in \mathscr{F}],
    \end{align*}
    where the inequality follows because  right after iteration $t$, $u^{(\ell)}$ 
    has higher marginal contribution to $\greedy{t}$ than $u'^{(\ell+1)}$, as the former is by definition the $\ell$-th highest marginal contribution after removing $|\neighbors{\greedy{t}}|$, while the latter is by definition the $\ell+1$-th highest marginal contribution after removing $|\neighbors{\greedy{t-1}}|$.
    Overall we have that:
    \begin{align*}
        \E[|\neighbors{\mathcal{Y}^{t}}| \mid \cE \in \mathscr{F}] 
        &= \E[|\neighbors{\greedy{t}}| \mid \cE \in \mathscr{F}] + \E[|\bigcup_{\ell=1}^{k-t}{X_{\ell}}| \mid \cE \in \mathscr{F}] \\
        &= |T| + \E[|\bigcup_{\ell=1}^{k-t}{X_{\ell}}| \mid \cE \in \mathscr{F}] \\
        &\geq |T| + \E[|\bigcup_{\ell=1}^{k-t}{Y_{\ell}}| \mid \cE \in \mathscr{F}] \\
        &= \E[|\neighbors{\greedy{t}}| \mid \cE \in \mathscr{F}] + \E[|\bigcup_{\ell=1}^{k-t}{Y_{\ell}}| \mid \cE \in \mathscr{F}] \\
        &= \E[|\neighbors{\mathcal{Y}^{t-1}}| \mid \cE \in \mathscr{F}],
    \end{align*}
    where the first line is due to Equation (1), the second line is by definition of $\mathscr{F}$, the third line is due to Equation (3), the fourth line is by definition of $\mathscr{F}$ and the fifth line is due to Equation (2).
    Since the claim holds for any $A, T, (D_i)_{i: u_i \in L \setminus A}$, it also holds unconditionally.
\end{proof}

\subsubsection{Structural Lemmas}
We show some structural lemmas that help us prove \cref{lem: greedy lower bound by smax}. In this paragraph, when the bipartite graph is clear from context, we use $\greedydet{t}(E)$ to indicate the output of $\greedyalgorithm$ when it is run on edge set $E$.

\begin{lemma}[Uniform marginal contributions] \label{lem: marginal contributions are uniformly distributed}
    Let $m \geq d_1 \geq \dots \geq d_n \geq 1$. Draw graph $\cB$ from random model $\textsc{GenR}( n, m, d_1, \dots, d_n )$ with edge set $\cE$.
    Also define the following arbitrarily:
    \begin{enumerate}
        \item $A = \{a_1, \dots, a_t\} \subseteq L$,
        \item $T \subseteq R$ such that $\min_i{d_i} \leq |T| \leq \min\{\sum_{i=1}^{t}d_i, m\}$,
        \item $D_i \in [d_i]$ for all $i$ such that $u_i \in L \setminus A$,
        \item $u_x \in L \setminus A$.
    \end{enumerate}
   Let $\mathscr{F} = \{E \subseteq L \times R: \greedydet{t}(E) = A, \neighborsin{E}{A} = T, |\neighborsin{E}{u_i} \setminus T| = D_i \ \forall i \text{ s.t. } u_i \in L \setminus A\}$.
    Then we have that:
    \begin{align*}
    \Pr\left[\neighborsin{\cE}{u_x} \setminus T = X \mid \cE \in \mathscr{F} \right]
    = \Pr\left[\neighborsin{\cE}{u_x} \setminus T = X' \mid \cE \in \mathscr{F} \right],
    \end{align*}
    for any $X, X' \subseteq R \setminus T$ such that $|X| = |X'| = D_x$.
\end{lemma}
\begin{proof}
    Let $E'_1 = \{u_x\} \times X$ and $E'_2 = \{u_x\} \times X'$, then we have:
    \begin{align*}
        \Pr[\neighborsin{\cE}{u_x} \setminus T = X \mid \cE \in \mathscr{F}]
        &= \Pr[\cE_{\{u_x\}} = E'_1 \mid \cE \in \mathscr{F}] \\
        &= \Pr[\cE_{\{u_x\}} = E'_2 \mid \cE \in \mathscr{F}] \\
        &= \Pr[\neighborsin{\cE}{u_x} \setminus T = X' \mid \cE \in \mathscr{F}],
    \end{align*}
    where the first equality is due to $\cE_{\{u_x\}} = \{(u_x, v) \in \cE: v \in R \setminus T\}$ and $E'_1 = \{u_x\} \times X$ so $\cE_{\{x\}} = E'_1$ is equivalent to $\neighborsin{\cE}{x} \setminus T = X$, the second equality is due to \cref{lem: greedy is not affected by irrelevant edges} and the third equality is due to $\cE_{\{u_x\}} = \{(u_x, v) \in \cE: v \in R \setminus T\}$ and $E'_2 = \{u_x\} \times X'$ so $\cE_{\{u_x\}} = E'_2$ is equivalent to $\neighborsin{\cE}{u_x} \setminus T = X'$.
\end{proof}
\begin{lemma}[Independent marginal contributions] \label{lem: marginal contributions are independent across remaining nodes}
    Let $m \geq d_1 \geq \dots \geq d_n \geq 1$. Draw graph $\cB$ from random model $\textsc{GenR}(n, m, d_1, \dots, d_n)$ with edge set $\cE$.
    Also define the following arbitrarily:
    \begin{enumerate}
        \item $A = \{a_1, \dots, a_t\} \subseteq L$,
        \item $T \subseteq R$ such that $\min_i{d_i} \leq |T| \leq \min\{\sum_{i=1}^{t}d_i, m\}$,
        \item $D_i \in [d_i]$ for all $i$ such that $u_i \in L \setminus A$,
        \item $u_x, u_y \in L \setminus A$.
    \end{enumerate}
    Let $\mathscr{F} = \{E \subseteq L \times R: \greedydet{t}(E) = A, \neighborsin{E}{A} = T, |\neighborsin{E}{u_i} \setminus T| = D_i \ \forall i \text{ s.t. } u_i \in L \setminus A\}$.
    Then we have that:
    \begin{align*}
        \Pr[\neighborsin{\cE}{u_x} \setminus T = X, \neighborsin{\cE}{u_y} \setminus T = Y \mid \cE \in \mathscr{F}] 
        = &\Pr[\neighborsin{\cE}{u_x} \setminus T = X \mid \cE \in \mathscr{F}] \\ 
        &\cdot \Pr[\neighborsin{\cE}{u_y} \setminus T = Y \mid \cE \in \mathscr{F}],
    \end{align*}
    for any $X, Y \subseteq R \setminus T$ and $|X| = D_x, |Y| = D_y$.
\end{lemma}
\begin{proof}
    To obtain the statement we apply \cref{lem: sufficient condition for independence} on random sets $\neighborsin{\cE}{u_x} \setminus T, \neighborsin{\cE}{u_y} \setminus T$, so it suffices to show the following properties:
    \begin{enumerate}
        \item $\Pr[\neighborsin{\cE}{u_x} \setminus T = X \mid \cE \in \mathscr{F}] = \Pr[\neighborsin{\cE}{u_x} \setminus T = X' \mid \cE \in \mathscr{F}]$ for any $X, X' \subseteq R \setminus T$ such that $|X| = |X'| = D_x$,
        \item $\Pr[\neighborsin{\cE}{u_x} \setminus T = X, \neighborsin{\cE}{u_y} \setminus = Y \mid \cE \in \mathscr{F}] = \Pr[\neighborsin{\cE}{u_x} \setminus T = X', \neighborsin{\cE}{u_y} \setminus T = Y \mid \cE \in \mathscr{F}]$, for any $X, X', Y \subseteq R \setminus T$ such that $|X| = |X'| = D_x$ and $|Y| = D_y$.
    \end{enumerate}
    Notice that the first property follows directly from \cref{lem: marginal contributions are uniformly distributed}.
    We proceed to show the second property.
    Let $E'_1 = \{u_x\} \times X \cup \{u_y\} \times Y$ and $E'_2 = \{u_x\} \times X' \cup \{u_y\} \times Y$, then we have:
    \begin{align*}
        \Pr[\neighborsin{\cE}{u_x} \setminus T = X, \neighborsin{\cE}{u_y} \setminus T = Y \mid \cE \in \mathscr{F}]
        &= \Pr[\cE_{\{u_x, u_y\}} = E'_1 \mid \cE \in \mathscr{F}] \\
        &= \Pr[\cE_{\{u_x, u_y\}} = E'_2 \mid \cE \in \mathscr{F}] \\
        &= \Pr[\neighborsin{\cE}{u_x} \setminus T = X', \neighborsin{\cE}{u_y} \setminus T = Y \mid \cE \in \mathscr{F}],
    \end{align*}
    where the first equality is due to $\cE_{\{u_x, u_y\}} = \{(u, v) \in \cE: u \in \{u_x, u_y\}, v \in R \setminus T\}$ and $E'_1 = \{u_x\} \times X \cup \{u_y\} \times Y$ so $\cE_{\{u_x, u_y\}} = E'_1$ is equivalent to $\neighborsin{\cE}{u_x} \setminus T = X, \neighborsin{\cE}{u_y} \setminus T = Y$, the second equality is due to \cref{lem: greedy is not affected by irrelevant edges} and the third equality is due to $\cE_{\{u_x, u_y\}} = \{(u, v) \in \cE: u \in \{u_x, u_y\}, v \in R \setminus T\}$ and $E'_2 = \{u_x\} \times X' \cup \{u_y\} \times Y$ so $\cE_{\{u_x, u_y\}} = E'_2$ is equivalent to $\neighborsin{\cE}{u_x} \setminus T = X', \cup \{u_y\} \times Y$.
\end{proof}
\begin{lemma} \label{lem: irrelevant edges are uniform}
    Let $m \geq d_1 \geq \dots \geq d_n \geq 1$ and $\cB \sim \textsc{GenR}(n, m, d_1, \dots, d_n)$ with edge set $\cE$.
    Also define the following arbitrarily:
    \begin{enumerate}
        \item $A = \{a_1, \dots, a_t\} \subseteq L$,
        \item $T \subseteq R$ such that $\min_i{d_i} \leq |T| \leq \min\{\sum_{i=1}^{t}d_i, m\}$,
        \item $D_i \in [d_i]$ for all $i$ such that $u_i \in L \setminus A$,
        \item $Z \subseteq L \setminus A$.
    \end{enumerate}
    Let $\mathscr{F} = \{E \subseteq L \times R: \greedydet{t}(E) = A, \neighborsin{E}{A} = T, |\neighborsin{E}{u_i} \setminus T| = D_i \ \forall i \text{ s.t. } u_i \in L \setminus A\}$.
    Also let $\cE_{Z} = \{(u, v) \in \cE: u \in Z, v \in R \setminus T\}$.
    Then we have that:
    \begin{align*}
        \Pr\left[\cE_{Z} = E'_1 \mid \cE \in \mathscr{F} \right]
        = \Pr\left[\cE_{Z} = E'_2 \mid \cE \in \mathscr{F} \right],
    \end{align*}
    for any $E'_1, E'_2 \subseteq Z \times (R \setminus T)$ such that $|\neighborsin{E'_1}{u_i}| = |\neighborsin{E'_2}{u_i}| = D_i$ for all $i$ such that $u_i \in Z$.     
\end{lemma}
\begin{proof}
    Let $\mathscr{J} = \{E \subseteq L \times R: \neighborsin{E}{A} = T, |\neighborsin{E}{u_i} \setminus T| = D_i \ \forall i \text{ s.t. } u_i \in L \setminus A\} \supseteq \mathscr{F}$.
    Also fix an $E' \in \mathscr{J}$.
    Then we have that:
    \setcounter{equation}{0}
    \begin{align}
        \Pr[&\cE_{Z} = E'_1 \mid \cE \in \mathscr{F}] = \nonumber \\
        & = \Pr\left[\cE_{Z} = E'_1 \mid \greedy{t} = A, \cE \in \mathscr{J} \right] \\
        &= \frac{1}{\Pr\left[\greedy{t} = A \mid \cE \in \mathscr{J} \right]} \cdot \Pr\left[\cE_{Z} = E'_1, \greedy{t} = A \mid \cE \in \mathscr{J} \right] \\
        &= \frac{1}{\Pr\left[\greedy{t} = A \mid \cE \in \mathscr{F} \right]} \cdot 
        \sum_{E \in \mathscr{J}} \Pr\left[\cE_{Z} = E'_1, \greedy{t} = A, \cE = E \mid \cE \in \mathscr{J} \right] \\
        &= \sum_{E \in \mathscr{J}} \frac{\Pr\left[\cE = E \mid \cE \in \mathscr{J}\right]}{\Pr\left[\greedy{t} = A \mid \cE \in \mathscr{J} \right]} \cdot \Pr\left[\cE_{Z} = E'_1, \greedy{t} = A \mid \cE = E \right] \\
        &= \sum_{E \in \mathscr{J}} \frac{\Pr\left[\cE = E\right]}{\Pr\left[\greedy{t} = A, \cE \in \mathscr{J}\right]} \cdot \mathbbm{1}\{E_{Z} = E'_1, \greedydet{t}(E) = A\} \\
        &= \frac{\Pr\left[\cE = E'\right]}{\Pr\left[\greedy{t} = A, \cE \in \mathscr{J} \right]} \cdot \sum_{E \in \mathscr{J}: E_{Z} = E'_1} \mathbbm{1}\{ \greedydet{t}(E) = A\}  \\
        &= \frac{\Pr\left[\cE = E'\right]}{\Pr\left[\greedy{t} = A, \cE \in \mathscr{J} \right]} \cdot \sum_{E \in \mathscr{J}: E_Z = E'_2} \mathbbm{1}\{\greedydet{t}(E) = A\}  \\
        &= \Pr\left[\cE_{Z} = E'_2 \mid \cE \in \mathscr{F} \right],
    \end{align}
    where (1) is by definition of $\mathscr{J}$, (2) is by Bayes' rule, (3) is by law of total probability, (4) is by Bayes' rule and including the constant in the sum, (5) is due to edges being fixed, therefore the runtime of $\greedyalgorithm$ is deterministic and using Bayes' rule on the fraction, (6) is by \cref{lem: all edge sets that respect degrees have the same probability} which implies that the fraction is constant and can be moved out of the sum, (7) is by \cref{lem: greedy is not affected by irrelevant edges} and (8) is by repeating the same steps in reverse.
    To see why (7) is correct, let edge set $E_1 \in \mathscr{J}$ such that $\neighborsin{E_1}{u_i} \setminus T = \neighborsin{E'_1}{u_i}$ for all $u_i \in Z$. Also let $E_2$ be an edge set such that: 
    \begin{align*}
        \neighborsin{E_2}{u_i} = 
        \begin{cases}
            \left(\neighborsin{E_1}{u_i} \setminus \neighborsin{E'_1}{u_i} \right) \cup \neighborsin{E'_2}{u_i} &\text{if } u_i \in Z \\
            \neighborsin{E_1}{u_i} &\text{otherwise }
        \end{cases},
    \end{align*}    
    and notice that $E_2 \in \mathscr{J}$ since $|\neighborsin{E'_1}{u_i}| = |\neighborsin{E'_2}{u_i}| = D_i$ for all $i$ such that $u_i \in Z$.
    Then by \cref{lem: greedy is not affected by irrelevant edges} we have that $\greedydet{t}(E_1) = A$ if and only if $\greedydet{t}(E_2) = A$, so summing on both sides we have:
    \begin{align*}
        \sum_{E_1 \in \mathscr{J}: \neighborsin{E_1}{u} \setminus T = N_1} \mathbbm{1}\{\greedydet{t}(E_1) = A\} 
        = \sum_{E_2 \in \mathscr{J}: \neighborsin{E_2}{u} \setminus T = N_2} \mathbbm{1}\{\greedydet{t}(E_2) = A\},
    \end{align*}
    which suffices by using $E$ instead of $E_1, E_2$ on each sum.
\end{proof}
\begin{lemma} \label{lem: greedy is not affected by irrelevant edges}
    Let $m \geq d_1 \geq \dots \geq d_n \geq 1$. Define the following arbitrarily:
    \begin{enumerate}
        \item $A = \{a_1, \ldots a_k\} \subseteq L$,
        \item $T \subseteq R$,
        \item $E' \subseteq (A \times R) \cup (L \times T)$ such that $\neighborsin{E'}{A} = T$, $|\neighborsin{E'}{u_i}| = d_i$ for all $i$ such that $u_i \in A$ and $|\neighborsin{E'}{u_i}| \leq d_i$ for all $i \in [n]$,
        \item $E_1'', E_2'' \subseteq (L \setminus A) \times (R \setminus T)$ such that $|\neighborsin{E''_1}{u_i}| = |\neighborsin{E''_2}{u_i}| = d_i - |\neighborsin{E'}{u_i} \cap T|$ for all $i$ such that $u_i \in L \setminus A$.
    \end{enumerate}
    Let $B_1 = (L, R, E_1)$ and $B_2 = (L, R, E_2)$ be two bipartite graphs such that 
    $E_1 = E' \cup E''_1$ and $E_2 = E' \cup E''_2$.
    Then, for all $t \in [k]$, $\greedydet{t}(E_1) = \{a_1, \ldots, a_t\}$ if and only if $\greedydet{t}(E_2) = \{a_1, \ldots, a_t\}$.
\end{lemma}
\begin{proof}
    We show by induction that for all $t \in [k]$, if $\greedydet{t}(E_1) = \{a_1, \dots, a_t\}$, then $\greedydet{t}(E_2) = \{a_1, \dots, a_t\}$. The other direction follows by symmetry.
    For the base case, $t = 1$, we have that by construction of $E_1, E_2$, $|\neighborsin{E_1}{u_i}| = |\neighborsin{E_2}{u_i}|$ for all $i \in [n]$, therefore the algorithm works the same in both graphs and $\greedydet{1}(E_1) = \greedydet{1}(E_2) = \{a_1\}$.
    For the inductive step, assume towards a contradiction that $\greedydet{t+1}(E_2) = \{a_1, \dots, a_t, a\}$ where $a \neq a_{t+1}$.
    We consider the two cases in which $a$ would be chosen instead of $a_{t+1}$:
    \begin{enumerate}
        \item The marginal contribution of $a$ was higher than that of $a_{t+1}$, i.e. $|\neighborsin{E_2}{a} \setminus \neighborsin{E_2}{\greedydet{t}(E_2)}| > |\neighborsin{E_2}{a_{t+1}} \setminus \neighborsin{E_2}{\greedydet{t}(E_2)}|$, or
        \item The marginal contribution of $a$ was equal to that of $a_{t+1}$, but $a$ was chosen due to the tie-breaking rule.
    \end{enumerate}
    We proceed to show that the marginal contribution of $a, a_{t+1}$ on iteration $t+1$ is the same regardless of which graph we run $\greedyalgorithm$ on. The marginal contribution of $a$ on $B_2$ is:
    \begin{align*}
        |\neighborsin{E_2}{a} \setminus \neighborsin{E_2}{\greedydet{t}(E_2)}|
        &= |\neighborsin{E_2}{a} \setminus \neighborsin{E_2}{\greedydet{t}(E_1)}| \\
        &= |\neighborsin{E_2}{a} \setminus \neighborsin{E_1}{\greedydet{t}(E_1)}| \\ 
        &= |\neighborsin{E_2}{a}| - |\neighborsin{E_2}{a} \cap \neighborsin{E_1}{\greedydet{t}(E_1)}| \\
        &= |\neighborsin{E_2}{a}| - |\neighborsin{E_1}{a} \cap \neighborsin{E_1}{\greedydet{t}(E_1)}| \\
        &= |\neighborsin{E_1}{a}| - |\neighborsin{E_1}{a} \cap \neighborsin{E_1}{\greedydet{t}(E_1)}| \\ 
        &= |\neighborsin{E_1}{a} \setminus \neighborsin{E_1}{\greedydet{t}(E_1)}|, \tag{1} \label{eq: marginal contribution of a is the same in B_1 and B_2}
    \end{align*}
    where the first equality is by inductive hypothesis $\greedydet{t}(E_1) = \greedydet{t}(E_2) = \{a_1, \dots, a_t\}$, the second is due to $\neighborsin{E_1}{u_i} = \neighborsin{E'}{u_i} = \neighborsin{E_2}{u_i}$ for all $i$ such that $u_i \in A$ therefore $\neighborsin{E_1}{\greedydet{t}(E_1)} = \neighborsin{E_2}{\greedydet{t}(E_1)}$ because $\greedydet{t}(E_1) \subseteq A$, the third is by refactoring, 
    the fourth is due to $\neighborsin{E_2}{a} \cap T = \{v \in T: (a,v) \in E_2\} = \{v \in T: (a, v) \in E'\} = \{v \in T: (a,v) \in E_1\} = \neighborsin{E_1}{a} \cap T$ and $\neighborsin{E_1}{\greedydet{t}(E_1)} \subseteq T$, the fifth is due to $|\neighborsin{E_1}{u_i}| = |\neighborsin{E_2}{u_i}|$ for all $i \in [n]$ by construction and the sixth by refactoring.
    The marginal contribution of $a_{t+1}$ on $B_2$ is:
    \begin{align*}
        |\neighborsin{E_2}{a_{t+1}} \setminus \neighborsin{E_2}{\greedydet{t}(E_2)}| 
        &= |\neighborsin{E_2}{a_{t+1}} \setminus \neighborsin{E_2}{\greedydet{t}(E_1)}| \\
        &= |\neighborsin{E_2}{a_{t+1}} \setminus \neighborsin{E_1}{\greedydet{t}(E_1)}| \\
        &= |\neighborsin{E_1}{a_{t+1}} \setminus \neighborsin{E_1}{\greedydet{t}(E_1)}|,  \tag{2} \label{eq: marginal contribution of at+1 is the same on B_1 and B_2}
    \end{align*}
    where the first equality is by inductive hypothesis $\greedydet{t}(E_1) = \greedydet{t}(E_2) = \{a_1, \dots, a_t\}$, the second is due to $\neighborsin{E_1}{u_i} = \neighborsin{E'}{u_i} = \neighborsin{E_2}{u_i}$ for all $i$ such that $u_i \in A$ therefore $\neighborsin{E_1}{\greedydet{t}(E_1)} = \neighborsin{E_2}{\greedydet{t}(E_1)}$ because $\greedydet{t}(E_1) \subseteq A$, the third is due to $\neighborsin{E_2}{a_{t+1}} \cap T = \{v \in T: (a_{t+1},v) \in E_2\} = \{v \in T: (a_{t+1}, v) \in E'\} = \{v \in T: (a_{t+1},v) \in E_1\} = \neighborsin{E_1}{a_{t+1}} \cap T$ and $\neighborsin{E_1}{\greedydet{t}(E_1)} \subseteq T$.
    We can proceed by considering the two cases.
    Assume case 1, then by \cref{eq: marginal contribution of a is the same in B_1 and B_2,eq: marginal contribution of at+1 is the same on B_1 and B_2} we have that:
    \begin{align*}
        |\neighborsin{E_1}{a} \setminus \neighborsin{E_1}{\greedydet{t}(E_1)}| > |\neighborsin{E_1}{a_{t+1}} \setminus \neighborsin{E_1}{\greedydet{t}(E_1)}|,
    \end{align*}
    which is a contradiction since $\greedydet{t+1}(E_1) = \{a_1, \dots, a_t, a_{t+1}\}$ and $\greedydet{t}(E_1) = \{a_1, \dots, a_t\}$ so $\greedydet{t+1}(E_1) \setminus \greedydet{t}(E_1) = \{a_{t+1}\}$.
    For case 2, since the tie-breaking rule relies on the order of the nodes of $L$ and it is the same in both $B_1, B_2$, $a$ could not be chosen instead of $a_{t+1}$.
    Therefore $\greedydet{t+1}(E_2) = \{a_1, \dots, a_{t+1}\}$, which concludes the argument.
\end{proof}
\begin{lemma} \label{lem: all edge sets that respect degrees have the same probability}
    Let Let $m \geq d_1 \geq \dots \geq d_n \geq 1$ and $\cB \sim \textsc{GenR}(n, m, d_1, \dots, d_n)$ with edge set $\cE$.
    Let $B_1 = (L, R, E_1)$ and $B_2 = (L, R, E_2)$ be two bipartite graphs such that 
    $|\neighborsin{E_1}{u_i}| = |\neighborsin{E_2}{u_i}| = d_i$ for all $u \in L$.
    Then $\Pr[\cE = E_1] = \Pr[\cE = E_2]$.
\end{lemma}
\begin{proof}
    We proceed to show that $\Pr[\cE = E_j] = \prod_{i \in [n]} \binom{m}{d_i}^{-1}$ for any $j \in \{1, 2\}$:
    \begin{align*}
        \Pr[\cE = E_i]
        = \Pr[\neighborsin{\cE}{u_i} = \neighborsin{E_j}{u_i} \ \forall i \in [n]] 
        = \prod_{i \in [n]}{\Pr[\neighborsin{\cE}{u_i} = \neighborsin{E_j}{u_i}]} 
        = \prod_{i \in [n]}{\binom{m}{d_i}^{-1}}, 
    \end{align*}
    where the first equality is due to $\cB$ being bipartite, the second is due to the nodes of $L$ selecting neighbors independently and the third is due to each node $u_i$ selecting $d_i$ neighbors uniformly at random from $R$.
\end{proof}
\begin{lemma} \label{lem: coverage of s dominated by smax}
    Let Let $m \geq d_1 \geq \dots \geq d_n \geq 1$ and $\cB \sim \textsc{GenR}(n, m, d_1, \dots, d_n)$ with edge set $\cE$.
    Let $S = \{ w_1, w_2, \dots, w_k \} \subseteq L$, then $|\neighborsin{\cE}{S}| \preccurlyeq |\neighborsin{\cE}{\smax{k}}|$.
\end{lemma}
\begin{proof}
    Without loss of generality assume that $w_i$ are in decreasing order of degree $|\neighbors{w_i}|$.
    Define the following random sets:
    \begin{align*}
        \cX = \neighborsin{\cE}{S} = \bigcup_{i}\neighborsin{\cE}{w_i},\
        \cY = \neighborsin{\cE}{\smax{k}} = \bigcup_{i}\neighborsin{\cE}{u_i},
    \end{align*} 
    for which we will construct a coupling $(\hat{\cX}, \hat{\cY})$ such that $|\hat{\cY}| \geq |\hat{\cX}|$ almost surely.
    Let $\cU_i \subseteq R \setminus \neighborsin{\cE}{w_i}$ be a uniformly random sample of size $|\neighborsin{\cE}{u_i}| - |\neighborsin{\cE}{w_i}|$ and independent of all other $\neighborsin{\cE}{w_j}, \neighborsin{\cE}{u_j}, \cU_j$ and define:
    \begin{align*}
        \hat{\cX} = \cX,\ \hat{\cY} = \bigcup_{i}\left(\neighborsin{\cE}{w_i} \cup \cU_i \right),
    \end{align*}
    for which it is easy to see that $|\hat{\cY}| \geq |\hat{\cX}|$ since by definition of $\smax{k}$ we have that $|\neighborsin{\cE}{u_i}| \geq |\neighborsin{\cE}{w_i}|$.
    What remains to show is that $(\hat{\cX}, \hat{\cY})$ is a coupling and it suffices to show that $\hat{\cY}$ is identically distributed to $\cY$. 
    To that end we show the following three properties:
    \begin{enumerate}
        \item $\neighborsin{\cE}{w_i} \cup \cU_i \ind \neighborsin{\cE}{w_j} \cup \cU_j$ for all $i \neq j \in [k]$,
        \item $|\neighborsin{\cE}{w_i} \cup \cU_i| = |\neighborsin{\cE}{u_i}|$,
        \item $\neighborsin{\cE}{w_i} \cup \cU_i$ is uniformly sampled from $R$.
    \end{enumerate}
    The first property follows by definition since $\neighborsin{\cE}{w_i} \ind \neighborsin{\cE}{w_j}$ and $\cU_i \ind \neighborsin{\cE}{w_j}, \cU_j$ since $i \neq j$.
    The second property also follows by definition since $\cU_i \in R \setminus \neighborsin{\cE}{w_i}$, therefore $|\neighborsin{\cE}{w_i} \cup \cU_i| = |\neighborsin{\cE}{w_i}| + |\cU_i| = |\neighborsin{\cE}{u_i}|$.
    The third property follows from \cref{lem: union of non overlapping uniform sets is uniform}, which concludes the proof.
\end{proof}

\subsubsection{Probability Lemmas}
\begin{lemma} \label{lem: union of non overlapping uniform sets is uniform}
    Let $\cA \subseteq R$ be uniformly sampled such that $|\cA| = a$ and $\cB \subseteq R \setminus \cA$ be uniformly sampled such that $|\cB| = b$. 
    Then $\cA \cup \cB$ is a uniform $(a+b)$-subset in $R$.
\end{lemma}
\begin{proof}
    Let $T \subseteq R$ such that $|T| = a + b$ and we have:
    \begin{align*}
        \Pr[\cA \cup \cB = T] 
        &= \sum_{N \subseteq R: |N| = a}\Pr[N \cup \cB = T \mid \cA = N] \cdot \Pr[\cA = N] \\
        &= \sum_{N \subseteq R: |N| = a, |T \setminus N| = b}\Pr[\cB = T \setminus N \mid \cA = N] \cdot \Pr[\cA = N] \\
        &= \sum_{N \subseteq R: |N| = a, |T \setminus N| = b}\frac{1}{ \binom{|R| - a}{b}} \cdot \frac{1}{\binom{|R|}{a}} \\
        &= \frac{\binom{b+a}{a}}{\binom{|R|}{a} \cdot \binom{|R| - a}{b}} \\ 
        &= \frac{\frac{(b+a)!}{a! b!}}{\frac{|R|!}{a! (|R| - a)!} \cdot \frac{(|R| - a)!}{b! (|R| -a - b)!}} \\
        &= \frac{1}{\binom{|R|}{a + b}},
    \end{align*}
    where the first equality is due to law of total probability, the second is by definition of $\cB$ and the condition $\cA = N$ which implies that $\cB$ is uniformly sampled from $R \setminus N$, the third is by definition of $\cA, \cB$ as uniformly sampled, the fourth is by counting all the different ways $N$ can be selected, the fifth is by definition of binomial coefficients and the sixth is by canceling terms in the fraction.
\end{proof}
\begin{lemma} \label{lem: uniform A1 A2 dominate B1 B2 then their union dominates as well}
    Let $\cA_1, \dots, \cA_{n}$ and $\cB_1, \dots, \cB_{n}$ be random subsets of $R$ such that:
    \begin{enumerate}
        \item $\cA_i \ind \cA_j$, $\cB_i \ind \cB_j$ for all $i \neq j \in [n]$,
        \item $\cA_i, \cB_i$ uniform in $R$ for all $i \in [n]$,
        \item $\E[|\cA_i|] \geq \E[|\cB_i|]$ for all $i \in [n]$.
    \end{enumerate}
    Then $\E[|\bigcup_{i=1}^{n}\cA_i|] \geq \E[|\bigcup_{i=1}^{n}\cB_i|]$.
\end{lemma}
\begin{proof}
    We start by showing that Assumptions (2), (3) imply higher probability of inclusion for $\cA_i$ than $\cB_i$. 
    We have that for $\cX_i \in \{\cA_i, \cB_i\}$ and for any $r \in R$:
    \begin{align*}
        \Pr[r \in \cX_i] 
        &= \sum_{x}{\Pr[r \in \cX_i \mid |\cX_i| = x] \cdot \Pr[|\cX_i| = x]} \\
        &= \sum_{x}{\frac{x}{|R|} \cdot \Pr[|\cX_i| = x]} \\
        &= \frac{1}{|R|} \cdot \sum_{x}{x \cdot \Pr[|\cX_i| = x]} \\ 
        &= \frac{1}{|R|} \cdot \E[|\cX_i|] \tag{*},
    \end{align*}
    where the first equality is due to the law of total probability, the second is due to Assumption (2) since $\cX_i$ is sampling $x$ elements of $R$ without replacement, the third is by moving $|R|$ out of the sum, and the fourth is by the definition of expectation.
    So we can show that:
    \begin{align*}
        \Pr[r \in \cA_i] 
        = \frac{1}{|R|} \cdot \E[|\cA_i|] 
        \geq \frac{1}{|R|} \cdot \E[|\cB_i|] 
        = \Pr[r \in \cB_i] \tag{2},
    \end{align*}
    where the equalities are due to Equation (1) and the inequality is due to Assumption (3).
    We can now show that the probability of inclusion in $\bigcup_{i}{\cA_i}$ is higher than that of $\bigcup_{i}{\cB_i}$ due to Assumption (1):
    \begin{align*}
        \Pr[r \in \bigcup_{i}{\cA_i}]
        &= 1 - \Pr[x \notin \cA_i \ \forall i] \\
        &= 1 - \prod_{i}{\left( 1 - \Pr[r \in \cA_i] \right)} \\
        &\geq 1 - \prod_{i}{\left( 1 - \Pr[r \in \cB_i] \right)} \\
        &= \Pr[r \in \bigcup_{i}{\cB_i}], \tag{3}
    \end{align*}
    where the second and third equality is due to independence by Assumption (1) and  the inequality is due to Equation (2).
    Overall we have that:
    \begin{align*}
        \E[|\bigcup_{i}{\cA_i}|] 
        = \E[\sum_{r \in R}{\mathbbm{1}\{r \in \bigcup_{i}{\cA_i}}]
        = \sum_{r \in R}{\Pr[r \in \bigcup_{i}{\cA_i}]}
        \geq \sum_{r \in R}{\Pr[r \in \bigcup_{i}{\cB_i}]}
        = \E[|\bigcup_{i}{\cB_i}],
    \end{align*}
    where the inequality is due to Equation (3).
\end{proof}
\begin{lemma} \label{lem: sufficient condition for independence}
    Let $\cA, \cB$ be random variables with support $A, B$ respectively that have the following properties:
    \begin{enumerate}
        \item $\Pr[\cA = a] = \Pr[\cA = a']$ for all $a, a' \in A$ and
        \item $\Pr[\cA=a, \cB=b] = \Pr[\cA=a', \cB=b]$ for all $a, a' \in A$ and all $b \in B$.
    \end{enumerate}
    Then $\cA, \cB$ are independent.
\end{lemma}
\begin{proof}
    We proceed to show that $\Pr[\cA = a] \cdot \Pr[\cB = b] = \Pr[\cA = a, \cB = b]$ for all $a \in A, b \in B$:
    \begin{align*}
        \Pr[\cA = a] \cdot \Pr[\cB = b] 
        &= \Pr[\cA = a] \cdot \sum_{a' \in A}\Pr[\cA = a', \cB = b] \\
        &= \Pr[\cA = a] \cdot |A| \cdot \Pr[\cA = a, \cB = b] \\
        &= \frac{1}{|A|} \cdot |A| \cdot \Pr[\cA = a, \cB = b] \\
        &= \Pr[\cA = a, \cB = b],
    \end{align*}
    where the first equality is due to law of total probability, the second is due to property (2) and the third is due to property (1) which implies that $\Pr[\cA = a] = \frac{1}{|A|}$.
\end{proof}

\subsection{Lower Bound By Concentration Of \texorpdfstring{$\optalgorithm$}{Opt}}
\begin{lemma}
\label{lem:greedy lower bound by concentration}
    Let $m \geq d_1 \geq \dots \geq d_n \geq 1$ and $\mathcal{B} \sim \textsc{GenR}(n, m, d_1, \dots, d_n)$.
    For any $\epsilon \in (0,1)$, we have that 
    $\E\left[ |\neighbors{\greedy{k}}| \right] \geq (1 - \epsilon) \cdot \E\left[ |\neighbors{\opt{k}}| \right] - m \cdot \exp\left(\log \binom{n}{k} - \frac{\epsilon^2}{3} \left( 1 - e^{- \sum_{i=1}^{k}d_i / m} \right) \cdot m \right) $.
\end{lemma}
\begin{proof}
    Let $\cX = \{ |\neighbors{\opt{k}}| \leq (1+\epsilon) \cdot \E\left[ |\neighbors{\smax{k}}| \right] \}$, then: 
    \begin{align*}
         \E\left[ |\neighbors{\opt{k}}| \right]
         &= \E\left[ |\neighbors{\opt{k}}| \middle| \cX  \right] \cdot \Pr\left[ \cX \right] + \E\left[ |\neighbors{\opt{k}}| \middle| \bar{\cX}  \right] \cdot \Pr\left[ \bar{\cX} \right] \\
         &\leq \E\left[ |\neighbors{\opt{k}}| \middle| \cX \right] \cdot 1 + m \cdot \Pr\left[ \bar{\cX} \right] \\
         &= \E\left[ |\neighbors{\opt{k}}| \middle| |\neighbors{\opt{k}}| \leq (1+\epsilon) \cdot \E\left[ |\neighbors{\smax{k}}| \right] \right] + m \cdot \Pr\left[ |\neighbors{\opt{k}}| > (1+\epsilon) \cdot \E\left[ |\neighbors{\smax{k}}| \right] \right] \\ 
         &\leq (1+\epsilon) \cdot \E\left[ |\neighbors{\smax{k}}| \right] + m \cdot \exp\left( \log{\binom{n}{k}} - \frac{\epsilon^2}{3} \E\left[ |\neighbors{\smax{k}}| \right] \right) \\
         &\leq (1+\epsilon) \cdot \E\left[ |\neighbors{\greedy{k}}| \right] + m \cdot \exp\left( \log{\binom{n}{k}} - \frac{\epsilon^2}{3} \left( 1 - e^{- \sum_{i=1}^{k}d_i / m} \right) \cdot m \right),
    \end{align*}
    where the first line is by tower rule, the second line is due to $|\neighbors{\opt{k}}| \leq m$ by \cref{lem: trivial upper bound for opt}, the third line is by definition of $\cX$, the fourth line is by \cref{lem: opt upper bound by concentration} and the fifth line is due to $\E\left[ |\neighbors{\smax{k}}| \right] \leq \E\left[ |\neighbors{\greedy{k}}| \right]$ by \cref{lem: greedy lower bound by smax} and $\E\left[ |\neighbors{\smax{k}}| \right] \geq \left( 1 - e^{-\sum_{i=1}^{k} \frac{d_i}{m}} \right) \cdot m$ by \cref{lem:smax expected coverage}.
    Solving for $\E\left[ |\neighbors{\greedy{k}}| \right]$ we have that:
    \begin{align*}
        \E\left[ |\neighbors{\greedy{k}}| \right] 
        &\geq \frac{1}{1 + \epsilon} \cdot \E\left[ |\neighbors{\opt{k}}| \right] - \frac{1}{1+\epsilon} m \cdot \exp\left(\log{\binom{n}{k}} - \frac{\epsilon^2}{3} \left( 1 - e^{- \sum_{i=1}^{k}d_i / m} \right) \cdot m \right) \\
        &\geq (1 - \epsilon) \cdot \E\left[ |\neighbors{\opt{k}}| \right] - m \cdot \exp\left(\log{\binom{n}{k}} - \frac{\epsilon^2}{3} \left( 1 - e^{- \sum_{i=1}^{k}d_i / m} \right) \cdot m \right),
    \end{align*}
    where the second line is due to $1 - \epsilon \leq \frac{1}{1+\epsilon} \leq 1$ since $\epsilon \in (0,1)$.
\end{proof}

\begin{lemma} \label{lem: opt upper bound by concentration}
    Let $n, m \in \mathbbm{N}$, $m \geq d_1 \geq \dots \geq d_n \geq 1, k \in [n]$, graph $\mathcal{B} \sim \textsc{GenR}(n, m, d_1, \dots, d_n)$.
    For any $\delta \in (0,1)$, $n, m, k$ we have  that $\f{\opt{k}} > (1+\delta) \E\left[ |\neighbors{\smax{k}}| \right]$ with probability at most $\exp\left( \log{\binom{n}{k}} - \frac{\delta^2}{3} \E\left[ |\neighbors{\smax{k}}| \right] \right)$.
\end{lemma}
\begin{proof}
    To upper bound the optimal solution, first notice that $|\neighbors{\opt{k}}| = \max_{S \in \binom{L}{k}}|\neighbors{S}|$, then we have:
    \begin{align*}
        \Pr\left[ |\neighbors{\opt{k}}| > (1 + \delta) \cdot \E\left[ |\neighbors{\smax{k}}| \right] \right]
        &= \Pr\left[ \max_{S \in \binom{L}{k}}|\neighbors{S}| > (1 + \delta) \cdot \E\left[ |\neighbors{\smax{k}}| \right] \right] \\
        &\leq \sum_{S \in \binom{L}{k}} \Pr\left[|\neighbors{S}| > (1 + \delta) \cdot \E\left[ |\neighbors{\smax{k}}| \right] \right] \\
        &\leq \sum_{S \in \binom{L}{k}} \Pr\left[|\neighbors{\smax{k}}| > (1 + \delta) \cdot \E\left[ |\neighbors{\smax{k}}| \right] \right] \\
        &\leq \sum_{S \in \binom{L}{k}} \exp\left( -\frac{\delta^2}{3} \cdot \E\left[ |\neighbors{\smax{k}}| \right] \right) \\
        &= \exp\left( \log{\binom{n}{k}} - \frac{\delta^2}{3} \E\left[ |\neighbors{\smax{k}}| \right]  \right), \tag{2}
    \end{align*}
    where the first line is by definition of $\opt{k}$, the second line is by union bound, the third line is because $|\neighbors{S}|$ is dominated by $|\neighbors{\smax{k}}|$ for all $S \in \binom{L}{k}$ by \cref{lem: coverage of s dominated by smax}, the fourth line is by \cref{lem: coverage of smax admits chernoff} since $\delta \in (0, 1)$ by Equation (1) and the fifth line is because $|\binom{L}{k}| = \binom{n}{k} = e^{\log{\binom{n}{k}}}$.
\end{proof}

\begin{lemma} \label{lem: trivial upper bound for opt}
    Let $n, m \in \mathbbm{N}$, $m \geq d_1 \geq \dots \geq d_n \geq 1$ and $|\neighbors{u_i}| = d_i$ for all $i \in [n]$. Then $\f{\opt{k}} \leq \min\{ \sum_{i=1}^{k}{d_i}, \f{L}, m \}$ for any $n,m,k,d$.
\end{lemma}
\begin{proof}
    Since the capacity constraint is $k$, $\f{\opt{k}} \leq \sum_{i=1}^{k}d_i$.
    Also, since $\f{L}$ is the total value of all the left nodes, $\f{\opt{k}} \leq \f{L}$ and $\f{L} \leq m$.
\end{proof}

\subsubsection{Concentration of \texorpdfstring{$\smax{k}$}{Hk}}
\label{sec: negative correlation}

We define a property that allows us to use Chernoff bounds for sums of variables that are not independent.
\begin{definition}[Section 1.10.2.2 in \cite{DOERR18}]
\label{negative_correlation}
    Random variables $X_j$ are negatively correlated if for any choice of $I \subseteq [m]$ we have
    \begin{align*}
        \Pr[\forall j \in I: X_j = 1] &\leq \prod_{j \in I}{\Pr[X_j = 1]} \quad \text{(1-negative correlation)} \\
        \Pr[\forall j \in I: X_j = 0] &\leq \prod_{j \in I}{\Pr[X_j = 0]} \quad \text{(0-negative correlation)}.
    \end{align*}
\end{definition}
For such variables, the usual Chernoff bounds apply.
\begin{lemma}[Theorem 1.10.23 (a) in \cite{DOERR18}] 
\label{lem: upper tail bounds for negatively correlated variables}
    Let $X = \sum_j{X_j}$ for $X_1, \dots, X_m$ identically distributed negatively correlated random variables, such that $X_j \in \{0, 1\}$ for all $j \in [m]$.
    Then $Y = \sum_{i=1}^{m} X_i$ satisfies the usual Chernoff bound:
    \begin{align*}
        \Pr[Y \geq (1 + \delta) \cdot \E[Y]] \leq \exp\left( - \frac{\delta^2}{3} \E[Y] \right).
    \end{align*}
\end{lemma}
The following is a useful lemma for showing that variables are negatively correlated. 
\begin{lemma}[Lemma 1.10.26 of \cite{DOERR18}]
\label{lem: doerr lemma for na variables}
    Let $k, N \in \mathbb{N}$. For all $i \in [k]$, let $n_i \in [N]$.  
    Let $S$ be some set of cardinality $N$; for convenience, $S = [N]$.  
    For all $i \in [k]$, let $U_i$ be a subset of $S$ having cardinality $n_i$, uniformly chosen among all such subsets.  
    Let the $U_i$ be stochastically independent.  
    For all $j \in S$, let $X_j$ be the indicator random variable for the event that $j \in U_i$ for some $i \in [k]$.  
    Then the random variables $X_1, \dots, X_N$ are negatively correlated.
\end{lemma}

Finally, we put everything together and show that $|\neighbors{\smax{k}}|$ admits Chernoff bounds.
\begin{lemma}
\label{lem: coverage of smax admits chernoff}
    Let $n, m \in \mathbbm{N}$, $m \geq d_1 \geq \dots \geq d_n \geq 1$, $k \in [n]$ and we draw bipartite graph $\cB = (L, R, \cE)$ from $\textsc{GenR}(n, m, d_1, \dots, d_n)$.
    Then we have that for all $\delta \in (0,1)$:
    \begin{align*}
        \Pr\left[ |\neighbors{\smax{k}}| > (1 + \delta) \cdot \E\left[ |\neighbors{\smax{k}}| \right] \right] \leq \exp\left( - \frac{\delta^2}{3} \E\left[ |\neighbors{\smax{k}}| \right] \right).
    \end{align*}
\end{lemma}
\begin{proof}
    Let $U_i = \neighbors{u_i}$ for all $i \in [n]$, which are independent and uniform in $\binom{R}{d_i}$ by definition of the random model.
    Also let random variables $X_j  = \mathbbm{1}\{ v_j \in \neighbors{\smax{k}} \} = \mathbbm{1}\{ v_j \in \cup_{i=1}^{k}\neighbors{u_i} \}$ for $j \in [m]$. 
    Then for $S = R$, $n_i = d_i$, $N = m$, by \cref{lem: doerr lemma for na variables} we have that $X_1, \dots, X_m$ are negatively correlated. 
    Then the statement follows by \cref{lem: upper tail bounds for negatively correlated variables} since $|\neighbors{\smax{k}}| = \sum_{j=1}^{m}X_j$.
\end{proof}

\section{Omitted Proofs Of Section 3}
\label{sec: omitted proofs of section 3}

\subsection{Technical Tools}
\label{sec: technical tools n=m appendix}

\paragraph{The differential equation method.}
\FirstPhaseAcceptsN*
\begin{proof}
    Due to the equivalence between $\acceptrejectalgorithm$ and $\greedyalgorithm$ \cref{lem: greedy is equivalent to acceptreject}, on any graph $B = (L, R, E)$ we have $t_d = |\acceptreject{d}{n}|$, since at the end of phase $d$, $\acceptrejectalgorithm$ has accepted all the nodes with marginal contribution $d$. Therefore, it suffices to estimate $|\acceptreject{d}{n}|$. Towards this end we define $Y(i) = |\acceptreject{d}{i}|$, i.e. the number of accepted nodes on iteration $i$ of phase $d$ and we have:
    \begin{align*}
        \E\left[ Y(i+1) - Y(i) \middle| Y(i) \right] 
        = \E\left[ \mathbbm{1}\{u_{i+1} \text{ accepted }\} \middle| Y(i) \right]
        = \Pr\left[ u_{i+1} \text{ accepted } \middle| Y(i) \right]
        = \frac{\binom{n - Y(i) d}{d}}{\binom{n}{d}},
    \end{align*}
    where the last equality is because $u_{i+1}$ is accepted only if $|\neighbors{u_{i+1}} \setminus \neighbors{\acceptreject{d}{i}}| = d$. But since $|\neighbors{\acceptreject{d}{i}}| = Y(i) \cdot d$ and $\neighbors{u_{i+1}}$ are selected uniformly at random from the $n$ nodes of $R$, the only way $u_{i+1}$ was accepted was if it sampled $d$ nodes from the $n - Y(i) \cdot d$ remaining.
    We can now show one by one that the assumptions of the differential equation hold:
    \begin{enumerate}
        \item Trend Hypothesis: $\bigabs{ \E\left[ Y(i+1) - Y(i) \middle| Y(i) \right] - \left( 1 - d \frac{Y(i)}{n} \right)^d } \leq \frac{2d^2}{n}$ by \cref{lem: approximation of success probability n=m} for $r = n - Y(i) d$ since $n \geq 2d^2$,
        \item Lipschitz Hypothesis: $F(y) = \bigpar{1 - d y}^d$ is $d^2$-Lipschitz since: 
        \begin{align*}
            \bigabs{F(y_1) - F(y_2)} 
            &= \bigabs{ \bigpar{1 - d y_1}^d - \bigpar{1 - d y_2}^d } \\
            &= d \cdot \bigabs{ \bigpar{\bigpar{1 - d y_1}^{d-1} + \dots + \bigpar{1 - d y_2}^{d-1}}} \cdot \bigabs{ y_2 - y_1 } \\
            &\leq d \cdot \bigabs{ 1 + \dots + 1 } \cdot \bigabs{ y_2 - y_1 } \\
            &= d^2 \cdot \bigabs{ y_1 - y_2 },
        \end{align*}
        \item Boundedness Hypothesis: $\bigabs{Y(i+1) - Y(i)} \leq 1$ as only one accept can occur,
        \item Initial Condition: $Y(0) = \hat{y} = 0$ so any $\lambda \geq \frac{2}{n} = \frac{d^2}{n} \min\{1, \frac{1}{d^2} \} + \frac{1}{n} = \frac{d^2}{n} \min\{1, \frac{1}{d^2} \} + \frac{1}{n}$ works. We choose $\lambda = \sqrt{\frac{8\log{(2n)}}{n}}$.
        
    \end{enumerate}
    
    Model $Y(i)$ by $n y\left( \frac{i}{n} \right)$ and we obtain the following differential equation:
    \begin{align*}
        y'(t) = \left( 1 - d \cdot y(t) \right)^d,\ y(0) = 0 
        \Rightarrow y(t) = \frac{1}{d} \cdot \left( 1 - \left( 1 + d(d-1) \cdot t \right)^{- \frac{1}{d-1}} \right).
    \end{align*}
    By \cref{thm:differential equation method} with probability $1 - 2 e^{-n \lambda^2 / 8 } = 1 - \frac{1}{n}$ we have that:
    \begin{align*}
        \bigabs{Y(i) - n y(i/n)} < 3 e^{d^2} \cdot \lambda n = 3 e^{d^2} \cdot \sqrt{8n \log(2n)},
    \end{align*}
    then for $i = n$ we obtain the statement.
\end{proof}

\begin{lemma} \label{lem: approximation of success probability n=m}
    Let $d \in \mathbbm{N}$, $n \geq 2d^2$, $r \in [n]$
    then $ |\binom{r}{d} / \binom{n}{d} -  \left( \frac{r}{n} \right)^d | \leq \frac{2d^2}{n}$.
\end{lemma}
\begin{proof}
    We start by showing the upper bound:
    \begin{align*}
        \frac{\binom{r}{d}}{\binom{n}{d}}
        = \frac{r (r - 1) \cdots (r - d + 1))}{n (n - 1) \cdots (n - d + 1)}
        = \prod_{i=0}^{d-1} \left( \frac{r - i}{n - i} \right)
        \leq \left( \frac{r}{n} \right)^d.
    \end{align*}
    We proceed to show the lower bound:
    \begin{align*}
        \frac{\binom{r}{d}}{\binom{n}{d}}
        = \prod_{i=0}^{d-1} \left( \frac{r - i}{n - i} \right)
        = \prod_{i=0}^{d-1} \left( \frac{r}{n} \cdot \frac{1 - \frac{i}{r}}{1 - \frac{i}{n}} \right)
        &= \left( \frac{r}{n} \right)^d \cdot \prod_{i=0}^{d-1} \left( 1 - \frac{i(n-r)}{r (n-i)} \right) \\
        &\geq \left( \frac{r}{n} \right)^d \cdot \left( 1 - \sum_{i=0}^{d-1} \frac{i(n-r)}{r (n-i)} \right) \\
        &\geq \left( \frac{r}{n} \right)^d \cdot \left( 1 - \frac{n-r}{r(n-d)} \cdot \sum_{i=0}^{d-1}i \right) \\
        &= \left( \frac{r}{n} \right)^d \cdot \left( 1 - \frac{n-r}{r(n-d)} \cdot \frac{d(d-1)}{2} \right) \\
        &= \left( \frac{r}{n} \right)^d - \left( \frac{r}{n} \right)^{d-1} \cdot \frac{n-r}{n-d} \cdot \frac{d^2}{2n} \\
        &\geq \left( \frac{r}{n} \right)^d - \frac{n}{n-d} \cdot \frac{d^2}{2n} \\
        &\geq \left( \frac{r}{n} \right)^d - \frac{2d^2}{n},
    \end{align*}
    where the first inequality is due to $\prod_{i=0}^{d-1}(1 - x_i) \geq 1 - \sum_{i=0}^{d-1} x_i$ for all $x_i \in [0,1]$, the second inequality is due to $n-i \geq n-d$ for all $i \in [d-1]$, the third inequality is due to $\frac{r}{n} \leq 1$ and the fourth inequality is because $d \leq \sqrt{n}$, so $\frac{n}{n-d} \leq \frac{n}{n - \sqrt{n}} = \frac{1}{1 - 1/\sqrt{n}} \leq 4$ for all $n \geq 2d^2 \geq 2$.
\end{proof}

\paragraph{The fixed set lower bound.}
\SmaxExpectedCoverageN*
\begin{proof}
    Notice that $\smax{k}$ covers node $v_j$ as long as at least one of its nodes covers $v_j$, so for all $j \in [n]$: 
    \begin{align*}
        \Pr\left[ v_j \in \neighbors{\smax{k}} \right] 
        = 1 - \prod_{u_i \in \smax{k}}{ \left( 1 - \Pr\left[ v_j \in \neighbors{u_i} \right] \right) }
        = 1 - \prod_{u_i \in \smax{k}}{\left(1 - \frac{d}{n}\right)}
        = 1 - \left( 1 - \frac{d}{n} \right)^{k}. \tag{*}
    \end{align*}
    We have that:
    \begin{align*}
        \E\left[ |\neighbors{\smax{k}}| \right] 
        = \E\left[ \sum_{j \in [n]}{\mathbbm{1}\{ v_j \in \neighbors{\smax{k}} \} } \right] 
        = \sum_{j \in [n]}{\Pr[v_j \in \neighbors{\smax{k}}]}
        = \left( 1 - \left( 1 - \frac{d}{n} \right)^{k} \right) \cdot n,
    \end{align*}    
    where the first equality is by definition, the second equality is by linearity of expectation and the third equality is due to Equation (*). The statement follows because $1 - \frac{d}{n} \leq e^{- \frac{d}{n}}$.
\end{proof}
The following lemma is implied by the more general \cref{lem: greedy lower bound by smax}.
\GreedyLowerBoundSmaxN*
As already mentioned, \cref{lem:greedy lower bound by smax for n=m} cannot be proved by simply comparing the marginal contributions and showing that $\E\left[ |\neighbors{\greedy{t}} \setminus \neighbors{\greedy{t-1}}| \right] \geq \E\left[ |\neighbors{\smax{t}} \setminus \neighbors{\smax{t-1}}| \right]$ for all iterations $[t]$. 
By plotting the average marginal contribution of $\greedyalgorithm$ against those of the fixed set, we show that there is an iteration $t \in [n]$, beyond which, the marginal contribution of $\greedyalgorithm$ is lower than that of the fixed set.
\begin{figure}[H] 
    \centering
    \includegraphics[width=0.7\textwidth]{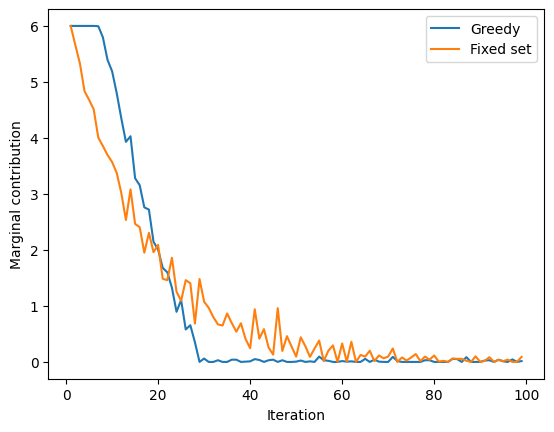}
    \caption{The computation is carried out for $n = m = 100, d=6$ and the marginal contributions of $\greedyalgorithm$ and the fixed set are calculated by averaging over $200$ runs.}
\label{fig:marginal of greedy vs smax}
\end{figure}

\subsection{Analysis For The Different Regions}
\label{sec: analysis for the different regions n=m appendix}

\CriticalRegionLargeDegreesND*
\begin{proof}
    To upper bound the optimal solution, first notice that $|\neighbors{\opt{k}}| = \max_{S \in \binom{L}{k}}|\neighbors{S}|$.
    Let $\delta = \sqrt{\frac{24}{\epsilon} \left(\frac{1}{\epsilon d} + \frac{\max(0, \log{\epsilon d})}{\epsilon d} + \frac{\log{n}}{n} \right)}$ and we show that $\delta \in (0,1)$:
    \begin{align*}
        \delta 
        = \sqrt{\frac{24}{\epsilon} \left(\frac{1}{\epsilon d} + \frac{\max(0, \log{\epsilon d})}{\epsilon d} + \frac{\log{n}}{n} \right)}
        \leq \sqrt{\frac{24}{\epsilon} \left( \frac{1}{\epsilon d} + \frac{1}{\sqrt{\epsilon d}} + \frac{1}{\sqrt{n}} \right)}
        \leq \sqrt{\frac{24}{\epsilon} \cdot \frac{3}{\epsilon \sqrt{d}}}
        = \frac{\sqrt{72}}{\epsilon d^{1/4}}
        \leq \frac{\epsilon}{2}, \tag{1}
    \end{align*}
    where the first inequality is due to $\frac{\max(0, \log{x})}{x} \leq \frac{1}{\sqrt{x}}$ for all $x$, the second inequality is due to $d \geq \sqrt{d}, \sqrt{\epsilon} \geq \epsilon$ and $n \geq d \geq \epsilon^2 d$ and the third inequality is by assumption $d \geq \frac{20^4}{\epsilon^8}$.
    By union bound over $S \in \binom{L}{k}$ and Chernoff bound:
    \begin{align*}
        \Pr\left[ |\neighbors{\opt{k}}| \geq (1 + \delta) \cdot \E\left[ |\neighbors{\smax{k}}| \right] \right]
        &= \Pr\left[ \max_{S \in \binom{L}{k}}|\neighbors{S}| \geq (1 + \delta) \cdot \E\left[ |\neighbors{\smax{k}}| \right] \right] \\
        &\leq \sum_{S \in \binom{L}{k}} \Pr\left[|\neighbors{S}| \geq (1 + \delta) \cdot \E\left[ |\neighbors{\smax{k}}| \right] \right] \\
        &\leq \sum_{S \in \binom{L}{k}} \Pr\left[|\neighbors{S}| \geq (1 + \delta) \cdot \E\left[ |\neighbors{S}| \right] \right] \\
        &\leq \sum_{S \in \binom{L}{k}} \exp\left( -\frac{\delta^2}{3} \cdot \E\left[ |\neighbors{S}| \right] \right) \\
        &= \sum_{S \in \binom{L}{k}} \exp\left( -\frac{\delta^2}{3} \cdot \E\left[ |\neighbors{\smax{k}}| \right] \right) \\
        &= \exp\left( \log{\binom{n}{k}} - \frac{\delta^2}{3} \E\left[ |\neighbors{\smax{k}}| \right]  \right), \tag{2}
    \end{align*}
    where the first line is by definition of $\opt{k}$, the second line is by union bound, the third line is because $\E\left[ |\neighbors{S}| \right] = \E\left[ |\neighbors{\smax{k}}| \right]$ for all $S \in \binom{L}{k}$, the fourth line is by \cref{lem: coverage of smax admits chernoff} since $\delta \in (0, 1)$ by Equation (1), the fifth line is because $\E\left[ |\neighbors{S}| \right] = \E\left[ |\neighbors{\smax{k}}| \right]$ and the sixth line is because $|\binom{L}{k}| = \binom{n}{k} = e^{\log{\binom{n}{k}}}$.
    We upper bound the exponent as follows:
    \begin{align*}
        \log{\binom{n}{k}} - \frac{\delta^2}{3} \E\left[ |\neighbors{\smax{k}}| \right]
        &\leq k \log\left( \frac{en}{k} \right) - 8 \left( \frac{ 1 + \max(0, \log{\epsilon d})}{\epsilon^2 d} + \frac{\log{n}}{\epsilon n} \right) \cdot \left( 1 - e^{-kd/n} \right) \cdot n \\
        &\leq \frac{1 + \log{ \frac{\epsilon d}{2} }}{\frac{\epsilon d}{2}} \cdot n - 8 \left( \frac{1 + \max(0, \log{\epsilon d})}{\epsilon^2 d} + \frac{\log{n}}{\epsilon n} \right) \cdot \left( 1 - e^{-\epsilon/2} \right) \cdot n \\
        &\leq \frac{1 + \log{\frac{\epsilon d}{2}}}{\frac{\epsilon d}{2}} \cdot n - 8 \left( \frac{1 + \max(0, \log{\epsilon d})}{\epsilon^2 d} + \frac{\log{n}}{\epsilon n} \right) \cdot \frac{\epsilon}{4} \cdot n \\
        &\leq \frac{1 + \max(0, \log{\epsilon d})}{\frac{\epsilon d}{2}} \cdot n - \left( \frac{1 + \max(0, \log{\epsilon d})}{\frac{\epsilon d}{2}} + \frac{2\log{n}}{n} \right) \cdot n \\
        &= - 2 \log{n} \\
        &\leq - \log{n}, \tag{2a}
    \end{align*}
    where the first line is due to $\binom{n}{k} \leq \left( \frac{en}{k} \right)^k$ and \cref{lem:smax expected coverage for n=m}, the second line is because $k \log\left( \frac{en}{k} \right)$ is increasing for all $k \in [n]$ and $\frac{\epsilon n}{2d} \leq k \leq \frac{2n}{\epsilon d}$ by assumption, the third line is because $1 - e^{-\epsilon} \geq \frac{\epsilon}{2}$ for all $\epsilon \in (0,1)$ and the fourth line is due to $\epsilon d \geq \frac{\epsilon d}{2}$.
    Let event $F = \{ |\neighbors{\opt{k}}| \geq (1 + \delta) \cdot \E\left[ |\neighbors{\smax{k}}| \right] \}$.
    We are ready to upper bound the expected value of $\optalgorithm$:
    \begin{align*}
        \E\left[ |\neighbors{\opt{k}}| \right] 
        &= \E\left[ |\neighbors{\opt{k}}| \middle| F \right] \cdot \Pr\left[ F \right] + \E\left[ |\neighbors{\opt{k}}| \middle| \bar{F} \right] \cdot \Pr\left[ \bar{F} \right] \\
        &\leq n \cdot \frac{1}{n} + \E\left[ |\neighbors{\opt{k}}| \middle| \bar{F} \right] \cdot 1 \\
        &\leq 1 + (1 + \delta) \cdot \E\left[ |\neighbors{\smax{k}}| \right] \\
        &\leq \left( 1 + \delta + \frac{\epsilon}{2} \right) \cdot \E\left[ |\neighbors{\smax{k}}| \right] \\
        &\leq \left( 1 + \epsilon \right) \cdot \E\left[ |\neighbors{\smax{k}} | \right] \\
        &\leq \left( 1 + \epsilon \right) \cdot \E\left[ |\neighbors{\greedy{k}} | \right], \tag{3}
    \end{align*}
    where the first line is by tower rule, the second line is because $|\neighbors{\opt{k}}| \leq n$ and $\Pr\left[ F \right] \leq \exp\left( - \log{n} \right) = \frac{1}{n}$ by substituting Equation (2a) to (2), the third line follows by definition of $F$, the fourth line is because $1 = \frac{\E\left[ |\neighbors{\smax{k}}| \right]}{\E\left[ |\neighbors{\smax{k}}| \right]} \leq \frac{\E\left[ |\neighbors{\smax{k}}| \right]}{\left( 1 - e^{-\epsilon/2} \right) \cdot n} \leq \frac{\E\left[ |\neighbors{\smax{k}}| \right]}{\frac{\epsilon}{4} \cdot n} \leq \frac{\epsilon}{2} \cdot \E\left[ |\neighbors{\smax{k}}| \right]$ by assumption $n \geq d \geq \frac{20^4}{\epsilon^8}$, the fifth line is because $\delta \leq \frac{\epsilon}{2}$ by Equation (1) and the sixth line is by \cref{lem:greedy lower bound by smax for n=m}.
    The statement follows directly from Equation (3):
    \begin{align*}
        \E\left[ |\neighbors{\greedy{k}}| \right] 
        \geq \frac{1}{1 + \epsilon} \cdot \E\left[ |\neighbors{\opt{k}}| \right]
        \geq \left( 1 - \epsilon \right) \cdot \E\left[ |\neighbors{\opt{k}}| \right].
    \end{align*}
    where the second inequality is because $\frac{1}{1 + \epsilon} \geq 1 - \epsilon$.
\end{proof}

\CriticalRegionSmallDegreesND*
\begin{proof}
    Let $t^*_d = \left( 1 - \left( 1 + d(d-1) \right)^{-\frac{1}{d-1}} \right) \cdot \frac{n}{d}$ and $\delta = 3e^{d^2} \cdot \sqrt{8n\log{2n}}$.
    Then:
    \begin{align*}
        \Pr\left[ |t_d - t^*_d| \leq \delta \right] \geq 1 - \frac{1}{n}, \tag{1}
    \end{align*}
    which follows by \cref{lem:first phase accepts for n=m} as $n \geq 2d^2$ by assumption.
    For $k \leq t^*_d - \delta$, $\greedyalgorithm$ gains $d$ on each iteration so:
    \begin{align*}
        \E\left[ |\neighbors{\greedy{k}}| \right] 
        \geq \Pr\left[ |t_d - t^*_d| \leq \delta \right] \cdot \E\left[ |\neighbors{\greedy{k}}| \middle| |t_d - t^*_d| \leq \delta \right] 
        \geq \left( 1 - \frac{1}{n} \right) \cdot kd
        \geq \left( 1 - \frac{1}{n} \right) \cdot \E\left[ |\neighbors{\opt{k}}| \right], \tag{2}
    \end{align*}
    where the second inequality is due to Equation (1).
    For $k \geq t^*_d - \delta$, $\greedyalgorithm$ gains $d$ on each iteration until $t^*_d - \delta$ and after that we use the worst case analysis which gives:
    \begin{align*}
        \E\left[ |\neighbors{\greedy{k}}| \right]
        &\geq \Pr\left[ |t_d - t^*_d| \leq \delta \right] \cdot \E\left[ |\neighbors{\greedy{k}} \middle| |t_d - t^*_d| \leq \delta \right] \\
        &\geq \left( 1 - \frac{1}{n} \right) \cdot \E\left[ |\neighbors{\greedy{k}}| \middle| |t_d - t^*_d| \leq \delta \right] \\
        &\geq \left( 1 - \frac{1}{n} \right) \cdot \left( 1 - \frac{1}{e} + \frac{1}{e} \cdot \left( \frac{t^*_d - \delta}{n/d} \right)^3 \right) \cdot \E\left[ |\neighbors{\opt{k}}| \middle| |t_d - t^*_d| \leq \delta \right], \tag{3}
    \end{align*}
    where the first line is by tower rule, the second line is due to Equation (1), the third line is due to \cref{lem:augmented worst case analysis for n=m} since $k \geq t^*_d - \delta \geq t_d$ by the condition.
    For the conditional value of $\optalgorithm$ we have:
    \begin{align*}
        \E&\left[ |\neighbors{\opt{k}}| \right] = \\
        &= \E\left[ |\neighbors{\opt{k}}| \middle| |t_d - t^*_d| \leq \delta \right] \cdot \Pr\left[ |t_d - t^*_d| \leq \delta \right] + \E\left[ |\neighbors{\opt{k}}| \middle| |t_d - t^*_d| > \delta \right] \cdot \Pr\left[ |t_d - t^*_d| > \delta \right] \\
        &\leq \E\left[ |\neighbors{\opt{k}}| \middle| |t_d - t^*_d| \leq \delta \right] \cdot 1 + n \cdot \frac{1}{n} \\
        &= \E\left[ |\neighbors{\opt{k}}| \middle| |t_d - t^*_d| \leq \delta \right] + 1 \\
        &\Rightarrow \E\left[ |\neighbors{\opt{k}}| \middle| |t_d - t^*_d| \leq \delta \right] 
        \geq \E\left[ |\neighbors{\opt{k}}| \right] - 1
        \geq \left( 1 - \frac{4}{\epsilon n} \right) \cdot \E\left[ |\neighbors{\opt{k}}| \right], \tag{3a}
    \end{align*}
    where the first line is by tower rule, the second line is due to $|\neighbors{\opt{k}}| \leq n$ and Equation (1), the third line is by cancellation and the fourth line is because $1 \leq \frac{\E\left[ |\neighbors{\opt{k}}| \right]}{\E\left[ |\neighbors{\smax{k}}| \right]} 
    \leq \frac{\E\left[ |\neighbors{\opt{k}}| \right]}{\left( 1 - e^{-\epsilon/2} \right) \cdot n} 
    \leq \frac{4}{\epsilon n} \cdot \E\left[ |\neighbors{\opt{k}}| \right]$ by \cref{lem:smax expected coverage for n=m}.
    Substituting Equation (3a) to (3) we have:
    \begin{align*}
        \E\left[ |\neighbors{\greedy{k}}| \right] 
        &\geq \left( 1 - \frac{1}{n} \right) \cdot \left( 1 - \frac{1}{e} + \frac{1}{e} \cdot \left( \frac{t^*_d - \delta}{n/d} \right)^3 \right) \cdot \left( 1 - \frac{4}{\epsilon n} \right) \cdot \E\left[ |\neighbors{\opt{k}}| \right] \\
        &\geq \left( 1 - \frac{1}{e} + \frac{1}{e} \cdot \left( \frac{t^*_d - \delta}{n/d} \right)^3 - \frac{8}{\epsilon n} \right) \cdot \E\left[ |\neighbors{\opt{k}}| \right], \tag{4}
    \end{align*}
    where the second line is because $\frac{1}{n} \leq \frac{4}{\epsilon n}$.
    At this point it suffices to show that $\frac{t^*_d - \delta}{n/d}$ is lower bounded by a positive constant:
    \begin{align*}
        \frac{t^*_d - \delta}{n/d} 
        \geq \frac{1}{d} - \frac{3e^{d^2} \cdot \sqrt{8n\log{2n}}}{n/d}
        \geq \frac{1}{d} - 12 d e^{d^2} \sqrt{\frac{\log{2n}}{2n}}
        \geq \frac{\epsilon^8}{20^4} - \frac{\frac{12 \cdot 20^4}{\epsilon^8} e^{\frac{20^4}{\epsilon^8}}}{(2n)^{1/4}}
        \geq \frac{\epsilon^8}{20^5}, \tag{4a}
    \end{align*} 
    where the first inequality is due to \cref{lem:simple lower bound for phase 1 n=m}, the second due to $d \leq \frac{8^4}{\epsilon^8}$, the third is by setting $n_{0}(\epsilon) = 8 \cdot \left( \frac{12 \cdot 20^8 e^{20^4/\epsilon^8}}{\epsilon^{16}}  \right)^4$.  
    Then we have:
    \begin{align*}
        \E\left[ |\neighbors{\greedy{k}}| \right]  
        &\geq \left( 1 - \frac{1}{e} + \frac{1}{e} \cdot \left( \frac{t^*_d - \delta}{n/d} \right)^3 - \frac{8}{\epsilon n} \right) \cdot \E\left[ |\neighbors{\opt{k}}| \right] \\
        &\geq \left( 1 - \frac{1}{e} + \frac{\epsilon^{24}}{20^{15} e} - \frac{8}{\epsilon n} \right) \cdot \E\left[ |\neighbors{\opt{k}}| \right] \\
        &\geq \left( 1 - \frac{1}{e} + \left( \frac{\epsilon}{20} \right)^{24} \right) \cdot \E\left[ |\neighbors{\opt{k}}| \right], \tag{5}
    \end{align*}
    where the first line is by Equation (4), the second line is by Equation (4a) and the third inequality is due to $n \geq 8 \cdot \left( \frac{20}{\epsilon} \right)^{26}$.
    By Equation (1) and Equation (5) we obtain the statement.
\end{proof}

\AugmentedWorstCaseAnalysisN*
\begin{proof}
    For all $k \geq t_d$ we have that:
    \begin{align*}
        |\neighbors{\greedy{k}}| 
        &\geq \left( 1 - \exp\left(- \left(1 - \frac{t_d}{k}\right) \right) \cdot \left( 1 - \frac{|\neighbors{\greedydet{t_d}}|}{|\neighbors{\optdet{k}}|} \right) \right) \cdot |\neighbors{\optdet{k}}| \\
        &\geq \left( 1 - \exp\left(- \left(1 - \frac{t_d}{k}\right) \right) \cdot \left( 1 - \frac{t_d \cdot d}{\min\{kd, n\}} \right) \right) \cdot |\neighbors{\optdet{k}}| \\
        &= \left( 1 - \exp\left(- \left(1 - \frac{1}{k/t_d}\right) \right) \cdot \left( 1 - \frac{1}{\min\{k/t_d, n/t_d d\}} \right) \right) \cdot |\neighbors{\optdet{k}}| \\
        &\geq \left( 1 - \exp\left(- \left(1 - \frac{t_d}{n/d}\right) \right) \cdot \left( 1 - \frac{t_d}{n/d} \right) \right) \cdot |\neighbors{\optdet{k}}| \\
        &\geq \left( 1 - \frac{1}{e} + \frac{1}{e} \cdot \left( \frac{t_d}{n/d} \right)^3 \right) \cdot |\neighbors{\optdet{k}}|,
    \end{align*}
    where the first line is due to \cref{lem:augmented worst case analysis n = m} for $t = t_d$, the second line is because $|\neighbors{\greedydet{t_d}}| = t_d \cdot d$ by definition of $t_d$ and $|\neighbors{\optdet{k}}| \leq \min\{kd, n \}$ for all $k$, the fourth line is by \cref{lem:critical region lower bound greedy m/d minimum n = m} since $k \geq t_d$ and the fifth line is due to \cref{lem:critical region lower bound greedy ratio simplification n = m} since $t_d \leq n/d$ by definition.
\end{proof}

\begin{lemma}
\label{lem:augmented worst case analysis n = m}
    Let $d \in \mathbbm{N}$ and $B = (L, R, E)$ be a bipartite graph. 
    Then for all $t \leq k$:
    $$|\neighbors{\greedydet{k}}| \geq \left( 1 - \exp\left(- \left( 1 - \frac{t}{k} \right) \right) \cdot \left( 1 - \frac{|\neighbors{\greedydet{t}}|}{|\neighbors{\optdet{k}}|} \right) \right) \cdot |\neighbors{\optdet{k}}|.$$
\end{lemma}
\begin{proof}
    We repeat the worst case analysis of $\greedyalgorithm$ for iterations $k \geq i > t$:
    \begin{align*}
        |\neighbors{\optdet{k}}| - |\neighbors{\greedydet{k}}| 
        \leq \left(1 - \frac{1}{k}\right)^{1} \cdot \left(|\neighbors{\optdet{k}}| - |\neighbors{\greedydet{k-1}}| \right)
        \leq \left(1 - \frac{1}{k}\right)^{k-t} \cdot \left(|\neighbors{\optdet{k}}| - |\neighbors{\greedydet{t}}| \right), \tag{1}
    \end{align*}
    and solving for $|\neighbors{\greedydet{k}}|$ we have:
    \begin{align*}
        |\neighbors{\greedydet{k}}| 
        &\geq \left( 1 - \left(1 - \frac{1}{k} \right)^{k - t} \cdot \left(1 - \frac{|\neighbors{\greedydet{t}}|}{|\neighbors{\optdet{k}}|} \right) \right) \cdot |\neighbors{\optdet{k}}| \\
        &\geq \left( 1 - \exp\left\{-\left(1 - \frac{t}{k}\right)\right\} \cdot \left(1 - \frac{|\neighbors{\greedydet{t}}|}{|\neighbors{\optdet{k}}|} \right) \right) \cdot |\neighbors{\optdet{k}}|,
    \end{align*}
    where the first line is by (1) and the second line is due to $1 - \frac{1}{k} \leq e^{-\frac{1}{k}}$.
\end{proof}

\begin{lemma}
\label{lem:critical region lower bound greedy m/d minimum n = m}
    Let $y \geq 1$ and $f(x) = 1 - \exp\left(- \left( 1 - \frac{1}{x} \right) \right) \cdot \left( 1 - \frac{1}{\min\{x, y\}} \right)$. 
    Then we have that
    $f(x) \geq 1 - \exp\left(- \left( 1 - \frac{1}{y} \right) \right) \cdot \left( 1 - \frac{1}{y} \right)$ for all $x \geq 1$.
\end{lemma}
\begin{proof}
    We start with case $x \leq y$:
    \begin{align*}
        f(k) 
        = 1 - \exp\left(- \left( 1 - \frac{1}{x} \right) \right) \cdot \left( 1 - \frac{1}{x} \right)
        \geq 1 - \exp\left(- \left( 1 - \frac{1}{y} \right) \right) \cdot \left( 1 - \frac{1}{y} \right).
    \end{align*}
    For $x > y$:
    \begin{align*}
        f(k) = 1 - \exp\left(- \left( 1 - \frac{1}{x} \right) \right) \cdot \left( 1 - \frac{1}{y} \right)
        \geq 1 - \exp\left(- \left( 1 - \frac{1}{y} \right) \right) \cdot \left( 1 - \frac{1}{y} \right).
    \end{align*}
\end{proof}

\begin{lemma}
\label{lem:critical region lower bound greedy ratio simplification n = m}
    For all $x \leq 1$ we have that $1 - e^{-(1-x)} \cdot (1 - x) \geq 1 - \frac{1}{e} + \frac{x^3}{e}$.
\end{lemma}
\begin{proof}
    We have that:
    \begin{align*}
        1 - e^{-(1-x)} \cdot (1 - x) 
        &= 1 - e^{-1} e^{x} \cdot (1 - x) \\
        &\geq 1 - e^{-1} (1 + x + x^2) \cdot (1 - x) \\
        &= 1 - e^{-1} (1 - x^3) \\
        &= 1 - \frac{1}{e} + \frac{x^3}{e},
    \end{align*}
    where the inequality is due to $e^x \leq 1 + x + x^2$ for all $x \leq 1$.
\end{proof}

\begin{lemma}
\label{lem:simple lower bound for phase 1 n=m}
    For all $d \geq 2$ we have that $1 - \left( 1 + d(d-1) \right)^{- \frac{1}{d-1}} \geq \frac{1}{d}$.
\end{lemma}
\begin{proof}
    Equivalently we show that $\left( 1 + d(d-1) \right)^{-\frac{1}{d-1}} \leq \frac{d-1}{d}$ or after taking logs:
    \begin{align*}
        \log\left( 1 + d(d-1) \right) \geq (d-1) \cdot \log\left( \frac{d}{d-1} \right).
    \end{align*}
    The left-hand side is greater than $1$ since $1+d(d-1) \geq e$ for all $d \geq 2$ and the right-hand side is less than $1$ since $(d-1) \cdot \log\left( \frac{d}{d-1} \right) \leq (d-1) \cdot \frac{1}{d-1} = 1$.
\end{proof}

\CLAIMONE*
\begin{proof}
    We have:
    \begin{align*}
        \frac{1 - e^{-\epsilon}}{\epsilon}
        \geq \frac{1 - (1 - \epsilon + \epsilon^2))}{\epsilon}
        = 1 - \epsilon
    \end{align*}
    where the inequality is due to $e^{-\epsilon} \leq 1 - \epsilon + \epsilon^2$ which holds for all $\epsilon \in (0,1)$.
\end{proof}

\CLAIMTWO*
\begin{proof}
    We have:
    \begin{align*}
        \left( 1 - \frac{1}{k} \right)^k
        = e^{k \log\left( 1 - \frac{1}{k} \right)}
        \leq e^{k \left( - \frac{1}{k} - \frac{1}{2k^2} \right)}
        = e^{-1 - \frac{1}{2k}}
        = e^{-1} \cdot e^{- \frac{1}{2k}}
        \leq e^{-1} \cdot \left( 1 - \frac{\frac{1}{2k}}{2} \right)
        = \frac{1}{e} - \frac{1}{4ek},
    \end{align*}
    where the first equality is by raising to the exponent, the first inequality is due to inequality $\log\left( 1 - x \right) \leq - x - \frac{x^2}{2}$ for all $x \geq 0$, the second equality is by simplifying the exponent, and the second inequality is due to inequality $e^{-x} \leq 1 - \frac{x}{2}$ for all $x \in [0,1]$.
\end{proof}

\section{Omitted Proofs Of Section 4}
\label{sec: omitted proofs of section 4}

\begin{definition}[Erdos-Renyi random graph]
    $\mathbbm{G}_{n, m}$ samples graphs $\cG = (V, \cE)$, with $|V| = n$ and $\cE \subseteq \binom{V}{2}$ chosen uniformly, such that $|\cE| = m$.
\end{definition}
\begin{definition}[Random graph with replacement]
    $\mathbbm{G}^r_{n, m}$ samples graphs $\cG^r = (V, \cE^r)$, with $|V| = n$. Also let $e_1, e_2, \dots, e_m \in \binom{V}{2}$ be independent uniform random samples (with replacement), then $\cE^r = \{e_1, \dots, e_m\}$.
\end{definition}
\begin{definition}[Maximum matching]
    Let $G = (V, E)$ and let $M$ be a maximum matching on $G$. We denote the size of this matching as $\mu(G) = |M|$ .
\end{definition}
\begin{definition}[Maximum matching on random graphs]
    Let $\cG \sim \mathbbm{G}_{n,m}$ and $\cG^r \sim \mathbbm{G}^r_{n,m}$. 
    Then $\mu(n,m) = \mu(\cG)$ denotes the size of the maximum matching on $\cG$ and $\mu^r(n,m) = \mu(\cG^r)$ denotes the size of the maximum matching on $\cG^r$.
\end{definition}

We state a classic result for the size of maximum matchings in random graphs by \cite{karp1981MaximumMatchingRandomGraphs}, but we use the statement of \cite{frieze1998MaximumMatchingRevisited}. The difference in some constants is because we assume $n = cm$, while in the original matching setting it was $n = cm/2$.
\begin{restatable}{lemma}{MaxMatchingSizeErdosRenyi}[Theorem 4 in \cite{frieze1998MaximumMatchingRevisited}]
\label{lem: maximum matching size in erdos-renyi}
    Let $n = cm$ for some constant $c > 0$ and sample graph $\cG$ from $\mathbbm{G}_{m, n}$.
    Then for any constant $\epsilon > 0$ we have that:
    \begin{align*}
        \Pr\left[ \middle| \frac{\mu(m, n)}{m} - \left(1 - \frac{\gamma^* + \gamma_* + \gamma^* \gamma_*}{4c} \right) \middle| \geq m^{-\frac{1}{7} + \epsilon} \right] = o(1),
    \end{align*}
    where $\gamma_*$ is the smallest root of $x = 2c e^{-2c e^{-x}}$ and $\gamma^* = 2c e^{-\gamma_*}$.
\end{restatable}

The following lemma shows that with high probability there is a sublinear number of pairs of duplicated edges in graphs drawn from model $\mathbbm{G}^r_{n, n}$.
\begin{lemma}
\label{lem: duplicated edges are sublinear}
    Draw $\cG^r(V, \cE^r) \sim \mathbbm{G}^r_{n, n}$.
    Then $\Pr\left[ n - |\cE^r| \geq \log{n} \right] \leq \frac{ 1 }{\log{n}}$.
\end{lemma}
\begin{proof}
    Let $X = n - |\cE^r|$, then $X \leq \sum_{1 \leq i < j \leq n} \mathbbm{1}\{e_i = e_j \}$, since the right-hand side counts more collisions when the same edge is sampled multiple times.
    Let $i \neq j$ arbitrary and we have that:
    \begin{align*}
        \Pr\left[ e_i = e_j \right]
        &= \sum_{\{v, v'\} \in \binom{V}{2}} \Pr\left[ e_i = \{v, v'\}, e_j = \{v, v'\} \right] \\
        &= \sum_{\{v, v'\} \in \binom{V}{2}} \Pr\left[ e_i = \{v, v'\} \right] \cdot \Pr\left[ e_j = \{v, v'\} \right] \\
        &= \binom{n}{2} \cdot \left( \frac{1}{\binom{n}{2}} \right)^2 \\
        &= \frac{1}{\binom{n}{2}}, \tag{1}
    \end{align*}
    where the first line is by law of total probability, the second line is due to independent choice of edges and the third line is because edges are chosen uniformly at random with replacement.
    For the expectation of $X$ we have that:
    \begin{align*}
        \E[X] \leq \sum_{1 \leq i < j \leq n} \Pr\left[ e_i = e_j \right]
        = \sum_{1 \leq i < j \leq n} \frac{1}{\binom{n}{2}}
        = \frac{\binom{n}{2}}{\binom{n}{2}}
        = 1 \tag{2}
    \end{align*}
    where the second equality is by (1).
    Finally we show the statement:
    \begin{align*}
        \Pr\left[ X \geq \log{n} \right] 
        \leq \frac{\E[X]}{\log{n}}
        \leq \frac{1}{\log{n}},
    \end{align*}
    where the first inequality is by Markov and the second is due to (2).
\end{proof}

Now we show that the maximum matching size in both random graph models is close.
\begin{lemma}
\label{lem:approximation of matching in random graph with replacement by erdos renyi}
    $\Pr\left[ \mu^r(n, n) \geq x \right] 
    \geq \Pr\left[ \mu(n, n) \geq x + \log{n} \right] \cdot \left( 1 - \frac{ 1 }{\log{n}} \right)$ for all $x$.
\end{lemma}
\begin{proof}
    We construct a coupling between $\mathbbm{G}^r_{n, n}$ and $\mathbbm{G}_{n, n}$ by first sampling $\hat{\cG}^r = (V, \hat{\cE}^r) \sim \mathbbm{G}^r_{n, n}$ and then sampling an additional $n - |\hat{\cE}^r|$ edges from $\binom{V}{2} \setminus \hat{\cE}^r$ without replacement to obtain $\hat{\cG}_{n, n} = (V, \hat{\cE}) \sim \mathbbm{G}_{n, n}$.
    Then we have that:
    \begin{align*}
        \Pr\left[ \mu^r(n, n) \geq x \right] 
        &= \Pr\left[ \mu(\hat{\cG^r}) \geq x \right] \\
        &\geq \Pr\left[ \mu(\hat{\cG^r}) \geq x \middle| |\hat{\cE} \setminus \hat{\cE}^r| \leq \log{n} \right] \cdot \Pr\left[ |\hat{\cE} \setminus \hat{\cE}^r| \leq \log{n} \right] \\
        &\geq \Pr\left[ \mu(\hat{\cG^r}) \geq x \middle| |\hat{\cE} \setminus \hat{\cE}^r| \leq \log{n} \right] \cdot \left( 1 - \frac{ 1 }{\log{n}} \right) \\
        &\geq \Pr\left[ \mu(\hat{\cG}) \geq x + \log{n} \middle| |\hat{\cE} \setminus \hat{\cE}^r| \leq \log{n} \right] \cdot \left( 1 - \frac{ 1 }{\log{n}} \right) \\
        &= \Pr\left[ \mu(\hat{\cG}) \geq x + \log{n} \right] \cdot \left( 1 - \frac{ 1 }{\log{n}} \right) \\
        &= \Pr\left[ \mu(n, n) \geq x + \log{n} \right] \cdot \left( 1 - \frac{ 1 }{\log{n}} \right),
    \end{align*}
    where the first line is by definition of $\hat{\cG}^r$, the second line is by law of total probability, the third line is by \cref{lem: duplicated edges are sublinear} since $|\hat{\cE} \setminus \hat{\cE}^r| = n - |\hat{\cE}^r|$, the fourth line is because conditional on $|\hat{\cE} \setminus \hat{\cE}^r| \leq \log{n}$ we have $\mu(\hat{\cG}^r) \geq \mu(\hat{\cG}) - \log{n}$, so $\mu(\hat{\cG}) \geq x + \log{n}$ implies $\mu(\hat{\cG^r}) \geq x$, the fifth line is due to the condition being independent of matching size $\mu(\hat{\cG})$ and the sixth line by definition of $\hat{\cG}$.
\end{proof}

The following lemma gives a bound for the size of the maximum matching in graphs sampled from $\mathbbm{G}^r_{n, n}$, with high probability. It is only slightly weaker than \cref{lem: maximum matching size in erdos-renyi}, which is the corresponding result for Erd\H{o}s-Renyi $\mathbbm{G}_{n,n}$.
\LBmaxmatchGB*
\begin{proof}   
    Then we have that:
    \begin{align*}
        \Pr&\left[ \mu(\cI_{\cB}) \geq \left(1 - \frac{\gamma^* + \gamma_* + \gamma^* \gamma_*}{4} - n^{-\frac{1}{9}} \right) \cdot n \right] = \\
        &= \Pr\left[ \mu^r(n, n) \geq \left(1 - \frac{\gamma^* + \gamma_* + \gamma^* \gamma_*}{4} - n^{-\frac{1}{9}} \right) \cdot n \right] \\
        &\geq \Pr\left[ \mu(n, n) \geq \left(1 - \frac{\gamma^* + \gamma_* + \gamma^* \gamma_*}{4} - n^{-\frac{1}{9}} \right) \cdot n + \log{n} \right] \cdot \left( 1 - \frac{ 1 }{\log{n}} \right) \\
        &= \Pr\left[ \mu(n, n) - \left(1 - \frac{\gamma^* + \gamma_* + \gamma^* \gamma_*}{4} \right) n \geq \log{n} - n^{1 - \frac{1}{9}} \right] \cdot \left( 1 - \frac{ 1 }{\log{n}} \right) \\
        &\geq \Pr\left[ \mu(n, n) - \left(1 - \frac{\gamma^* + \gamma_* + \gamma^* \gamma_*}{4} \right) n \geq - n^{1 - \frac{1}{8} }\right] \cdot \left( 1 - \frac{ 1 }{\log{n}} \right) \\
        &\geq \Pr\left[ \abs{\mu(n, n) - \left(1 - \frac{\gamma^* + \gamma_* + \gamma^* \gamma_*}{4} \right) n} \leq n^{1 - \frac{1}{8}} \right] \cdot \left( 1 - \frac{ 1 }{\log{n}} \right)\\
        &= 1 - \Pr\left[ \abs{ \frac{\mu(n, n)}{n} - \left(1 - \frac{\gamma^* + \gamma_* + \gamma^* \gamma_*}{4} \right)} \geq n^{- \frac{1}{8}} \right] \cdot \left( 1 - \frac{ 1 }{\log{n}} \right) \\
        &= 1 - o(1),
    \end{align*}
    the first line is by definition of $\mu^r(n, n)$, the second line is due to \cref{lem:approximation of matching in random graph with replacement by erdos renyi}, the fourth line is due to $\log{n} - n^{1-1/9} = \log{n} - n^{8/9} \leq - n^{7/8} = -n^{1 - 1/8}$ by the assumption for large $n$, the fifth line is because this event is a subset, the sixth line is by taking the complement and the seventh line is due to \cref{lem: maximum matching size in erdos-renyi} for $\epsilon = \frac{1}{7} - \frac{1}{8} > 0$.
\end{proof}

\section{Generalization to Unbalanced Bipartite Graphs}
\label{sec: theorem 1 generalization n = c m}
In this section we generalize \cref{thm:Theorem 1} to the case $n = O(m)$. We assume $n = c \cdot m$ and we show a lower bound for the ratio that decreases as $c$ increases. Nevertheless, for all constant $c$, we obtain a bound that is a constant away from $1 - \frac{1}{e}$.
Our analysis is very similar to that of \cref{thm:Theorem 1}, except for the case where the number of left nodes is much smaller than the number of right nodes, i.e. $c$ much smaller than $1$. 
We use random model $\textsc{ULRR}$, which we define as follows. 
\begin{definition}
    Let $n, m \in \mathbbm{N}$ and $d \in [m]$. The random model \textsc{ULRR}($n, m, d$) samples bipartite graphs $\cB = (L, R, \cE)$ from \textsc{GenR}($n, m, d$).
\end{definition}
In this section we use results from \cref{sec:generalization of technical tools}, which hold since $\textsc{GenR}$ is a generalization.

\subsection{Linear And Saturated Regions.}
\begin{lemma}
\label{lem:ratio for k away from m/d}
    Let $\epsilon \in (0, 1)$, $m \in \mathbbm{N}$, $c \geq 0$, $n = cm$, $d \in [m], k \in \left[0, \frac{\epsilon m}{2d}\right] \cup \left[\frac{2m}{\epsilon d}, n \right]$ and $\cB \sim \textsc{ULRR}(n, m, d)$.
    Then we have $\E\left[ |\neighbors{\greedy{k}}| \right] \geq  \left( 1 - \epsilon \right) \cdot \E\left[ |\neighbors{\opt{k}}| \right]$.
    
\end{lemma}
\begin{proof}
    For $k \leq \frac{\epsilon m}{2d}$ we have that:
    \begin{align*}
        \E\left[ |\neighbors{\greedy{k}}| \right] 
        &\geq \E\left[ |\neighbors{\smax{k}}| \right] \\
        &\geq \left( 1 - e^{-kd/m} \right) \cdot m \\
        &= \frac{1 - e^{-kd/m}}{kd/m} \cdot kd  \\
        &\geq \frac{1 - e^{- \epsilon}}{\epsilon} \cdot kd \\
        &\geq \left( 1 - \epsilon \right) \cdot kd \\
        &\geq \left(1 - \epsilon \right) \cdot \E\left[ |\neighbors{\opt{k}}| \right], \tag{1}
    \end{align*}
    where the first line is by \cref{lem: greedy lower bound by smax}, the second line is by \cref{lem:smax expected coverage}, the fourth line is due to assumption $k \leq \frac{\epsilon m}{2d} \leq \frac{\epsilon m}{d}$, the fifth line is because $\frac{1-e^{-\epsilon}}{\epsilon} \geq (1 - \epsilon)$ and the sixth line is due to $kd \geq |\neighbors{\opt{k}}|$ by \cref{lem: trivial upper bound for opt}.
    For $k \geq \frac{2m}{\epsilon d}$ we have that:
    \begin{align*}
        \E\left[ |\neighbors{\greedy{k}}| \right] 
        &\geq \E\left[ |\neighbors{\smax{k}}| \right] \\
        &\geq \left( 1 - e^{-kd/m} \right) \cdot m  \\
        &\geq \left( 1 - e^{-\frac{1}{\epsilon}} \right) \cdot m \\ 
        &\geq \left( 1 - \epsilon \right) \cdot m \\
        &\geq \left( 1 - \epsilon \right) \cdot \E\left[ |\neighbors{\opt{k}}| \right], \tag{2}
    \end{align*}
    where the first line is due to \cref{lem: greedy lower bound by smax}, the second line is due to \cref{lem:smax expected coverage}, the third is by assumption $k \geq \frac{2m}{\epsilon d} \geq \frac{m}{\epsilon d}$, the fourth line is due to $1-e^{-\frac{1}{\epsilon}} \geq 1 - \frac{1}{1/\epsilon} = 1 - \epsilon$ for all $\epsilon \in (0,1)$ and the fifth line is due to $m \geq |\neighbors{\opt{k}}|$ by \cref{lem: trivial upper bound for opt}.
    By (1) and (2) we obtain the statement.
\end{proof}

\subsection{Critical Region: Large Degrees.}
\begin{lemma}
\label{lem:ratio for large d}
    Let $\epsilon \in (0, 1)$, $m \in \mathbbm{N}$, $c \geq 0$, $n = cm$, $d \geq \frac{20^4 \max(1, c^2)}{\epsilon^8}, k \in [\frac{\epsilon m}{2 d}, \frac{2 m}{\epsilon d}]$ and $\cB \sim \textsc{ULRR}(n, m, d)$.
    Then we have $\E\left[ |\neighbors{\greedy{k}}| \right] \geq \left( 1 - \epsilon \right) \cdot \E[|\neighbors{\opt{k}}|]$.
\end{lemma}
\begin{proof}
    To upper bound the optimal solution, first notice that $|\neighbors{\opt{k}}| = \max_{S \in \binom{L}{k}}|\neighbors{S}|$.
    Also let $$\delta = \sqrt{\frac{24}{\epsilon} \left( \max\left(0, \frac{ \log{e c} + \log{\frac{\epsilon d}{2}}}{\epsilon d} \right)  + \frac{\log{m}}{m} \right)}.$$ 
    We proceed to show that $\delta \in (0,1)$:
    \begin{align*}
        \delta 
        &\leq \sqrt{\frac{24}{\epsilon} \left( \max\left(0, \frac{\log{e c}}{\sqrt{d}} \cdot \frac{1}{\epsilon \sqrt{d}} + \frac{\log{\epsilon d}}{\epsilon d} \right) + \frac{\log{m}}{m} \right)} \\
        &\leq \sqrt{\frac{24}{\epsilon} \left( \max\left(0, \frac{1}{\epsilon \sqrt{d}} + \frac{1}{\sqrt{\epsilon d}} \right) + \frac{1}{\sqrt{m}} \right)} \\
        &\leq \frac{\sqrt{72}}{\epsilon d^{1/4}} \\
        &\leq \frac{\epsilon}{2}, \tag{1}
    \end{align*}
    where the second line is due to $d \geq \frac{20^4 \max(1, c^2)}{\epsilon^8} \geq (\max(1, \log{e c}))^2$ and $\frac{\log{x}}{x} \leq \frac{1}{\sqrt{x}}$ for all $x$, the third line is due to $\sqrt{\epsilon} \geq \epsilon, m \geq d \geq \epsilon^2 d$ and the fourth line is by assumption $d \geq \frac{20^4}{\epsilon^8}$.
    Then we have:
    \begin{align*}
        \Pr\left[ |\neighbors{\opt{k}}| \geq (1 + \delta) \cdot \E\left[ |\neighbors{\smax{k}}| \right] \right]
        &= \Pr\left[ \max_{S \in \binom{L}{k}}|\neighbors{S}| \geq (1 + \delta) \cdot \E\left[ |\neighbors{\smax{k}}| \right] \right] \\
        &\leq \sum_{S \in \binom{L}{k}} \Pr\left[|\neighbors{S}| \geq (1 + \delta) \cdot \E\left[ |\neighbors{\smax{k}}| \right] \right] \\
        &\leq \sum_{S \in \binom{L}{k}} \Pr\left[|\neighbors{\smax{k}}| \geq (1 + \delta) \cdot \E\left[ |\neighbors{\smax{k}}| \right] \right] \\
        &\leq \sum_{S \in \binom{L}{k}} \exp\left( -\frac{\delta^2}{3} \cdot \E\left[ |\neighbors{\smax{k}}| \right] \right) \\
        &= \exp\left( \log{\binom{n}{k}} - \frac{\delta^2}{3} \E\left[ |\neighbors{\smax{k}}| \right]  \right), \tag{2}
    \end{align*}
    where the first line is by definition of $\opt{k}$, the second line is by union bound, the third line is because $|\neighbors{S}|$ is dominated by $|\neighbors{\smax{k}}$ for all $S \in \binom{L}{k}$ by \cref{lem: coverage of s dominated by smax}, the fourth line is by \cref{lem: coverage of smax admits chernoff} since $\delta \in (0, 1)$ by Equation (1) and the fifth line is because $|\binom{L}{k}| = \binom{n}{k} = e^{\log{\binom{n}{k}}}$.
    We upper bound the exponent as follows:
    \begin{align*}
        \log{\binom{n}{k}} - \frac{\delta^2}{3} \E\left[ |\neighbors{\smax{k}}| \right]
        &\leq k \log\left( \frac{en}{k} \right) - \frac{\delta^2}{3} \cdot \left( 1 - e^{-kd/m} \right) \cdot m \\
        &\leq \frac{\log{\frac{e c \epsilon d}{2}}}{\frac{\epsilon d}{2}} \cdot m - \frac{\delta^2}{3} \cdot \left( 1 - e^{-\epsilon/2} \right) \cdot m \\
        &\leq \frac{\log{e c} + \log{\frac{\epsilon d}{2}}}{\frac{\epsilon d}{2}} \cdot m - \frac{\delta^2}{3} \cdot \frac{\epsilon}{4} \cdot m \\
        &\leq \frac{\log{e c} + \log{\frac{\epsilon d}{2}}}{\frac{\epsilon d}{2}} \cdot m - \left( \frac{\log{e c} + \log{\frac{\epsilon d}{2}}}{\frac{\epsilon d}{2}} + \frac{2\log{m}}{m} \right) \cdot m \\
        &\leq - \log{m}, \tag{2a}
    \end{align*}
    where the first line is due to $\binom{n}{k} \leq \left( \frac{en}{k} \right)^k$ and \cref{lem:smax expected coverage}, the second line is because $k \log\left( \frac{en}{k} \right)$ is increasing for all $k \in [n]$ and $\frac{\epsilon m}{2 d} \leq k \leq \frac{2 m}{\epsilon d}$ by assumption, the third line is because $1 - e^{-\epsilon/2} \geq \frac{\epsilon}{4}$ for all $\epsilon \in (0,1)$, the fourth line is by definition of $\delta$ and the fifth line is due to $-2 \log{m} \leq - \log{m}$.
    Let event $F = \{ |\neighbors{\opt{k}}| \geq (1 + \delta) \cdot \E\left[ |\neighbors{\smax{k}}| \right] \}$.
    We are ready to upper bound the expected value of $\optalgorithm$:
    \begin{align*}
        \E\left[ |\neighbors{\opt{k}}| \right] 
        &= \E\left[ |\neighbors{\opt{k}}| \middle| F \right] \cdot \Pr\left[ F \right] + \E\left[ |\neighbors{\opt{k}}| \middle| \bar{F} \right] \cdot \Pr\left[ \bar{F} \right] \\
        &\leq n \cdot \frac{1}{m} + \E\left[ |\neighbors{\opt{k}}| \middle| \bar{F} \right] \cdot 1 \\
        &\leq c + (1 + \delta) \cdot \E\left[ |\neighbors{\smax{k}}| \right] \\
        &\leq \left( 1 + \delta + \frac{4c}{\epsilon m} \right) \cdot \E\left[ |\neighbors{\smax{k}}| \right] \\
        &\leq \left( 1 + \epsilon \right) \cdot \E\left[ |\neighbors{\smax{k}} | \right] \\
        &\leq \left( 1 + \epsilon \right) \cdot \E\left[ |\neighbors{\greedy{k}} | \right], \tag{3}
    \end{align*}
    where the first line is by tower rule, the second line is because $|\neighbors{\opt{k}}| \leq m$ and $\Pr\left[ F \right] \leq \exp\left( - \log{m} \right) = \frac{1}{m}$ by substituting Equation (2a) to (2), the third line follows by definition of $F$, the fourth line is because $c = c \frac{\E\left[ |\neighbors{\smax{k}}| \right]}{\E\left[ |\neighbors{\smax{k}}| \right]} \leq c\frac{\E\left[ |\neighbors{\smax{k}}| \right]}{\left( 1 - e^{-\epsilon/2} \right) \cdot m} \leq \frac{\E\left[ |\neighbors{\smax{k}}| \right]}{\frac{\epsilon}{4c} \cdot m}$, the fifth line is because $\delta \leq \frac{\epsilon}{2}$ by Equation (1) and $m \geq d \geq \frac{20^4 c^2}{\epsilon^8} \geq \frac{8 c}{\epsilon^2}$ and the sixth line is by \cref{lem: greedy lower bound by smax}.
    The statement follows directly from Equation (3):
    \begin{align*}
        \E\left[ |\neighbors{\greedy{k}}| \right] 
        \geq \frac{1}{1 + \epsilon} \cdot \E\left[ |\neighbors{\opt{k}}| \right]
        \geq \left( 1 - \epsilon \right) \cdot \E\left[ |\neighbors{\opt{k}}| \right].
    \end{align*}
    where the second inequality is because $\frac{1}{1 + \epsilon} \geq 1 - \epsilon$ for all $\epsilon \in (0, 1)$.
\end{proof}

\subsection{Critical Region: Small \texorpdfstring{$n$}{n}.}
\begin{lemma}
\label{lem:ratio for small n/m}
    Let $\epsilon \in (0,1)$, $m \geq \frac{8}{\epsilon^2}$, $0 \leq c \leq \frac{\epsilon^3}{48} - \frac{\log{m}}{m}$, $n = cm$, $d \in [m], k \in [\frac{\epsilon m}{2d}, \frac{2m}{\epsilon d}]$ and $\cB \sim \textsc{ULRR}(n, m, d)$.
    Then we have $\E\left[ |\neighbors{\greedy{k}}| \right] \geq \left( 1 - \epsilon \right) \cdot \E\left[ |\neighbors{\opt{k}}| \right]$.
\end{lemma}
\begin{proof}
    For $\delta = \frac{\epsilon}{2} \in (0,1)$, by \cref{lem:greedy lower bound by concentration} we have that:
    \begin{align*}
        \E\left[ |\neighbors{\greedy{k}}| \right] 
        \geq \left(1 - \frac{\epsilon}{2} \right) \cdot \E\left[ |\neighbors{\opt{k}}| \right] - m \cdot \exp\left(\log{\binom{n}{k}} - \frac{\epsilon^2}{12} ( 1 - e^{-kd/m}) \cdot m \right).  \tag{1}
    \end{align*}
    We upper bound the exponent of the second term:
    \begin{align*}
        \log{\binom{n}{k}} - \frac{\epsilon^2}{12} \cdot (1 - e^{- kd/m}) \cdot m 
        &\leq \log{2} \cdot n - \frac{\epsilon^2}{12} \cdot (1 - e^{-\epsilon/2}) \cdot m \\
        &\leq n - \frac{\epsilon^3}{48} \cdot m \\
        &= \left( c - \frac{\epsilon^3}{48} \right) \cdot m \\
        &\leq - \frac{\log{m}}{m} \cdot m \\
        &= - \log{m} \tag{1a}
    \end{align*}
    where the first line is due to $\binom{n}{k} \leq 2^n$ and $k \geq \frac{\epsilon m}{2d}$ by the assumption, the second line is due to $\log{2} \leq 1$ and $1 - e^{-\epsilon/2} \geq \frac{\epsilon}{4}$ for $\epsilon \in (0,1)$, the third line is by assumption $n = c m$ and the fourth line is by assumption $c \leq \frac{\epsilon^3}{48} - \frac{\log{m}}{m}$.
    Substituting (1a) into (1) we have:
    \begin{align*}
        \E\left[ |\neighbors{\greedy{k}}| \right] 
        &\geq (1 - \frac{\epsilon}{2}) \cdot \E\left[ |\neighbors{\opt{k}}| \right] - m \cdot \exp\left(- \log{m} \right) \\
        &= (1 - \frac{\epsilon}{2}) \cdot \E\left[ |\neighbors{\opt{k}}| \right] - 1 \\
        &\geq \left( 1 - \epsilon \right) \cdot \E\left[ |\neighbors{\opt{k}}| \right],
    \end{align*}
    where the second line is by cancellation and the third line is due to $1 \leq \frac{\E\left[ |\neighbors{\opt{k}}| \right]}{ \E\left[ |\neighbors{\smax{k}}| \right]} \leq \frac{4}{\epsilon m} \cdot \E\left[ |\neighbors{\opt{k}}| \right]$ and $\frac{4}{\epsilon m} \leq \frac{\epsilon}{2}$.
\end{proof}

\subsection{Critical Region: Small Degrees.}
\begin{lemma}
\label{lem:ratio for small d large n/m}
    Let $\epsilon \in (0,1)$, $m \geq \frac{1}{2c} \cdot \left( c \left( \frac{20^6 c^2}{\epsilon^8} \right)^2 e^{\frac{20^4 c^3}{\epsilon^8}} \right)^4$, $c \geq \frac{\epsilon^3}{48} - \frac{\log{m}}{m}$, $n = cm$, $d \leq \frac{20^4 \max(1, c^2)}{\epsilon^8}, k \in [\frac{\epsilon m}{2d}, \frac{2m}{\epsilon d}]$ and $\cB \sim \textsc{ULRR}(n, m, d)$.
    Then there exists a function $f(x, y)$ increasing in $x$ and decreasing in $y$, such that $\E\left[ |\neighbors{\greedy{k}}| \right] \geq \left( 1 - \frac{1}{e} + f(\epsilon, c) \right) \cdot \E[|\neighbors{\opt{k}}|]$.
\end{lemma}
\begin{proof}
    Let random variable $t_d = \max\{t \in \{1, \dots, n\}: |\neighbors{\greedy{t}} \setminus \neighbors{\greedy{t-1}}| = d \}$. 
    Notice that by definition, $t_d = |\acceptreject{d}{n}|$, i.e. $t_d$ is the number of accepts in the first phase of $\acceptrejectalgorithm$. 
    Also let 
    \begin{align*}
        t^*_d = \frac{m}{d} \cdot \left( 1 - \left( 1 + c d (d-1) \right)^{-\frac{1}{d-1}} \right),\
        \delta = 3 e^{c d^2} \cdot \sqrt{8n \log(2n)}.
    \end{align*}
    Then we have:
    \begin{align*}
        \E\left[ |\neighbors{\greedy{k}}| \right]  
        &\geq \E\left[ |\neighbors{\greedy{k}}| \middle| |t_d - t^*_d| \leq \epsilon \right] \cdot \Pr\left[ |t_d - t^*_d| \leq \epsilon \right] \\
        &\geq \E\left[ |\neighbors{\greedy{k}}| \middle| |t_d - t^*_d| \leq \epsilon \right] \cdot \left( 1 - \frac{1}{n} \right) \tag{1}
    \end{align*}
    where the first inequality is by tower rule and the second is due to $\Pr\left[ |t_d - t^*_d| \leq \epsilon \right] \geq 1 - \frac{1}{n} $ by \cref{lem:first phase number of accepts}.
    For $k \leq t^*_d - \epsilon$:
    \begin{align*}
        \E\left[ |\neighbors{\greedy{k}}| \right]
        &\geq \E\left[ |\neighbors{\greedy{k}}| \middle| |t_d - t^*_d| \leq \epsilon \right] \cdot \left( 1 - \frac{1}{n} \right) \\
        &=\left( 1 - \frac{1}{n} \right) \cdot kd \\
        &\geq \left( 1 - \frac{1}{n} \right) \cdot \E\left[ |\neighbors{\opt{k}}| \right] \tag{2}
    \end{align*}
    where the first line is by Equation (1), the second is because $t_d \geq t^*_d - \epsilon$ so up to $k \leq t^*_d - \epsilon$ $\greedyalgorithm$ still gains $d$ per iteration and the third line is because $\E\left[ |\neighbors{\opt{k}}| \right] \leq kd$ by \cref{lem: trivial upper bound for opt}.
    For $k \geq t^*_d - \delta$, $\greedyalgorithm$ gains $d$ on each iteration until $t^*_d - \delta$ and after that we use the worst case analysis which gives:
    \begin{align*}
        \E\left[ |\neighbors{\greedy{k}}| \right]
        &\geq \Pr\left[ |t_d - t^*_d| \leq \delta \right] \cdot \E\left[ |\neighbors{\greedy{k}} \middle| |t_d - t^*_d| \leq \delta \right] \\
        &\geq \left( 1 - \frac{1}{n} \right) \cdot \E\left[ |\neighbors{\greedy{k}}| \middle| |t_d - t^*_d| \leq \delta \right] \\
        &\geq \left( 1 - \frac{1}{n} \right) \cdot \left( 1 - \frac{1}{e} + \frac{1}{e} \cdot \left( \frac{t^*_d - \delta}{m/d} \right)^3 \right) \cdot \E\left[ |\neighbors{\opt{k}}| \middle| |t_d - t^*_d| \leq \delta \right], \tag{3}
    \end{align*}
    where the first line is by tower rule, the second line is due to Equation (1), the third line is due to \cref{lem:augmented worst case analysis} since $k \geq t^*_d - \delta \geq t_d$ by the condition.
    For the conditional value of $\optalgorithm$ we have:
    \begin{align*}
        \E&\left[ |\neighbors{\opt{k}}| \right] = \\
        &= \E\left[ |\neighbors{\opt{k}}| \middle| |t_d - t^*_d| \leq \delta \right] \cdot \Pr\left[ |t_d - t^*_d| \leq \delta \right] + \E\left[ |\neighbors{\opt{k}}| \middle| |t_d - t^*_d| > \delta \right] \cdot \Pr\left[ |t_d - t^*_d| > \delta \right] \\
        &\leq \E\left[ |\neighbors{\opt{k}}| \middle| |t_d - t^*_d| \leq \delta \right] \cdot 1 + m \cdot \frac{1}{n} \\
        &= \E\left[ |\neighbors{\opt{k}}| \middle| |t_d - t^*_d| \leq \delta \right] + \frac{1}{c} \\
        &\Rightarrow \E\left[ |\neighbors{\opt{k}}| \middle| |t_d - t^*_d| \leq \delta \right] 
        \geq \E\left[ |\neighbors{\opt{k}}| \right] - \frac{1}{c}
        \geq \left( 1 - \frac{4}{\epsilon n} \right) \cdot \E\left[ |\neighbors{\opt{k}}| \right], \tag{3a}
    \end{align*}
    where the first line is by tower rule, the second line is due to $|\neighbors{\opt{k}}| \leq m$ and Equation (1), the third line is by cancellation and the fourth line is because $\frac{1}{c} \leq \frac{\E\left[ |\neighbors{\opt{k}}| \right]}{c \E\left[ |\neighbors{\smax{k}}| \right]} \leq \frac{\E\left[ |\neighbors{\opt{k}}| \right]}{ \epsilon n / 4}$.
    Substituting Equation (3a) to (3) we have:
    \begin{align*}
        \E\left[ |\neighbors{\greedy{k}}| \right] 
        &\geq \left( 1 - \frac{1}{n} \right) \cdot \left( 1 - \frac{1}{e} + \frac{1}{e} \cdot \left( \frac{t^*_d - \delta}{m/d} \right)^3 \right) \cdot \left( 1 - \frac{4}{\epsilon n} \right) \cdot \E\left[ |\neighbors{\opt{k}}| \middle| |t_d - t^*_d| \leq \delta \right] \\
        &\geq \left( 1 - \frac{1}{e} + \frac{1}{e} \cdot \left( \frac{t^*_d - \delta}{m/d} \right)^3 - \frac{8}{\epsilon n} \right) \cdot \E\left[ |\neighbors{\opt{k}}| \middle| |t_d - t^*_d| \leq \delta \right]. \tag{4}
    \end{align*}
    At this point it suffices to lower bound $\frac{t^*_d - \delta}{m/d}$.
    For $c \in \left( \frac{\epsilon^3}{6} - \frac{\log{m}}{m}, 1 \right)$:
    \begin{align*}
        \frac{t^*_d}{m/d} 
        = 1 - \left( 1 + c d (d-1) \right)^{- \frac{1}{d-1}} 
        \geq 1 - e^{-\frac{c}{d}}
        \geq 1 - e^{-\frac{c}{\frac{20^4}{\epsilon^8}}}
        \geq 1 - e^{-\frac{\epsilon^3/7}{\frac{20^4}{\epsilon^8}}}
        \geq 1 - e^{- \frac{\epsilon^{11}}{7 \cdot 20^4}}, \tag{4a}
    \end{align*}
    where the first inequality is due to \cref{lem:simple lower bound for phase 1}, the second inequality is due to assumption $d \leq \frac{20^4 \max(1, c^2)}{\epsilon^8} = \frac{20^4}{\epsilon^8}$ and the third inequality is due to $c \geq \frac{\epsilon^3}{6} - \frac{\log{m}}{m} \geq \frac{\epsilon^3}{7}$.
    For $c \geq 1$:
    \begin{align*}
        \frac{t^*_d}{m/d} 
        = 1 - \left( 1 + c d (d-1) \right)^{- \frac{1}{d-1}} 
        \geq 1 - \left( 1 + d (d-1) \right)^{- \frac{1}{d-1}} 
        \geq \frac{1}{d} 
        \geq \frac{\epsilon^8}{20^4 c^2}, \tag{4b}
    \end{align*}
    where the first inequality is due to $c \geq 1$, the second inequality is due to \cref{lem:simple lower bound for phase 1 n=m} and the third inequality is by assumption $d \leq \frac{20^4 c^2}{\epsilon^8}$.
    We also have:
    \begin{align*}
        \frac{\delta}{m/d} 
        &= \frac{3 e^{cd^2} \sqrt{8 n \log{2n}}}{m/d} \\
        &= 3d e^{cd^2} \sqrt{16c^2  \frac{\log{2cm}}{2cm}} \\
        &\leq 12 c d e^{cd^2} \frac{1}{(2cm)^{1/4}} \\
        &\leq 12 c \frac{20^4 c^2}{\epsilon^8} e^{c \left( \frac{20^4 c^2}{\epsilon^8} \right)^2}  \cdot \frac{1}{(2cm)^{1/4}} \\
        &\leq \frac{20^6 c^3}{\epsilon^8} e^{c \left( \frac{20^4 c^2}{\epsilon^8} \right)^2} \cdot \frac{1}{(2cm)^{1/4}} \\
        &\leq \frac{\epsilon^8}{20^6 c^2}, \tag{4c}
    \end{align*}
    where the first line is by definition, the second line is due to $n = cm$, the third line is due to $\frac{\log{2cm}}{2cm} \leq \frac{1}{\sqrt{2cm}}$, the fourth line is due to $d \leq \frac{20^4 c^2}{\epsilon^8}$, the fifth line is due to $12 < 20^2$ and the sixth line is by assumption $m \geq \frac{1}{2c} \cdot \left( c \left( \frac{20^6 c^2}{\epsilon^8} \right)^2 e^{\frac{20^4 c^3}{\epsilon^8}} \right)^4$.
    Overall, by Equations (4a), (4b), (4c), there exists a positive function $g(x,y)$ that is increasing in $x$ and decreasing in $y$ such that:
    \begin{align*}
        \frac{t^*_d - \delta}{m/d} \geq g(\epsilon, c). \tag{4d}
    \end{align*}
    By Equation (1) and substituting Equation (4d) to (4) we obtain the statement for $f(x, y) = \frac{g^3(x, y)}{2e}$.
\end{proof}

\begin{lemma}
\label{lem:augmented worst case analysis}
    Let $d \in \mathbbm{N}$ and $B = (L, R)$ be a bipartite graph such that $|\neighbors{u}| = d$ for all $u \in L$. 
    Then $|\neighbors{\greedy{k}}| \geq \left( 1 - \exp\{- (1 - \frac{t}{k}) \} \cdot \left( 1 - \frac{|\neighbors{\greedy{t}}|}{|\neighbors{\opt{k}}|} \right) \right) \cdot |\neighbors{\opt{k}}|$ for all $t \leq k$.
\end{lemma}
\begin{proof}
    We repeat the worst case analysis of $\greedyalgorithm$ for iterations $k \geq i > t$:
    \begin{align*}
        |\neighbors{\opt{k}}| - |\neighbors{\greedy{k}}| 
        \leq \left(1 - \frac{1}{k}\right)^{1} \cdot \left(|\neighbors{\opt{k}}| - |\neighbors{\greedy{k-1}}| \right)
        \leq \left(1 - \frac{1}{k}\right)^{k-t} \cdot \left(|\neighbors{\opt{k}}| - |\neighbors{\greedy{t}}| \right) \tag{1}
    \end{align*}
    and solving for $|\neighbors{\greedy{k}}|$ we have:
    \begin{align*}
        |\neighbors{\greedy{k}}| 
        &\geq \left( 1 - \left(1 - \frac{1}{k} \right)^{k - t} \cdot \left(1 - \frac{|\neighbors{\greedy{t}}|}{|\neighbors{\opt{k}}|} \right) \right) \cdot |\neighbors{\opt{k}}| \\
        &\geq \left( 1 - \exp\left\{-\left(1 - \frac{t}{k}\right)\right\} \cdot \left(1 - \frac{|\neighbors{\greedy{t}}|}{|\neighbors{\opt{k}}|} \right) \right) \cdot |\neighbors{\opt{k}}|,
    \end{align*}
    where the first line is by (1) and the second line is due to $1 - \frac{1}{k} \leq e^{-\frac{1}{k}}$.
\end{proof}

\begin{lemma}
\label{lem:critical region lower bound greedy m/d minimum}
    Let $y \geq 1$ and $f(x) = 1 - \exp\left(- \left( 1 - \frac{1}{x} \right) \right) \cdot \left( 1 - \frac{1}{\min\{x, y\}} \right)$. 
    Then we have that
    $f(x) \geq 1 - \exp\left(- \left( 1 - \frac{1}{y} \right) \right) \cdot \left( 1 - \frac{1}{y} \right)$ for all $x \geq 1$.
\end{lemma}
\begin{proof}
    We start with case $x \leq y$:
    \begin{align*}
        f(k) 
        = 1 - \exp\left(- \left( 1 - \frac{1}{x} \right) \right) \cdot \left( 1 - \frac{1}{x} \right)
        \geq 1 - \exp\left(- \left( 1 - \frac{1}{y} \right) \right) \cdot \left( 1 - \frac{1}{y} \right).
    \end{align*}
    For $x > y$:
    \begin{align*}
        f(k) = 1 - \exp\left(- \left( 1 - \frac{1}{x} \right) \right) \cdot \left( 1 - \frac{1}{y} \right)
        \geq 1 - \exp\left(- \left( 1 - \frac{1}{y} \right) \right) \cdot \left( 1 - \frac{1}{y} \right).
    \end{align*}
\end{proof}

\begin{lemma}
\label{lem:critical region lower bound greedy ratio simplification}
    For all $x \leq 1$ we have that $1 - e^{-(1-x)} \cdot (1 - x) \geq 1 - \frac{1}{e} + \frac{x^3}{e}$.
\end{lemma}
\begin{proof}
    We have that:
    \begin{align*}
        1 - e^{-(1-x)} \cdot (1 - x) 
        &= 1 - e^{-1} e^{x} \cdot (1 - x) \\
        &\geq 1 - e^{-1} (1 + x + x^2) \cdot (1 - x) \\
        &= 1 - e^{-1} (1 - x^3) \\
        &= 1 - \frac{1}{e} + \frac{x^3}{e},
    \end{align*}
    where the inequality is due to $e^x \leq 1 + x + x^2$ for all $x \leq 1$.
\end{proof}

\begin{lemma}
\label{lem:simple lower bound for phase 1}
    For all $d \geq 2$, $c \in (0,1)$ we have that $1 - \left( 1 + c d(d-1) \right)^{- \frac{1}{d-1}} \geq 1 - e^{- c / d}$.
\end{lemma}
\begin{proof}
    Equivalently we show that $\left( 1 + cd(d-1) \right)^{-\frac{1}{d-1}} \leq \left( \frac{d-1}{d} \right)^c \leq e^{- c / d}$ or after taking logs:
    \begin{align*}
        \log\left( 1 + c d(d-1) \right) \geq (d-1) \cdot \log\left( \frac{d}{d-1} \right)^c.
    \end{align*}
    The left-hand side is greater than $c$ since $\log\left( 1+ c d(d-1) \right) \geq \log\left( 1 + 2c \right) \geq \log{e^c} = c$ for all $d \geq 2$ and $c \in (0,1)$, the right-hand side is less than $c$ since $(d-1) \cdot \log\left( \frac{d}{d-1} \right)^c \leq (d-1) \cdot \frac{c}{d-1} = c$.
\end{proof}

\subsection{Putting Everything Together}
\label{sec:generalization of theorem 1}

\begin{theorem} 
\label{thm:ratio lower bound for greedy}
     Let $\epsilon \in (0,1)$, $m \in \mathbbm{N}$, $c \geq 0$, $n = cm$, $d \in [m], k \in [n]$ and draw bipartite graph $\cB = (L, R, \cE)$ from random model $\textsc{ULRR}(n, m, d)$.
     Then, there exists a positive function $f(x,y)$ that is increasing in $x$ and decreasing in $y$ such that $\E\left[ |\neighbors{\greedy{k}}| \right] \geq \left( 1 - \frac{1}{e} + f(\epsilon, c) \right) \cdot \E\left[ |\neighbors{\opt{k}}| \right]$.
\end{theorem}
\begin{proof}
    Let $\epsilon = 0.3$ and $m_0(x, y)$ be the function from the condition on $m$ in \cref{lem:ratio for small d large n/m}.
    For $m \geq m_0(\epsilon, c)$,  the statement follows by \cref{lem:ratio for k away from m/d,lem:ratio for large d,lem:ratio for small n/m,lem:ratio for small d large n/m}.
    For $m < m_{0}(\epsilon, c)$, the worst-case approximation ratio of greedy is $1 - \left( 1 - \frac{1}{k} \right)^k$, so we get
    \begin{align*}
         1 - \left( 1 - \frac{1}{k} \right)^k  \geq  1 - \frac{1}{e} + \frac{1}{4ek}  \geq  1 - \frac{1}{e} + \frac{1}{4e c m_{0}(\epsilon, c)},
    \end{align*}
    where  the first inequality is by \cref{claim}  and the second since $k \leq n = cm \leq c m_{0}(\epsilon, c)$. 
\end{proof}

\section{Generalization to Non Regular Bipartite Graphs}
\label{sec:nonregular}

In this section, we generalize \cref{thm:Theorem 2} to unbalanced, non left-regular bipartite graphs and we use results from \cref{sec:generalization of technical tools}.

\subsection{Linear And Saturated Regions}
\begin{lemma}
\label{lem:asymptotic optimality outside critical region for unequal degrees}
    Let $n, m \in \mathbbm{N}$ and $m \geq d_1 \geq d_2 \geq \dots \geq d_n \geq 1$, arbitrary $\epsilon \in (0, 1)$ and $\cB \sim \textsc{GenR}(n, m, d_1, \dots, d_n)$.
    Then we have the following cases: 
    \begin{enumerate}
        \item $\E\left[ |\neighbors{\greedy{k}}| \right] \geq \left( 1 - \epsilon \right) \cdot \sum_{i=1}^{k} d_i$ for all $k$ such that $\sum_{i=1}^{k}d_i \leq \frac{\epsilon m}{2}$,
        \item $\E\left[ |\neighbors{\greedy{k}}| \right] \geq \left( 1 - \epsilon \right) \cdot m$ for all $k$ such that $\sum_{i=1}^{k}d_i \geq \frac{2m}{\epsilon}$.
    \end{enumerate}
\end{lemma}
\begin{proof} 
    For $\sum_{i=1}^{k}d_i \leq \frac{\epsilon m}{2}$ we have that:
    \begin{align*}
        \E\left[ |\neighbors{\greedy{k}}| \right] 
        &\geq \E\left[ |\neighbors{\smax{k}}| \right] \\
        &= \left( 1 - \prod_{i=1}^{k}\left( 1 - \frac{d_i}{m} \right) \right) \cdot m \\
        &\geq \left( 1 - e^{- \frac{\sum_{i=1}^{k}d_i}{m}} \right) \cdot m \\
        &= \frac{1 - e^{- \frac{\sum_{i=1}^{k}d_i}{m}}}{\frac{\sum_{i=1}^{k}d_i}{m}} \cdot \sum_{i=1}^{k}d_i \\
        &\geq \frac{1 - e^{- \epsilon}}{\epsilon} \cdot \sum_{i=1}^{k}d_i \\
        &\geq \frac{1 - (1 - \epsilon + \epsilon^2)}{\epsilon} \cdot \sum_{i=1}^{k}d_i \\
        &\geq \left( 1 - \epsilon \right) \cdot \sum_{i=1}^{k}d_i, \tag{1}
    \end{align*}
    where the first line is by \cref{lem: greedy lower bound by smax}, the second line is by direct calculation, the third line is due to $1 - \frac{d_i}{m} \leq e^{-\frac{d_i}{m}}$ for all $i$, the fourth line is by multiplying and dividing by the sum, the fifth line is because $f(x) = \frac{1 - e^{-x}}{x}$ is decreasing for all $x$ and by assumption $x = \frac{\sum_{i=1}^{k}d_i}{m} \leq \frac{\epsilon}{2} \leq \epsilon$, the sixth line is due to $e^{-\epsilon} \leq 1 - \epsilon + \epsilon^2$ for all $\epsilon < 1$.
    For $\sum_{i=1}^{k}d_i \geq \frac{2m}{\epsilon}$ we have that:
    \begin{align*}
        \E\left[ |\neighbors{\greedy{k}}| \right] 
        &\geq \E\left[ |\neighbors{\smax{k}}| \right] \\
        &= \left( 1 - \prod_{i=1}^{k}\left( 1 - \frac{d_i}{m} \right) \right) \cdot m \\
        &\geq \left( 1 - e^{- \frac{\sum_{i=1}^{k}d_i}{m}} \right) \cdot m \\
        &\geq \left( 1 - e^{- \frac{1}{\epsilon}} \right) \cdot m \\
        &\geq \left( 1 - \epsilon \right) \cdot m,  \tag{2}
    \end{align*}
    where the first line is by \cref{lem: greedy lower bound by smax}, the second line is by direct calculation, the third line is due to $1 - \frac{d_i}{m} \leq e^{-\frac{d_i}{m}}$, the fourth line is due to $1-e^{-x}$ being increasing for all $x$ and $x = \frac{\sum_{i=1}^{k}d_i}{m} \geq \frac{2}{\epsilon} \geq \frac{1}{\epsilon}$ by assumption, the fifth line is due to $e^{-x} \leq \frac{1}{x}$ for all $x$.
    By (1) and (2) we obtain the statement.
\end{proof}

\subsection{Critical Region: Unequal Large Degrees}
\begin{lemma}
\label{lem:asymptotic optimality for large unequal degrees}
    Let $n, m \in \mathbbm{N}$ and $m \geq d_1 \geq d_2 \geq \dots \geq d_n \geq 1$, arbitrary $\epsilon \in (0, 1)$ and $\cB \sim \textsc{GenR}(n, m, d_1, \dots, d_n)$.
    If $\max\left( \frac{\epsilon m}{2}, \frac{20^4 \max(1, (n/m)^2)}{\epsilon^8} k \right) \leq \sum_{i=1}^{k}d_i \leq \frac{2m}{\epsilon}$, then we have $\E\left[ |\neighbors{\greedy{k}}| \right] \geq \left( 1 - \epsilon \right) \cdot \E\left[ |\neighbors{\opt{k}}| \right]$.
\end{lemma}
\begin{proof}
    Let $d = \frac{1}{k} \sum_{i=1}^{k}{d_i}$ be the average degree among the $k$ highest and $c = n/m$.
    To upper bound the optimal solution, first notice that $|\neighbors{\opt{k}}| = \max_{S \in \binom{L}{k}}|\neighbors{S}|$.
    Also let $$\delta = \sqrt{\frac{24}{\epsilon} \left( \max\left(0, \frac{ \log{e c} + \log{\frac{\epsilon d}{2}}}{\epsilon d} \right)  + \frac{\log{m}}{m} \right)}.$$ 
    We proceed to show that $\delta \in (0,1)$:
    \begin{align*}
        \delta 
        &\leq \sqrt{\frac{24}{\epsilon} \left( \max\left(0, \frac{\log{e c}}{\sqrt{d}} \cdot \frac{1}{\epsilon \sqrt{d}} + \frac{\log{\epsilon d}}{\epsilon d} \right) + \frac{\log{m}}{m} \right)} \\
        &\leq \sqrt{\frac{24}{\epsilon} \left( \max\left(0, \frac{1}{\epsilon \sqrt{d}} + \frac{1}{\sqrt{\epsilon d}} \right) + \frac{1}{\sqrt{m}} \right)} \\
        &\leq \frac{\sqrt{72}}{\epsilon d^{1/4}} \\
        &\leq \frac{\epsilon}{2}, \tag{1}
    \end{align*}
    where the second line is due to $d \geq \frac{20^4 \max(1, c^2)}{\epsilon^8} \geq (\max(1, \log{e c}))^2$ and $\frac{\log{x}}{x} \leq \frac{1}{\sqrt{x}}$ for all $x$, the third line is due to $\sqrt{\epsilon} \geq \epsilon, m \geq d \geq \epsilon^2 d$ and the fourth line is by assumption $d \geq \frac{20^4}{\epsilon^8}$.
    Then we have:
    \begin{align*}
        \Pr\left[ |\neighbors{\opt{k}}| \geq (1 + \delta) \cdot \E\left[ |\neighbors{\smax{k}}| \right] \right]
        &= \Pr\left[ \max_{S \in \binom{L}{k}}|\neighbors{S}| \geq (1 + \delta) \cdot \E\left[ |\neighbors{\smax{k}}| \right] \right] \\
        &\leq \sum_{S \in \binom{L}{k}} \Pr\left[|\neighbors{S}| \geq (1 + \delta) \cdot \E\left[ |\neighbors{\smax{k}}| \right] \right] \\
        &\leq \sum_{S \in \binom{L}{k}} \Pr\left[|\neighbors{\smax{k}}| \geq (1 + \delta) \cdot \E\left[ |\neighbors{\smax{k}}| \right] \right] \\
        &\leq \sum_{S \in \binom{L}{k}} \exp\left( -\frac{\delta^2}{3} \cdot \E\left[ |\neighbors{\smax{k}}| \right] \right) \\
        &= \exp\left( \log{\binom{n}{k}} - \frac{\delta^2}{3} \E\left[ |\neighbors{\smax{k}}| \right]  \right), \tag{2}
    \end{align*}
    where the first line is by definition of $\opt{k}$, the second line is by union bound, the third line is because $|\neighbors{S}|$ is dominated by $|\neighbors{\smax{k}}$ for all $S \in \binom{L}{k}$ by \cref{lem: coverage of s dominated by smax}, the fourth line is by \cref{lem: coverage of smax admits chernoff} since $\delta \in (0, 1)$ by Equation (1) and the fifth line is because $|\binom{L}{k}| = \binom{n}{k} = e^{\log{\binom{n}{k}}}$.
    We upper bound the exponent as follows:
    \begin{align*}
        \log{\binom{n}{k}} - \frac{\delta^2}{3} \E\left[ |\neighbors{\smax{k}}| \right]
        &\leq k \log\left( \frac{en}{k} \right) - \frac{\delta^2}{3} \cdot \left( 1 - e^{-kd/m} \right) \cdot m \\
        &\leq \frac{\log{\frac{e c \epsilon d}{2}}}{\frac{\epsilon d}{2}} \cdot m - \frac{\delta^2}{3} \cdot \left( 1 - e^{-\epsilon/2} \right) \cdot m \\
        &\leq \frac{\log{e c} + \log{\frac{\epsilon d}{2}}}{\frac{\epsilon d}{2}} \cdot m - \frac{\delta^2}{3} \cdot \frac{\epsilon}{4} \cdot m \\
        &\leq \frac{\log{e c} + \log{\frac{\epsilon d}{2}}}{\frac{\epsilon d}{2}} \cdot m - \left( \frac{\log{e c} + \log{\frac{\epsilon d}{2}}}{\frac{\epsilon d}{2}} + \frac{2\log{m}}{m} \right) \cdot m \\
        &\leq - \log{m}, \tag{2a}
    \end{align*}
    where the first line is due to $\binom{n}{k} \leq \left( \frac{en}{k} \right)^k$ and \cref{lem:smax expected coverage}, the second line is because $k \log\left( \frac{en}{k} \right)$ is increasing for all $k \in [n]$ and $\frac{\epsilon m}{2 d} \leq k \leq \frac{2 m}{\epsilon d}$ by assumption, the third line is because $1 - e^{-\epsilon/2} \geq \frac{\epsilon}{4}$ for all $\epsilon \in (0,1)$, the fourth line is by definition of $\delta$ and the fifth line is due to $-2 \log{m} \leq - \log{m}$.
    Let event $F = \{ |\neighbors{\opt{k}}| \geq (1 + \delta) \cdot \E\left[ |\neighbors{\smax{k}}| \right] \}$.
    We are ready to upper bound the expected value of $\optalgorithm$:
    \begin{align*}
        \E\left[ |\neighbors{\opt{k}}| \right] 
        &= \E\left[ |\neighbors{\opt{k}}| \middle| F \right] \cdot \Pr\left[ F \right] + \E\left[ |\neighbors{\opt{k}}| \middle| \bar{F} \right] \cdot \Pr\left[ \bar{F} \right] \\
        &\leq m \cdot \frac{1}{m} + \E\left[ |\neighbors{\opt{k}}| \middle| \bar{F} \right] \cdot 1 \\
        &\leq 1 + (1 + \delta) \cdot \E\left[ |\neighbors{\smax{k}}| \right] \\
        &\leq \left( 1 + \delta + \frac{4}{\epsilon m} \right) \cdot \E\left[ |\neighbors{\smax{k}}| \right] \\
        &\leq \left( 1 + \epsilon \right) \cdot \E\left[ |\neighbors{\smax{k}} | \right] \\
        &\leq \left( 1 + \epsilon \right) \cdot \E\left[ |\neighbors{\greedy{k}} | \right], \tag{3}
    \end{align*}
    where the first line is by tower rule, the second line is because $|\neighbors{\opt{k}}| \leq m$ and $\Pr\left[ F \right] \leq \exp\left( - \log{m} \right) = \frac{1}{m}$ by substituting Equation (2a) to (2), the third line follows by definition of $F$, the fourth line is because $1 =  \frac{\E\left[ |\neighbors{\smax{k}}| \right]}{\E\left[ |\neighbors{\smax{k}}| \right]} \leq \frac{\E\left[ |\neighbors{\smax{k}}| \right]}{\left( 1 - e^{-\epsilon/2} \right) \cdot m} \leq \frac{\E\left[ |\neighbors{\smax{k}}| \right]}{\frac{\epsilon}{4} \cdot m}$, the fifth line is because $\delta \leq \frac{\epsilon}{2}$ by Equation (1) and $m \geq d \geq \frac{20^4 c^2}{\epsilon^8} \geq \frac{8 c}{\epsilon^2}$ and the sixth line is by \cref{lem: greedy lower bound by smax}.
    The statement follows directly from Equation (3):
    \begin{align*}
        \E\left[ |\neighbors{\greedy{k}}| \right] 
        \geq \frac{1}{1 + \epsilon} \cdot \E\left[ |\neighbors{\opt{k}}| \right]
        \geq \left( 1 - \epsilon \right) \cdot \E\left[ |\neighbors{\opt{k}}| \right].
    \end{align*}
    where the second inequality is because $\frac{1}{1 + \epsilon} \geq 1 - \epsilon$ for all $\epsilon \in (0, 1)$.
\end{proof}

\subsection{Main Results}
\label{sec:generalization of theorem 2}

\begin{theorem}[Generalization of \cref{thm:Theorem 2}]
\label{lem: generalization of theorem 2 in non regular}
    Let $n, m \in \mathbbm{N}$ and $m \geq d_1 \geq d_2 \geq \dots \geq d_n \geq 1$, arbitrary $\epsilon \in (0, 1)$ and $\cB \sim \textsc{GenR}(n, m, d_1, \dots, d_n)$.
    Suppose that one of the following conditions holds:
    \begin{enumerate}
        \item $\sum_{i=1}^{k}d_i \not \in [(\epsilon/2) m, (\epsilon/2)^{-1}m]$ or
        \item $\frac{1}{k} \sum_{i=1}^{k}d_i \geq \frac{20^4 \max(1, (n/m)^2)}{\epsilon^8}$.
    \end{enumerate}
    Then greedy achieves, in expectation, a $1- \epsilon$ approximation.
\end{theorem}
\begin{proof}
    Follows by \cref{lem:asymptotic optimality outside critical region for unequal degrees,lem:asymptotic optimality for large unequal degrees}.
\end{proof}

\begin{theorem}[Power Law Left Degrees]
\label{thm:power law asymptotic optimality}
    Let $n,m \in \mathbbm{N}$, $X_1, X_2, \dots, X_n$ be i.i.d. samples from Pareto distribution with parameters $a > 1$ and $x_{min} = 1$.
    Also let $d_1 \geq \dots \geq d_n$ be their order statistics, $\epsilon \in (0,1)$ arbitrary and draw graph $\cB$ from random model $\textsc{GenR}(n, m, \floor{d_1}, \dots, \floor{d_n} )$.
    Then $\E\left[ |\neighbors{\greedy{k}}| \right] \geq \left( 1 - \epsilon - O\left( \frac{1}{\log^a{m}} \right) \right) \cdot \E[|\neighbors{\opt{k}}|]$
    for all $k \leq \frac{\epsilon m n^{-1/a}}{2 \log{m}}$.
\end{theorem}
\begin{proof}
    Let $x \geq 1$, then for the highest degree we have that:
    \begin{align*}
        \Pr\left[ \floor{d_1} \leq x \right]
        &\overset{\floor{d_1} \leq d_1}{\geq} \Pr\left[ d_1 \leq x \right] \\
        &= \Pr\left[ \max_{i \in [n]}X_i \leq x \right] \\
        &= \left( \Pr\left[ X_1 \leq x \right] \right)^n \\
        &\overset{\text{Pareto}}{=} \left( 1 - x^{-a} \right)^n \\
        &\overset{\text{Bernoulli}}{\geq} 1 - n x^{-a}. \tag{1}
    \end{align*}
    For the sum of the $k$ highest degrees we have that:
    \begin{align*}
        \Pr\left[ \sum_{i=1}^{k} \floor{d_i} \leq \frac{\epsilon m}{2} \right]
        &\geq \Pr\left[ k \floor{d_1} \leq \frac{\epsilon m}{2} \right] \\
        &= \Pr\left[ \floor{d_1} \leq \frac{\epsilon m}{2k} \right] \\
        &\geq \Pr\left[ \floor{d_1} \leq n^{1/a} \log{m} \right] \\
        &\geq 1 - n \left( n^{1/a} \log{m} \right)^{-a} \\
        &= 1 -  \frac{1}{\log^a{m}}, \tag{2}
    \end{align*}
    where the first line is due to $\sum_{i=1}^{k} \floor{d_i} \leq k \floor{d_1}$, the second line is by dividing by $k$, the third line is by assumption $k \leq \epsilon m n^{-1/a} / 2\log{m}$, the fourth line is by substituting $x = n^{1/a} \log{m} \geq 1$ for $n$ large enough into (1) and the fifth line is by rearranging.
    Then for the value of $\greedyalgorithm$ we have that:
    \begin{align*}
        \E\left[ |\neighbors{\greedy{k}}| \right]
        &\geq \E\left[ |\neighbors{\greedy{k}}| \middle| \sum_{i=1}^{k} \floor{d_i} \leq \frac{\epsilon m}{2} \right] \cdot \Pr\left[ \sum_{i=1}^{k} \floor{d_i} \leq \frac{\epsilon m}{2} \right] \\
        &\geq (1 - \epsilon) \cdot \E\left[\sum_{i=1}^{k} \floor{d_i} \right] \cdot \left( 1 -  \frac{1}{\log^a{m}} \right) \\
        &\geq \left(1 - \epsilon -  \frac{1}{\log^a{m}} \right) \cdot \E\left[ \sum_{i=1}^{k}  \floor{d_i} \right] \\
        & \geq \left(1 - \epsilon - O\left( \frac{1}{\log^a{m}} \right) \right) \cdot \E\left[ |\neighbors{\opt{k}}| \right],
    \end{align*}
    where the first line is by tower rule, the second line is by case 1 of \cref{lem:asymptotic optimality outside critical region for unequal degrees} and (2), the third line is due to $(1 - x)(1 - y) = 1 - x - y + xy \geq 1 - x - y$ for all $x, y \geq 0$ and the fourth line is due to $\E\left[ \sum_{i=1}^{k} \floor{d_i} \right] \geq \E\left[ |\neighbors{\opt{k}}| \right]$ by \cref{lem: trivial upper bound for opt}.
\end{proof}

Our techniques unfortunately do not suffice for extending Theorem~\ref{thm:Theorem 1} to the setting from this section with non left-regular graphs. The main reason is that there are regimes where we only count the number of times we accept nodes
whose contribution is the highest possible. But in non left-regular graphs, we could have
just one node with a high degree and all other nodes with low degrees, in which case the
contribution of the one node is insignificant overall.

\end{document}